%% file: arXiv.tex
\documentclass[a4paper,12pt]{article}
\usepackage[noblocks]{authblk}
\usepackage{Gawthrop}
\usepackage{times}
\usepackage{a4wide}
\usepackage{fancyhdr,titling}
\usepackage[pagebackref]{hyperref}
\hypersetup{colorlinks=true,citecolor=blue,linkcolor=blue,urlcolor=blue}
\usepackage{doi}

\title{Intermittent Control in Man and Machine}
\author{Peter Gawthrop}
\affil{Systems Biology Laboratory,
   Melbourne School of Engineering,\authorcr
   University of Melbourne,
   Victoria 3010, Australia.}
\author{Henrik Gollee}
\affil{School of Engineering, University of Glasgow, Glasgow, UK}
\author{Ian Loram}
\affil{School of Healthcare Science, Manchester Metropolitan
  University, \authorcr Manchester, UK}

\begin{document}
\pagenumbering{roman}
\maketitle 
\begin{abstract}
  Intermittent control has a long history in the physiological
  literature and there is strong experimental evidence that some human
  control systems are intermittent. Intermittent control has also
  appeared in various forms in the engineering literature. This
  article discusses a particular mathematical model of Event-driven
  Intermittent Control which brings together engineering and
  physiological insights and builds on and extends previous work in
  this area. Illustrative examples of the properties of Intermittent
  Control in a physiological context are given together with suggestions for
  future research directions in both physiology and engineering.
\end{abstract}

\newpage
\tableofcontents\newpage
\listoffigures
\newpage
\pagenumbering{arabic}
\include{Gawthrop}

\end{document}

%% file: Gawthrop.tex
\section{Introduction}
\label{sec:intro}
Conventional sampled-data control uses a \emph{zero-order} hold, which
produces a piecewise-constant control signal \citep{FraPowEma94}, and
can be used to give a sampled-data implementation which approximates a
previously-designed continuous-time controller.  
In contrast to conventional sampled data control, intermittent control
\citep{GawWan07}, explicitly embeds the underlying continuous-time
closed-loop system in a \emph{generalised hold}.  A number of version
of the generalised hold are available; this chapter focuses on the
\emph{system-matched} hold \citep{GawWan11} which expicitly generates
an open-loop intersample control trajectory based on the underlying
continuous-time closed-loop control system. 
Other versions of the generalised hold include Laguerre function based
holds \citep{GawWan07} and a ``tapping'' hold \citep{GawGol12}.

There are at three areas where intermittent control has been used:
\begin{enumerate}
\item continuous-time model-based predictive control where
  the intermittency is associated with on-line optimisation
  \citep{RonArsGaw99,GawWan09,GawWan10},
\item event-driven control systems where the intersample
  interval is time varying and determined by the event times
  \citep{GawWan09a,GawWan11},
\item and physiological control systems which, in some cases, have an
  event-driven intermittent character \citep{LorLak02,GawLorLakGol11}.
  This intermittency may be due to the ``computation'' in the central
  nervous system. Although this  Chapter is orientated towards physiological
  control systems, but we believe that it is more widely applicable.
\end{enumerate}
%
% In addition intermittent control has the potential for control over
% low-bandwidth channels where there is a finite rate at which feedback
% information can be received and transmitted
% \citep{NaiEva03,NaiEvaMarMor04,HesNagYon07}.

Intermittent control has a long history in the
physiological literature including
\citep{Cra47a,Cra47b,Vin48,NavSta68,NeiNei88,MiaWeiSte93a,BhuSha99,LorLak02,LorGolLakGaw10,GawLorLakGol11}.
There is strong experimental evidence that some human control systems
are intermittent
\citep{Cra47a,Vin48,NavSta68,BotCasMorSan05,LorKamGolGaw12,KamGawGolLor13}
and it has been suggested that this intermittency arises in the
central nervous system (CNS) \citep{KamGawGolLakLor13}.  For this
reason, computational models of intermittent control are important
and, as discussed below, a number of versions with various
characteristics have appeared in the literature.
Intermittent control has also appeared in various forms in the
engineering literature including
\citep{RonArsGaw99,ZhiMid03,MonAnt03,Ins06,Ast08,GawWan07,GawWan09a,GawNeiWag12}.

%%% Timed v. event driven
Intermittent control action may be initiated at regular intervals determined by a
clock, or at irregular intervals determined by events; an event is
typically triggered by an error signal crossing a threshold.
Clock-driven control is discussed by \citet{NeiNei88} and
\citet{GawWan07} and analysed in the frequency domain by \citet{Gaw09}.
Event-driven control is used by
\citet{BotCasMorSan05,BotYasNimCasMor08}, \citet{Ast08},
\citet{AsaTasNomCasMor09}, \citet{GawWan09a} and
\citet{KowGleBro12}. 
\citet[Section 4]{GawLorLakGol11} discuss event-driven control but with a
lower limit $\Delta_{min}$ on the time interval between events; this
gives a range of behaviours including continuous, timed and
event-driven control. Thus, for example, threshold based event-driven
control becomes effectively clock driven with interval $\Delta_{min}$
if the threshold is small compared to errors caused by relatively
large disturbances.
There is evidence that human control systems are, in fact, event
driven \citep{NavSta68,LorKamGolGaw12,KamGawGolLakLor13,LorKamLakGolGaw14}.
For this reason, this Chapter focuses on event-driven control.

%% State feedback
As mentioned previously, intermittent control is based on an
\emph{underlying continuous-time design method}; in particular the
classical state-space approach is the basis of the intermittent
control of \citet{GawLorLakGol11}. There are two relevant versions of
this approach: state feedback and output feedback.
State-feedback control requires that the current system state (for
example angular position and velocity of an inverted pendulum) is
available for feedback. In contrast, output feedback requires a
measurement of the system output (for example angular position of an
inverted pendulum). The classical approach to output feedback in a
state-space context \citep{KwaSiv72,GooGraSal01} is to use an observer
(or the optimal version, a Kalman filter) to deduce the state from
the system output.

%% Prediction
Human control systems are associated with time-delays. In engineering
terms, it is well-known that a predictor can be used to overcome time
delay \citep{Smi59,Kle69,Gaw82b}.
As discussed by many authors
\citep{KleBarLev70,BarKleLev70,McR80,MiaWeiWolSte93b,WolMiaKaw98,BhuSha99,KooJacKooHel01,GawLakLor08,GawLorLak09,GawLorLakGol11,LorKamGolGaw12},
it is plausible that physiological control systems have built in
model-based prediction. Following \citet{GawLorLakGol11} this chapter bases intermittent
controller on an underlying predictive design.
% ; whereas other authors,
% such as \citet{BotYasNimCasMor08}, \citet{AsaTasNomCasMor09} and
% \citet{KowGleBro12} do not.

%% SMH
The use of networked control systems leads to the ``sampling period
jitter problem'' \citep{Sal07} where uncertainties in transmission
time lead to unpredictable non-uniform sampling and stability issues
\citep{CloWouHeeNij09}. A number of authors have suggested that
performance may be improved by replacing the standard zero-order hold
by a generalised hold \citep{Sal05,Sal07} or using a dynamical model
of the system between samples \citep{ZhiMid03,MonAnt03}.
Similarly, event-driven control \citep{HeeSanBos08,Ast08}, where
sampling is determined by events rather than a clock, also leads to
unpredictable non-uniform sampling. Hence strategies for event-driven
control would be expected to be similar to strategies for networked
control. One particular form of event-driven control where events
correspond to the system state moving beyond a fixed boundary has been
called Lebesgue sampling in contrast to the so-called Riemann sampling
of fixed-interval sampling \citep{AstBer02,AstBer03}.
In particular, \citet{Ast08} uses a  ``control signal generator'':
essentially a dynamical model of the system between samples as
advocated by \citet{ZhiMid03} for the networked control case.

As discussed previously, intermittent control has an interpretation
which contains a generalised hold \citep{GawWan07}.  One particular
form of hold is based on the closed-loop system dynamics of an
underlying continuous control design: this will be called the
\emph{System-Matched Hold} (SMH) in this Chapter. Insofar as this
special case of intermittent control uses a dynamical model of the
controlled system to generate the (open-loop) control between sample
intervals, it is related to the strategies of both \citet{ZhiMid03}
and \citet{Ast08}. However, as shown in this Chapter, intermittent
control provides a framework within which to analyse and design a
range of control systems with unpredictable non-uniform sampling
possibly arising from an event-driven design.
In particular, it is shown by \citet{GawWan11} that the SMH-based
intermittent controller is associated with a separation principle
similar to that of the underlying continuous-time controller, which
states that the closed-loop poles of the intermittent control system
consist of the control system poles and the observer system poles, and
the interpolation using the system matched hold does not lead to the
changes of closed-loop poles.
As discussed by \citet{GawWan11}, this separation principle is only
valid when using the SMH. For example, intermittent control based on
the standard zero-order hold (ZOH) does not lead to such a separation
principle and therefore closed-loop stability is compromised when the
sample interval is not fixed.

%%Multivariable
%%Synergies
Human movement is characterised by low-dimensional goals achieved
using high-dimensional muscle input \citep{ShaWis05}; in control system
terms the system has redundant actuators. As pointed out by
\citet{Lat12}, the abundance of actuators is an advantage rather than
a problem. One approach to redundancy is by using the concept of
\emph{synergies} \citep{NeiNei05}: groups of muscles which act in
concert to give a desired action. It has been shown that such
synergies arise naturally in the context of optimal control
\citep{Tod04,TodJor02} and experimental work has verified the
existence of synergies \emph{in vivo} \citep{Tin07,SafTin12}.
Synergies may be arranged in hierarchies. For example, in the context
of posture, there is a natural three-level hierarchy with increasing
dimension comprising task space, joint space and muscle space. Thus,
for example, a balanced posture could be a task requirement
achievable by a range of possible joint torques each of which in turn
corresponds to a range of possible muscle activation. This chapter
focuses on the  task space -- joint space hierarchy previously
examined in the context of robotics \citep{Kha87}.

%%Sensor fusion
In a similar way, humans have an abundance of measurements available;
in control system terms the system has redundant sensors. As discussed
by \citet{KooJacKooGro99} and \citet{KooJacKooHel01}, such sensors are
utilised with appropriate sensor integration. In control system
terms, sensor redundancy can be incorporated into state-space control
using observers or Kalman-Bucy filters \citep{KwaSiv72,GooGraSal01};
this is the dual of the optimal control problem.
Again sensors can be arranged in a hierarchical fashion.
Hence optimal control and filtering provides the basis for a
continuous-time control system that simultaneously applies sensor
fusion to utilise sensor redundancy and optimal control to utilise
actuator redundancy. 

For these reasons, this Chapter extends the single-input single-output
intermittent controller of \citet{GawLorLakGol11} to the multivariable
case. As the formulation of \citet{GawLorLakGol11} is set in the
state-space, this extension is quite straightforward. Crucially, the
generalised hold, and in particular the system matched hold (SMH),
remains as the heart of multivariable intermittent control.

The particular mathematical model of intermittent control proposed by
\citet{GawLorLakGol11} combines event-driven control action based on
estimates of the controlled system state (position, velocity etc.)
obtained using a standard continuous-time state observer with
continuous measurement of the system outputs. This model of
intermittent control can be summarised as ``continuous attention with
intermittent action''. However, the state estimate is only used at the
event-driven sample time; hence, it would seem that it is not necessary
for the state observer to monitor the controlled system all of the
time. Moreover, the experimental results of \citet{Osb13} suggest that
humans can perform well even when vision is intermittently occluded.
This Chapter proposes an intermittent control model where a
continuous-time observer monitors the controlled system
intermittently: the periods of monitoring the system measurements are
interleaved with periods where the measurement is \emph{occluded}.  
%
% In particular, the intermittent observer is modelled by setting the
% observer gain vector $L$ to zero thus occluding the system
% measurements and at the same time evolving the system state according
% to the internal model. These occlusion periods are synchronised with
% the control events.
%
This model of intermittent control can be summarised as ``intermittent
attention with intermittent action''.

This chapter has two main parts:
\begin{itemize}
\item Sections
\ref{sec:continuous}--~\ref{sec:ex:comp} which give basic ideas about
intermittent control and
\item Sections \ref{sec:constrained}--~\ref{sec:ex:reach} which
  explore more advanced topics and applications.
\end{itemize}

\section{Continuous control}
\label{sec:continuous}
\begin{figure}
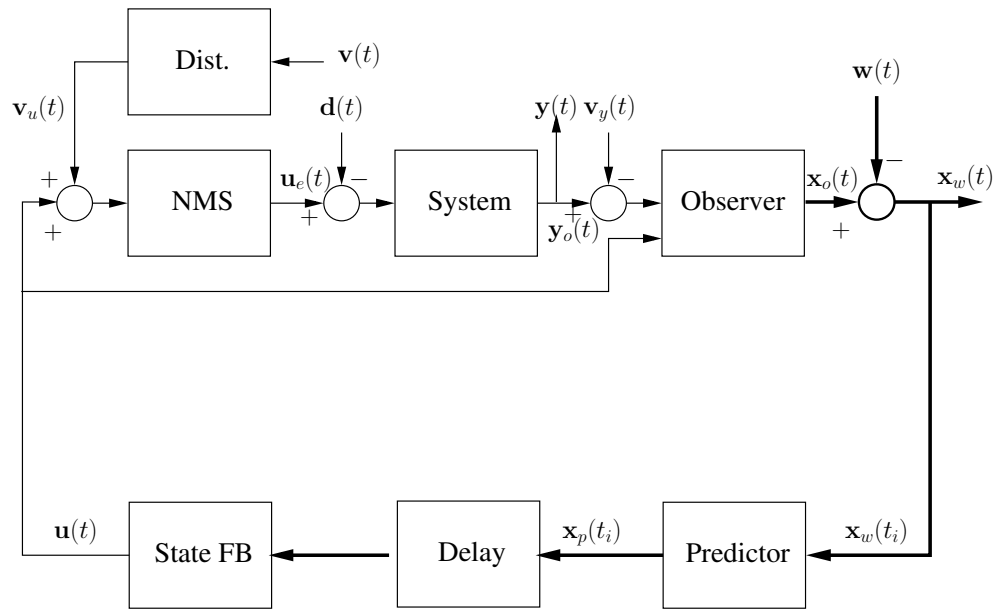

  \centering
  \Fig{MSS_arch}{\Figsize}
  \caption[The Observer, Predictor, State-feedback (OPF) model.]  {The
    Observer, Predictor, State-feedback (OPF) model. The block
    labelled ``NMS'' is a linear model of the neuro-muscular dynamics
    with input $u(t)$; in the engineering context, this would
    represent actuator dynamics. ``System'' is the linear external
    controlled system driven by the externally observed control signal
    $u_e$ and disturbance $d$, and with output $y$ and associated
    measurement noise $v_y$.  The input disturbance $v_u$ is modelled
    as the output of the block labelled ``Dist.''  and driven by the
    external signal $v$.  The block labelled ``Delay'' is a pure
    time-delay of $\del$ which accounts for the various delays in the
    human controller. The block labelled ``Observer'' gives an
    estimate $\xo$ of the state $\xx$ of the composite ``NMS'' and
    ``System'' (and, optionally, the ``Dist.'') blocks. The predictor
    provides an estimate of the future state error $\xop(t)$ the
    delayed version of which is multiplied by the feedback gain vector
    $\kk$ (block ``State FB'') to give the feedback control signal
    $u$. This figure is based on \citet[Fig. 1]{GawLorLakGol11} which
    is in turn based on \citet[Fig. 2]{Kle69} }
   \label{fig:MSS_arch}
\end{figure}

Intermittent control is based on an \emph{underlying design method}
which, in this Chapter, is taken to be conventional state-space based
observer/state-feedback control \citep{KwaSiv72,GooGraSal01} with the
addition of a state predictor \citep{Ful68,Kle69,SagMel71,Gaw76a}.
Other control design approaches have been used in this context
including pole-placement \citep{GawRon02} and cascade control
\citep{GawLeeHalODw13}. It is also noted that many control designs can
be embedded in LQ design \citep{Mac07,FooWey11} and thence used as a
basis for intermittent control \citep{GawWan10}.

\citet{GawLorLakGol11} consider a single-input single-output
formulation of intermittent control; this Chapter considers a
multi-input multi-output formulation.  As in the single-input
single-output case, this Chapter considers linear time invariant
systems with an $n\times 1$ vector state $\xx$. As discussed by
\citet{GawLorLakGol11}, the system, neuro-muscular (NMS) and
disturbances can be combined into a state-space model. For simplicity,
the measurement noise signal $v_y$ will be omitted in this Chapter
except where needed.
In contrast, however, this Chapter is based on a multiple input,
multiple output formulation.  Thus the corresponding state-space system
has multiple outputs represented by the $n_y\times 1$ vector $\yy$ and
$n_o\times 1$ vector $\yy_o$, multiple control inputs represented by
the $n_u\times 1$ vector $\uu$ and multiple unknown disturbance inputs
represented by the $n_u\times 1$ vector $\ddd$ where:
\begin{equation}
  \label{eq:sys}
  \begin{cases}
      \ddt{{\xx}}(t) &= \A {\xx}(t) + \B \uu(t) + \B_d \ddd(t)\\
      \yy(t) &= \C {\xx}(t)\\
      \yy_o(t) &= \C_o {\xx}(t)
  \end{cases}
\end{equation}
% \tbd{sec:continuous}{Make sure equation (\ref{eq:sys}) is consistent
%   with equations throughout the chapter}
$\A$ is an $n \times n$ matrix, $\B$ and $\B_d$ are a $n \times
n_u$ matrices, $\C$ is a $n_y \times n$ matrix and $\C_o$ is a $n_o
\times n$ matrix. The $n \times 1$ column vector $\xx$ is the system
state.
In the multivariable context, there is a distinction between the $n_y
\times n$ task vector $\yy$ and the $n_o \times n$ observed vector
$\yy_o$: the former corresponds to control objectives whereas the
latter corresponds to system sensors and so provides information to
the observer.
Equation (\ref{eq:sys}) is identical to \citet[Equation
(5)]{GawLorLakGol11} except that the scalar output $y$ is replaced by
the vector outputs $\yy$ and $\yy_o$, the scalar input $u$ is replaced by
the vector input $\uu$ and the scalar input disturbance $d$ is replaced by
the vector input disturbance $\ddd$.
Following standard practice \citep{KwaSiv72,GooGraSal01}, it is
assumed that $\A$ and $\B$ are such that the system (\ref{eq:sys}) is
\emph{controllable} with respect to $\uu$ and that $\A$ and $\C_o$ are
such that the system (\ref{eq:sys}) is \emph{observable} with respect
to $\yy_o$.

As described previously \citep{GawLorLakGol11}, Equation
(\ref{eq:sys}) subsumes a number of subsystems including the
neuromuscular (actuator dynamics in the engineering
  context) and disturbance subsystems of Figure
\ref{fig:MSS_arch}.

\subsection{Observer design and sensor fusion}
\label{sec:obsdes}
The system states $\xx$ of Equation (\ref{eq:sys}) are rarely
available directly due to sensor placement or sensor noise. As
discussed in the textbooks \citep{KwaSiv72,GooGraSal01}, an
\emph{observer} can be designed based on the system model
(\ref{eq:sys}) to approximately deduce the system states $\xx$ from
the measured signals encapsulated in the vector $\yy_o$.  
In particular, the observer is given by:
\begin{align}
  \ddt{{\xo}}(t)  &= \A_o {\xo}(t) + \B \uu(t) + \LL [\yy_o(t) - \vv_y(t) ] \label{eq:obs}\\
  \text{where } \A_o &= \A - \LL\C_o \label{eq:A_o}
\end{align}
where the signal $\vv_y(t)$ is the measurement noise. The $n \times
n_o$ matrix $\LL$ is the \emph{observer gain matrix}. As discussed by,
for example \citet{KwaSiv72} and \citet{GooGraSal01}, it is
straightforward to design $\LL$ using a number of approaches including
pole-placement and the linear-quadratic optimisation approach. The
latter is used here and thus
\begin{equation}
  \label{eq:L_0}
  \LL = \LL_0
\end{equation}
where $\LL_0$ is the observer gain matrix obtained using
linear-quadratic optimisation.

The observer deduces system states from the $n_o$ observed signals
contained in $\yy_o$; it is thus a particular form of sensor fusion
with properties determined by the $n \times n_y$ matrix  $\LL$.

As discussed by \citet{GawLorLakGol11}, because the system
(\ref{eq:sys}) contains the disturbance dynamics of Figures
\ref{fig:MSS_arch} and \ref{fig:MIC_arch}, the corresponding observer
deduces not only the state of the blocks labeled ``System'' and
``NMS'' in Figures
\ref{fig:MSS_arch} and \ref{fig:MIC_arch}, but also the state of block
labelled ``Dist.''; thus it acts as a \emph{disturbance
  observer} \citep[Chap. 14]{GooGraSal01}. A simple example appears in
Section \ref{sec:ex:elementary}.

\subsection{Prediction}
\label{sec:pred}
Systems and controllers may contain pure time delays. Time delays are
traditionally overcome using a \emph{predictor}. The predictor of
\citet{Smi59} (discussed by \citet{Ast77}) was an early attempt at
predictor design which, however, cannot be used when the controlled
system is unstable. State-space based predictors have been developed
and used by a number of authors including \citet{Ful68},
\citet{Kle69},  \citet{SagMel71} and  \citet{Gaw76a}. 

In particular, following \citet{Kle69}, a state \emph{predictor} is
given by:
\begin{equation}
  \label{eq:pred}
  \xop(t+\del) = e^{\A\del}\xo(t) + \int_0^{\del}
  e^{\A \tpr}\B \uu(t-\tpr) d \tpr
\end{equation}
Again, apart from the scalar $u$ being replaced by the vector $\uu$
and $\B$ becoming an $n_u\times n$ matrix, Equation (\ref{eq:pred}) is
the same as in the single input ($n_u=1$) case.

\subsection{Controller design and motor synergies}
\label{sec:condes}
As described in the textbooks, for example \citet{KwaSiv72} and
\citet{GooGraSal01}, the LQ controller problem involves minimisation
of:
\begin{equation}
  \label{eq:LQ}
  \int_0^{t_1} \xx^T(t)\Q_c\xx(t) + \uu^T(t)\R_c\uu(t) \; dt
\end{equation}
and letting $t_1 \rightarrow \infty$. $\Q_c$ is the $n\times n$
state-weighting matrix and $\R_c$ is the $n_u \times n_u$
control-weighting matrix.  $\Q_c$ and $\R_c$ are used as design
parameters in the rest of this Chapter. As discussed previously
\citep{GawLorLakGol11}, the resultant state-feedback gain $\kk$
($n\times n_u$) may be
combined with the predictor equation (\ref{eq:pred}) to give the
control signal $\uu$
\begin{align}
%%  \uu(t) &=  \kk \xop(t-\del) \label{eq:u}\\
  \uu(t) &=  \kk \xw(t) \label{eq:u}\\
  \text{where } \xw &= \xop(t) - \xss w(t) \label{eq:x_w}
\end{align}
As discussed by \citet{Kle69}, the use of the state predictor gives a
closed-loop system with no feedback delay and dynamics determined by
the delay-free closed loop system matrix $\A_c$ given by:
\begin{equation}
  \label{eq:A_c}
  \A_c = \A - \B\kk
\end{equation}

As mentioned by \citet{TodJor02} and \citet{Tod04}, control synergies
arise naturally from optimal control and are defined by the elements
the $n_u \times n$ matrix $\kk$.

A key result of state-space design in the delay free case is the
\emph{separation principle} (see \citet[Section 5.3]{KwaSiv72} and
\citet[Section 18.4]{GooGraSal01}) whereby the observer and controller can
be design separately. 
% As discussed in Section \ref{sec:error}, this
% separation result can be extended to the delay/predictor case.

\subsection{Steady-State design}
\label{sec:steady-state}
As discussed in the single-input, single output case by
\citet{GawLorLakGol11}, there are many ways to include the setpoint in
the feedback controller and one way is to compute the steady-state
state $\xss$ and control signal $\uss$ corresponding to the
equilibrium of the ODE (\ref{eq:sys}):
\begin{align}
  \ddt{{\xx}} &= \Z_{n\times1}   \label{eq:equilibrium}\\
  \yss &= \C \xss\label{eq:equilibrium_y}
\end{align}
corresponding to a given constant value of output $\yss$. As discussed
by \citet{GawLorLakGol11}, the scalars $\xss$ and $\uss$ are uniquely
determined by $\yss$.
In contrast, the multivariable case has additional flexibility; this
section takes advantage of this flexibility by extending the
equilibrium design in various ways.

In particular, Equation (\ref{eq:equilibrium_y}) is replaced by
\begin{equation}
  \label{eq:y_ss}
  \yss = \C_{ss} \xss
\end{equation}
where $\yss$ is a constant $n_{ss}\times m_{ss}$ matrix,
$\xss$ is a constant $n \times m_{ss}$ matrix,
and $\C_{ss}$ is an $n_{ss}\times m_{ss}$ matrix.

% where the $i$th of the $m_{ss}$ columns of the $n \times m_{ss}$
% steady-state matrix $\xss$ depends on the $i$th column of
% $\yss$.
%
Typically, the equilibrium space defined by $\yss$ corresponds to
the task space so that, with reference to Equation (\ref{eq:sys}),
each column of $\yss$ is a steady-state value of $\yy$ (for
example, $\yss = \I_{n_y\times n_y}$) and $\C_{ss} = \C$.
Further, assume that the disturbance $\ddd(t)$ of (\ref{eq:sys}) has
$m_{ss}$ alternative constant values which form the columns of the
$n_u \times m_{ss}$ matrix $\dss$.

Substituting the steady-state condition of Equation
(\ref{eq:equilibrium}) into Equation (\ref{eq:sys}) and combining with
Equation (\ref{eq:y_ss}) gives:
\begin{align}
  \Ss
  \begin{bmatrix}
    \xss \\ \uss
  \end{bmatrix}
  &=
  \begin{bmatrix}
    -\B_d \dss\\yss
  \end{bmatrix}  \label{eq:x_ss}\\
  \text{where } \Ss &=  
  \begin{bmatrix}
    \A & \B\\
    \C_{ss} & \Z_{n_{ss} \times n_u}
  \end{bmatrix} \label{eq:Ss}
\end{align}
The matrix $\Ss$, has $n+n_{ss}$ rows $n+n_{u}$ columns, thus there
are three possibilities:
%
%\tbd{sec:steady-state}{Expand on rank condition}
%
\begin{description}
\item [$n_{ss}=n_u$] If $\Ss$ is full rank, Equation (\ref{eq:x_ss})
  has a unique solution for ${\xx}_{ss}$ and ${\uu}_{ss}$.
\item [$n_{ss}<n_u$] Equation (\ref{eq:x_ss}) has many solutions
  corresponding to a low dimensional manifold in a high dimensional
  space. A particular solution may be chosen to satisfy an additional
  criterion such as a minimum norm solution. An example is given in
  Section~\ref{sec:track}.
\item [$n_{ss}>n_u$] Equation (\ref{eq:x_ss}) is over-determined; a
  least-squares solution is possible. This case is considered in more
  detail in Section \ref{sec:constrained} and an example is given in
  Section~\ref{sec:dist}.
\end{description}

Having obtained a solution for $\xss$, each of the $m_{ss}$ columns of
the $n \times m_{ss}$ steady-state matrix $\xss$ can be associated
with an element of a $m_{ss} \times 1$ \emph{weighting vector}
$\ww(t)$.  The error signal $\xow(t)$ is then defined as as the
difference between the estimated state $\xo(t)$ and the weighted
columns of $\xss$ as:
\begin{align}
  \xow(t) &= \xo(t)- \xss\ww(t) \label{eq:xow}
\end{align}
Following \citet{GawLorLakGol11}, $\xow(t)$ replaces $\xo$ in the
predictor equation (\ref{eq:pred}) and the state feedback controller
remains Equation (\ref{eq:u}).

\paragraph{Remarks.}
%%\tbd{sec:steady-state}{Rewrite remarks as a paragraph}
\begin{enumerate}
\item In the single input case ($n_u=1$) setting $\yss=1$ and
  $\dss=0$ gives the same formulation as given by
  \citet{GawLorLakGol11} and $\ww(t)$ is the setpoint.
\item Disturbances may be unknown. Thus using this approach requires
  disturbances to be estimated in some way.
\item Setpoint tracking is considered in Section \ref{sec:track}.
\item The effect of a constant disturbance is considered in Section
  \ref{sec:dist}.
\item Constrained solutions are considered in Section
  \ref{sec:constrained_steady-state}.
\end{enumerate}

\section{Intermittent control}
\label{sec:ic}
\begin{figure}
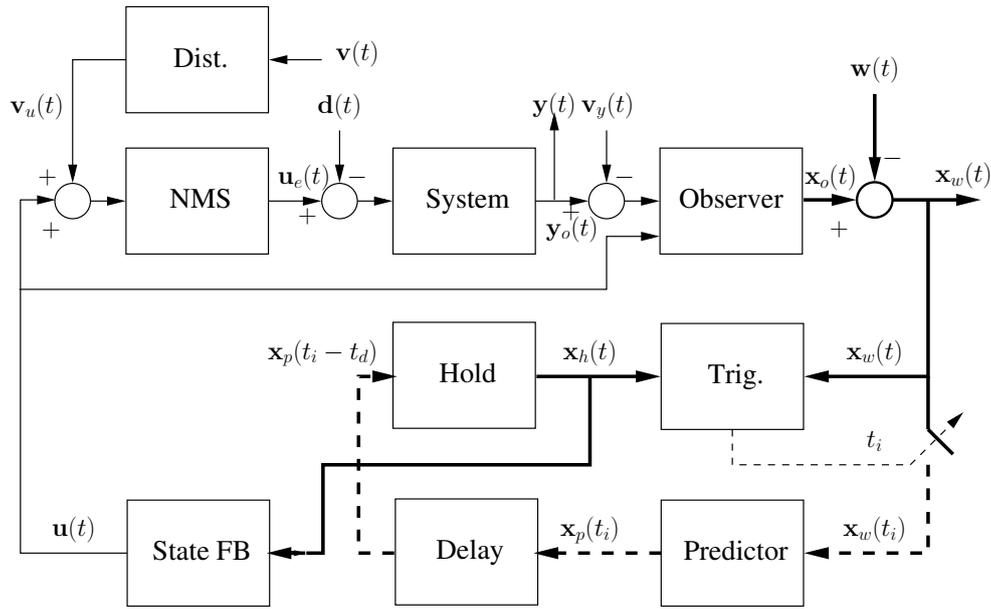

  \centering
  \Fig{MIC_arch}{\Figsize}
    \caption[Intermittent control.]{Intermittent control. This diagram has blocks in common
    with those of the OPF of Figure \ref{fig:MSS_arch}: ``NMS'', ``Dist.'',
    ``System'', ``Observer'', ``Predictor'' and ``State FB'' which
    have the same function; the continuous-time ``Predictor'' block of
    Figure \ref{fig:MSS_arch} is replaced by the much simpler intermittent
    version here. There are three new elements: a sampling element
    which samples $\xow$ at discrete times $t_i$; the block labelled
    ``Hold'', the system-matched hold, which provides the
    continuous-time input to the ``State FB'' block and and the event
    detector block labelled ``Trig.''  which provides the trigger for
    the sampling times $t_i$. The dashed lines represent sampled
    signals defined only at the sample instants $t_i$.
    This figure is based on \citet[Fig. 2]{GawLorLakGol11}.  }
  \label{fig:MIC_arch}
\end{figure}
Intermittent control is based on the underlying continuous-time design
of Section \ref{sec:continuous}. The purpose is to allow control
computation to be performed intermittently at discrete time points
-- which may be determined by time (clock-driven) or the system state
(event-driven) -- whilst retaining much of the continuous-time
behaviour. 

A disadvantage of traditional clock-driven discrete-time control
\citep{FraPow80,Kuo80} based on the zero-order hold is that the control needs
to be redesigned for each sample interval. This also means that the
zero-order hold approach is inappropriate for event-driven
control. The intermittent approach avoids these issues by replacing
the zero-order hold by the \emph{system-matched hold} (SMH).  Because
the SMH is based on the system state, it turns out that it does not
depend on the number of system inputs $n_u$ or outputs $n_y$ and
therefore the SMH described by \citet{GawLorLakGol11} in the single
input $n_u=1$, single output context $n_y=1$ context carries over to
the multi-input $n_u>1$ and multi-output $n_y>1$ case.

This section is a tutorial introduction to the SMH based intermittent
controller in both clock-driven and event-driven cases. Section
\ref{sec:ic_time} looks at the various time-frames involved, Section
\ref{sec:ic_hold} describes the system-matched hold (SMH) and Sections
\ref{sec:ic_obs}~--~\ref{sec:ic_sfb} look at the
observer, predictor and feedback control, developed in the
continuous-time context in Section \ref{sec:continuous}, in the
intermittent context. Section \ref{sec:ic_event} looks at the event
detector used for the event-driven version of intermittent control.

% The intermittent control model of \citet{GawLorLakGol11} is based on
% an underlying continuous-time state-space system and corresponding
% observer/predictor/state-feedback design method identical to that of
% Section \ref{sec:continuous} except that this Chapter allows more than one
% system output ($n_y>0$) and more than one system input ($n_u>1$).  It
% is assumed that the intermittent controller is associated with a
% control time-delay $\del$.

\subsection{Time  frames}
\label{sec:ic_time}
As discussed by \citet{GawLorLakGol11}, intermittent control makes use
of three time frames:
\begin{enumerate}
\item \textbf{continuous-time}, within  which the controlled system
  \eqref{eq:sys} evolves, which is denoted by $t$.
\item \textbf{discrete-time} points at which feedback occurs indexed by
  $i$. Thus, for example, the discrete-time time instants are denoted
  $t_i$ and the corresponding estimated state is $\xo_i=\xo(t_i)$.
  The $i$th
  \textbf{intermittent interval} $\Delta_{ol}=\Delta_i$\footnote{Within
  this chapter, we will use $\Delta_{ol}$ to refer to the generic
  concept of intermittent interval and $\Delta_i$ to refer to the
  length of the $i$th interval} is defined as
  \begin{equation}
    \label{eq:Delta_i}
    \Delta_{ol} = \Delta_i = t_{i+1}-t_i
  \end{equation}
This Chapter distinguishes between event times $t_i$ and the
corresponding sample times $t^s_i$. In particular, the model of
\citet{GawLorLakGol11} is extended so that sampling occurs a fixed time
$\Delta_s$ after an event at time $t_i$ thus:
\begin{equation}
  \label{eq:t_s_i}
  t^s_i = t_i + \Delta_s
\end{equation}
$\Delta_s$ is called the \emph{sampling delay} in the sequel.
\item \textbf{intermittent-time} is a continuous-time variable, denoted
  by $\tau$, restarting at each intermittent interval. Thus, within the
  $i$th intermittent interval:
  \begin{equation}
    \label{eq:tau}
    \tau = t-t_i
  \end{equation}
  Similarly, define the intermittent time $\tau^s$ after a sample by:
  \begin{equation}
    \label{eq:taus}
    \tau^s = t-t^s_i
  \end{equation}

A lower bound $\Delta_{min}$ is imposed on each
intermittent interval $\Delta_i>0$ (\ref{eq:Delta_i}):
\begin{equation}
  \label{eq:PRP}
  \Delta_i > \Delta_{min} >0
\end{equation}
As discussed by \citet{GawLorLakGol11} and in Section \ref{sec:ex:prp},
$\Delta_{min}$ is related to the the Psychological Refractory Period
(PRP) of \citet{Tel31} as discussed by \citet{Vin48} to explain the
human response to double stimuli.
As well as corresponding to the PRP explanation, the lower bound of
(\ref{eq:PRP}) has two implementation advantages.  
Firstly, as discussed by \citet{RonArsGaw99}, the time taken to
compute the control signal (and possibly other competing tasks) can be
up to $\Delta_{min}$. It thus provides a model for a single processor
bottleneck.
Secondly, as discussed by \citet{GawLorLakGol11}, the predictor
equations are particularly simple if the system time-delay
$\del\le\Delta_{min}$.
\end{enumerate}

\subsection{System-matched hold}
\label{sec:ic_hold}
The system-matched hold (SMH) is the key component of the intermittent
control. As described by \citet[Equation (23)]{GawLorLakGol11}, the
SMH state $\xh$ evolves in the \emph{intermittent} time frame $\tau$
as
    \begin{align}
    \ddtau{{\xh}}(\tau) &= \A_h \xh (\tau)\label{eq:hold}\\
    \text{where } \A_h &= \A_c\\
    \xh(0) &= \xp(t^s_i-\del) \label{eq:x_w_smh}
  \end{align}
where $\A_c$ is the closed-loop system matrix (\ref{eq:A_c}) and
$\xop$ is given by the predictor equation (\ref{eq:pred}). 
The hold state $\xh$ replaces the predictor state $\xop$ in the
controller equation (\ref{eq:u}).
Other holds (where $\A_h \ne \A_c$) are possible
\citep{GawWan07,GawGol12}.

% Equation (\ref{eq:hold}) does not explicitly depend on the number of
% system outputs ($n_y$) and inputs ($n_u$) and is thus unchanged in the
% multivariable context. In essence, the only differences between the
% multivariable intermittent controller of this Chapter and single-input
% single-output intermittent controller are related to the underlying
% design method:
% \begin{enumerate}
% \item the observer gain $\LL$ is an $n \times n_y$ matrix rather than
%   a column vector with $n$ elements;
% \item the controller gain $\kk$  is an $n_u \times n$ matrix rather than
%   a row vector with $n$ elements
% \item and (as discussed in Sections \ref{sec:steady-state} and
%   \ref{sec:constrained_steady-state}) there is more flexibility in
%   choosing the equilibrium state $\xss$ in the context of
%   tracking.
% \end{enumerate}

The intermittent controller generates an open loop control signal
based on the hold state $\xh$ (\ref{eq:hold}). At the intermittent
sample times $t_i$, the hold state is reset to the estimated system
state $\xow$ generated by the observer (\ref{eq:obs}); thus feedback
occurs at the intermittent sample times $t_i$. The sample times are
constrained by (\ref{eq:PRP}) to be at least $\Delta_{min}$
apart. But, in addition to this constraint, feedback only takes place
when it is needed; the event detector discussed in Section
\ref{sec:ic_event} provides this information.

\subsection{Intermittent observer}
\label{sec:ic_obs}
\begin{figure}[htbp]
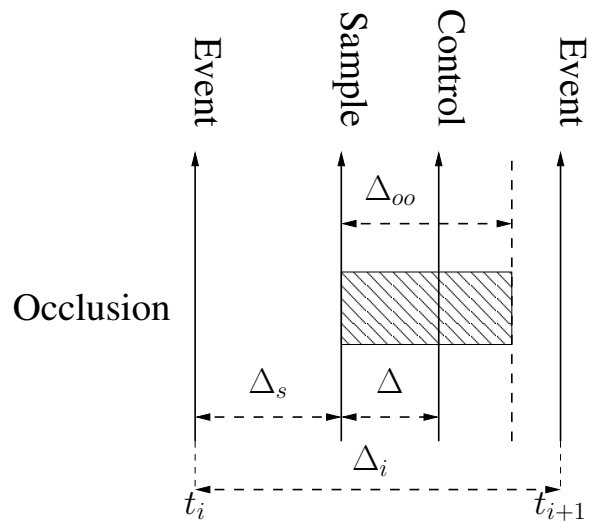

  \centering
  \Fig{timing}{0.5}
  \caption[Self-occlusion]{Self-occlusion. Following an event, the
    observer is sampled at a time $\Delta_s$ and a new control
    trajectory is generated. The observer is then occluded for a
    further time $\Delta_{oo}=\Delta_{o}$ where $\Delta_{o}$ is the
    internal occlusion interval. Following that time, the
    observer is operational and an event can be detected
    (\ref{eq:ED_e}). The actual time between events $\Delta_i >
    \Delta_s + \Delta_{oo}$}
  \label{fig:self-occlusion}
\end{figure}
The intermittent controller of \citet{GawLorLakGol11} uses continuous
observation however, motivated by the occlusion experiments of
\citet{Osb13}, this chapter looks a intermittent observation.

As discussed in Section \ref{sec:ic_hold}, the predictor state $\xp$
is only sampled at discrete-times $t_i$. Further, from Equation
(\ref{eq:pred}), $\xp$ is a function of $\xo$ at these times. Thus the
only the observer performance at the discrete-times $t_i$ is
important. With this in mind, this Chapter proposes, in the context of
intermittent control, that the continuous observer is replaced by an
intermittent observer where periods of monitoring the system
measurements are interleaved with periods where the measurement is
\emph{occluded}. In particular, and with reference to Figure
\ref{fig:self-occlusion}, this Chapter examines the situation where
observation is occluded for a time $\Delta_{oo}$ following sampling.
Such occlusion is equivalent to setting the observer gain $L=0$ in
Equation (\ref{eq:obs}).  Setting $L=0$ has two consequences: the
measured signal $y$ is ignored and the observer state evolves as the
disturbance-free system.

With reference to Equation (\ref{eq:x_w_smh}); the intermittent
controller only makes use of the state estimate at the discrete time
points at $t=t^s_i$ (\ref{eq:t_s_i}); moreover, in the event-driven
case, the observer state estimate is used in Equation (\ref{eq:ED_e})
to determine the event times $t_i$ and thus $t^s_i$.
Hence, a good state estimate immediately after an sample at time $t^s_i$
is not required and so one would expect that occlusion ($L=0$) would
have little effect immediately after $t=t^s_i$. For this reason, define
the occlusion time.  $\Delta_{oo}$ as the time after $t=t^s_i$ for
which the observer is open-loop $L=0$. That is, the constant observer
gain  is replaced by the time varying
observer gain:
\begin{equation}
  \label{eq:L}
  L(t) =
  \begin{cases}
    0 & \tau^s <\Delta_{oo}\\
    L_o & \tau^s \ge \Delta_{oo}
  \end{cases}
\end{equation}
where $L_o$ is the observer gain designed using standard techniques
\citep{KwaSiv72,GooGraSal01} and the intermittent time $\tau^s$ is given
by (\ref{eq:taus}). 

\subsection{Intermittent predictor}
\label{sec:ic_pred}
The continuous-time predictor of equation (\ref{eq:pred}) contains a
convolution integral which, in general, must be approximated for
real-time purposes and therefore has a speed-accuracy trade-off. This
section shows that the use of intermittent control, together with the
hold of Section \ref{sec:ic_hold}, means that equation
(\ref{eq:pred}) can be replaced by a simple exact formula.

Equation (\ref{eq:pred}) is the solution of the differential equation
(in the intermittent time $\tau$ (\ref{eq:tau}) time frame)
\begin{equation}
  \label{eq:pred_de}
    \begin{cases}
    \ddtau{{\xp}}(\tau) &= \A \xp (\tau) + \B u(\tau)\\
    \xp(0) &= \xw(t^s_i)
  \end{cases}
\end{equation}
evaluated at time $\tau^s=\del$ where $\tau^s_i$ is given by Equation
(\ref{eq:t_s_i}). However, the control signal $\uu$ is not arbitrary but
rather given by the hold equation (\ref{eq:hold}).  Combining
equations (\ref{eq:hold}) and (\ref{eq:pred_de}) gives
\begin{equation}
  \label{eq:X}
  \begin{cases}
    \ddtau{\X}(\tau) &= \A_{ph} \X(\tau) \\
    \X(0) &= \X_i
  \end{cases} 
\end{equation}
where
\begin{align}
  \X(\tau) &=
  \begin{pmatrix}
    \xp(\tau)\\\xh(\tau)
  \end{pmatrix}\\
  \X_i &=
  \begin{pmatrix}
    \xw(t_i)\\\xp(t_i-\del)
  \end{pmatrix}\\
  \text{and } \A_{ph} &=
  \begin{pmatrix}
    \A & -\B \kk\\
    \Z_{n \times n} & \A_h
  \end{pmatrix}\label{eq:A_xu}
\end{align}
where $\Z$ is a zero matrix of the indicated dimensions and the hold
matrix $\A_h$ can be $\A_c$ (SMH) or $\Z$ (ZOH).
% With
% reference to Equation (\ref{eq:zoh}), $\A_c$ is replaced by $\Z$ in
% (\ref{eq:A_xu}) in the case of the zero-order hold.

The equation (\ref{eq:X}) has an explicit solution at time $\tau=\del$
given by:
\begin{equation}
  \label{eq:X_explicit}
  \X(\del) = e^{\A_{ph}\del} \X_i
\end{equation}
The prediction $\xp$ can be extracted from (\ref{eq:X_explicit}) to
give:
\begin{equation}
  \label{eq:x_p_i}
  \xp(t_i) =  \E_{pp} \xw(t_i) +  \E_{ph} \xh(t_i)
\end{equation}
where the $n\times n$ matrices $\E_{pp}$ and $\E_{ph}$ are partitions
of the $2n\times 2n$ matrix $\E$:
\begin{align}
    \E &= \begin{pmatrix}
    \E_{pp} &  \E_{ph}\\
    \E_{hp} &  \E_{hh}
  \end{pmatrix}\\
   \text{where } \E &= e^{\A_{ph}\del}\label{eq:E}
\end{align}

The intermittent predictor (\ref{eq:x_p_i}) replaces the
continuous-time predictor (\ref{eq:pred}); there is no convolution
involved and the matrices $\E_{pp}$ and $\E_{ph}$ can be computed
off-line and so do not impose a computational burden in real-time.

\subsection{State feedback}
\label{sec:ic_sfb}
%%\tbd{sec:ic_sfb}{This section needs some explanatory text.}
The ``state-feedback'' block of Figure \ref{fig:MIC_arch} is
implemented as:
\begin{equation}
  \label{eq:sfb_smh}
  \uu(t) = -\kk \xx_h(t)
\end{equation}
This is similar to the conventional state feedback of Figure
\ref{fig:MSS_arch} given by Equation (\ref{eq:u}) but the continuous
predicted state $\xw(t)$ is repaces by the hold state $\xx_h(t)$
generated by Equation (\ref{eq:hold}).

\subsection{Event detector}
\label{sec:ic_event}
The purpose of the event detector is to generate the intermittent
sample times $t_i$ and thus trigger feedback.  Such feedback is
required when the open-loop hold state $\xh$ (\ref{eq:hold}) differs
significantly from the closed-loop observer state $\xow$
(\ref{eq:xow}) indicating the presence of disturbances.
There are many ways to measure such a discrepancy; following
\citet{GawLorLakGol11}, the one chosen here is to look for a quadratic
function of the error $\ehp$ exceeding a threshold $q_t^2$:
\begin{align}
E &=  \ehp^T(t) \Q_{t} \ehp(t) -  q_t^2 \ge 0\label{eq:ED_e}\\
\text{where } \ehp(t) &= \xh(t) - \xow(t)
\end{align}
where $\Q_{t}$ is a positive semi-definite matrix.

\subsection{The intermittent-equivalent setpoint}
\label{sec:equivalent_setpoint}
\citet{LorKamGolGaw12} introduce the concept of the \emph{equivalent
  setpoint} for intermittent control. This section extends the concept and
there are two differences:
\begin{enumerate}
\item the setpoint sampling occurs at $t_i+\Delta_s$ rather than at
  $t_i$ and
\item the filtered setpoint $\ww_f$ (rather than $\ww$) is sampled.
\end{enumerate}
Define the sample time $t^s_i$ (as opposed to the event time $t_i$ and
the corresponding intermittent time $\tau^s$ by
\begin{align}
  t^s_i &= t_i+\Delta_s\\
  \tau^s &= \tau-\Delta_s = t - t_i - \Delta_s = t - t^s_i
\end{align}

In particular, the sampled setpoint $\ww_s$ becomes:
\begin{equation}
  \label{eq:w_s}
  \ww_s(t) = \ww_f(t^s_i) \text{ for $t^s_i \le t < t^s_{i+1}$}
\end{equation}
where $\ww_f$ is the \emph{filtered} setpoint $\ww$. That is the sampled
setpoint $\ww_s$ is the filtered setpoint at time $t^s_i = t_i+\Delta_s$. 

The equivalent setpoint $\ww_{ic}$ is then given by:
\begin{align}
  \label{eq:w_ic}
  \ww_{ic}(t) &= \ww_s(t-t_d) \\
  &= \ww_f(t^s_i-t_d)\\
  &= \ww_f(t-\tau^s-t_d) \text{ for $t^s_i \le t < t^s_{i+1}$}
\end{align}
This corresponds to the previous result \citep{LorKamGolGaw12} when
$\Delta_s=0$ and $\ww_f(t)=\ww(t)$.

If, however, the setpoint $\ww(t)$ is such that $\ww_f(t^s_i) \approx
\ww(t_s)$ (ie no second stimulus within the filter settling time and
$\Delta_s$ is greater than the filter settling time) then Equation
(\ref{eq:w_ic}) may be approximated by:
\begin{align}
   \ww_{ic}(t) &\approx \ww(t^s_i-t_d) \text{ for $t^s_i \le t <t^s_{i+1}$}\\
   &= \ww(t-\tau^s-t_d) \text{ for $t^s_i \le t < t^s_{i+1}$}
\end{align}

As discussed in Section \ref{sec:iic:intermittency}, the
intermittent-equivalent setpoint is the basis for identification of
intermittent control.

\subsection{The intermittent separation principle.}
\label{sec:ISP}
%%\tbd{sec:ISP}{This section will summarise the results of \cite{GawWan11}}
As discussed in Section \ref{sec:ic_hold}, the Intermittent Controller
contains a \emph{System-Matched Hold} which can be views as a
particular form of generalised hold \citep{GawWan07}. Insofar as this special case of intermittent
control uses a dynamical model of the controlled system to generate
the (open-loop) control between sample intervals, it is related to the
strategies of both \citet{ZhiMid03} and \citet{Ast08}. However, as
shown in this chapter, intermittent control provides a framework within
which to analyse and design a range of control systems with
unpredictable non-uniform sampling possibly arising from an
event-driven design.

In particular, it is shown by \citet{GawWan11}, that the SMH-based
intermittent controller is associated with a separation principle
similar to that of the underlying continuous-time controller, which
states that the closed-loop poles of the intermittent control system
consist of the control system poles and the observer system poles, and
the interpolation using the system matched hold does not lead to the
changes of closed-loop poles. 
As discussed by \citet{GawWan11}, this separation principle is only
valid when using the SMH. For example, intermittent control based on
the standard zero-order hold (ZOH) does not lead to such a separation
principle and therefore closed-loop stability is compromised when the
sample interval is not fixed.

As discussed by \citet{GawWan11}, an important consequence of this
separation principle is that the neither the design of the SMH, nor
the stability of the closed-loop system in the fixed sampling case, is
dependent on sample interval. It is therefore conjectured that the SMH
is particularly appropriate when sample times are unpredictable or
non-uniform, possibly arising from an event-driven design.

\section{Examples: basic properties of  intermittent control}
\label{sec:ex:comp}
This section uses simulation to illustrate key properties of
intermittent control. Section \ref{sec:ex:elementary} illustrates
\begin{itemize}
\item timed \& event-driven control (Section \ref{sec:ic_event}),
\item the roles of the disturbance observer and series integrator
  (Section \ref{sec:obsdes}),
\item the choice of event threshold (Section \ref{sec:ic_event}),
\item the difference between control-delay \& sampling delay (Section \ref{sec:ic_time}),
\item the effect of low \& high observer gain  (Section \ref{sec:obsdes}) and
\item the effect of occlusion  (Section \ref{sec:ic_pred}).
\end{itemize}
Sections \ref{sec:ex:prp} and \ref{sec:ex:atf} illustrates how the
intermittent controller models two basic psychological phenomenon: the
\emph{Psychological Refractory Period} and the \emph{Amplitude
  Transition Function}.

\subsection{Elementary examples}
\label{sec:ex:elementary}
\begin{figure}[htbp]
  \centering
  \SubFig{simple_y_clock}{Timed: $y$}{0.45}
  \SubFig{simple_u_clock}{Timed: $u$}{0.45}\\
  \SubFig{simple_y_event}{Event-driven: $y$}{0.45}
  \SubFig{simple_u_event}{Event-driven: $u$}{0.45}
  \caption{Elementary example: timed \& event-driven}
  \label{fig:ex:elementary:timed}
\end{figure}
This section illustrates the basic properties of intermittent control
using simple examples. In all cases, the system is given by:
\begin{xalignat}{2}
  G_0(s) &= \frac{1}{s^2-1} = \frac{1}{(s-1)(s+1)} &\text{Second-order unstable system}\label{eq:elem:G_0}\\
  G_v(s) &= \frac{1}{s} &\text{Simple integrator for disturbance observer}\label{eq:elem_DO}
\end{xalignat}

The corresponding state-space system (\ref{eq:sys}) is:
\begin{align}
  \A &=
  \begin{pmatrix}
    0 & 0 & 1\\
    1 & 0 & 0\\
    0 & 0 & 0
  \end{pmatrix}\\
  \B &= \B_d =
  \begin{pmatrix}
    1\\0\\0
  \end{pmatrix}\\
  \B_v &=
  \begin{pmatrix}
    0\\0\\1
  \end{pmatrix}\\
  \C &=
  \begin{pmatrix}
    0 & 1 & 0
  \end{pmatrix}
\end{align}

All signals are zero except:
 \begin{xalignat}{2}
  w(t) &= 1 & t \ge 1.1\\
  d(t) &= 0.5 & t \ge 5.1
\end{xalignat}

Except where stated, the intermittent control parameters are:
 \begin{xalignat}{2}
   \Delta_{min} &= 0.5 & \text{Min. intermittent interval}~(\ref{eq:PRP})\notag\\
   q_t &= 0.1 & \text{Threshold} (\ref{eq:ED_e})~\notag\\
   \Delta &= 0 & \text{Control delay} (\ref{eq:pred})~\notag\\
   \Delta_s &= 0 & \text{Sampling delay} (\ref{eq:t_s_i})~\notag
\end{xalignat}

Figures \ref{fig:ex:elementary:timed}--~\ref{fig:ex:elementary:occ}
are all of the same format. The left column of figures shows the
system output $y$ together with the setpoint $w$ and the output $y_c$
corresponding to the underlying continuous-time design; the right
column shows the corresponding control signal $u_e$ together with the
negative disturbance $-d$ and the control $u_c$. In each case the
$\bullet$ symbol corresponds to an event.

Figure \ref{fig:ex:elementary:timed} contrasts timed and event driven
control. In particular, Figures \ref{subfig:simple_y_clock} and
\ref{subfig:simple_u_clock} correspond to zero threshold ($q_t=0$) and
thus timed intermittent control with fixed interval $\Delta_{min} =
0.5$ and Figures \ref{subfig:simple_y_event} and
\ref{subfig:simple_u_event} correspond to event-driven control.  The
event driven case has two advantages: the controller responds
immediately to the setpoint change at time $t=1.1$ whereas the timed
case has to wait until the next sample at $t=1.5$ and the control is
only computed when required. In particular, the initial setpoint
response does not need to be corrected, but the unknown disturbance
means that the observer state is different from the system state for a
while and so corrections need to be made until the disturbance is
correctly deduced by the observer.

\begin{figure}[htbp]
  \centering
  \SubFig{simple_y_noint}{No integrator: $y$}{0.45}
  \SubFig{simple_u_noint}{No integrator: $u$}{0.45}\\
  \SubFig{simple_y_int}{Integrator: $y$}{0.45}
  \SubFig{simple_u_int}{Integrator: $u$}{0.45}
  \caption{Elementary example: no disturbance observer, with and without integrator}
  \label{fig:ex:elementary:int}
\end{figure}
The simulation of Figure \ref{fig:ex:elementary:timed} includes the
disturbance observer implied by the integrator of Equation
\ref{eq:elem_DO}; this means that the controllers are able to
asymptotically eliminate the constant disturbance $d$. Figures
\ref{subfig:simple_y_noint} and \ref{subfig:simple_u_noint} show the
effect of not using the disturbance observer. The constant disturbance
$d$ is not eliminated and the intermittent controller exhibits limit
cycling behaviour (analysed further by \citet{Gaw09}). As an
alternative to the disturbance observer used in the simulation of
Figure \ref{fig:ex:elementary:timed}, a series integrator can be used
by setting:
\begin{xalignat}{2}
  G_s(s) &= \frac{1}{s} &\text{Series integrator for disturbance rejection}\label{eq:elem_int}
\end{xalignat}
The corresponding simulation is shown in Figures
\ref{subfig:simple_y_int} and \ref{subfig:simple_u_int}\footnote{The
  system dynamics are now different; the LQ design parameter is
  set to $Q_c=100$ to account for this.}. Although the constant
disturbance $d$ is now asymptotically eliminated, the additional
integrator increases both the system order and the system relative
degree by one giving a more difficult system to control.

\begin{figure}[htbp]
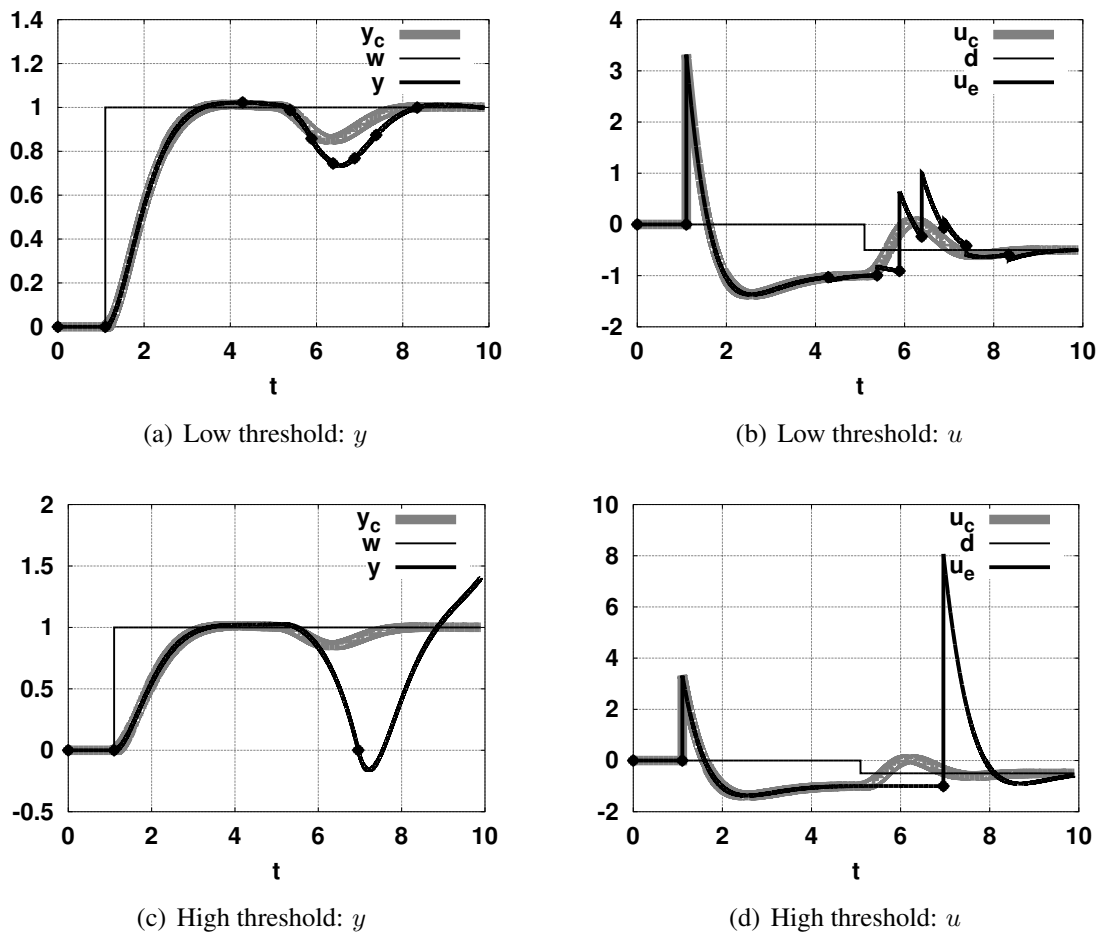

  \centering
  \SubFig{simple_y_threshold_low}{Low threshold: $y$}{0.45}
  \SubFig{simple_u_threshold_low}{Low threshold: $u$}{0.45}\\
  \SubFig{simple_y_threshold_hi}{High threshold: $y$}{0.45}
  \SubFig{simple_u_threshold_hi}{High threshold: $u$}{0.45}
  \caption{Elementary example: low \& high threshold}
  \label{fig:ex:elementary:theshold}
\end{figure}
The event detector behaviour depends on the threshold $q_t$
(\ref{eq:ED_e}); this has already been examined in the simulations of
Figure \ref{fig:ex:elementary:timed}. Figure
\ref{fig:ex:elementary:theshold} shows the effect of a low
($q_t=0.01$) and high ($q_t=1$) threshold. As discussed in the context
of Figure \ref{fig:ex:elementary:timed}, the initial setpoint response
does not need to be corrected, but the unknown disturbance generates
events. The simulations of Figure \ref{fig:ex:elementary:theshold}
indicate the trade-off between performance and event rate determined by
the choice of the threshold $q_t$.

\begin{figure}[htbp]
  \centering
  \SubFig{simple_y_control_del}{Control delay: $y$}{0.45}
  \SubFig{simple_u_control_del}{Control delay: $u$}{0.45}\\
  \SubFig{simple_y_sample_del}{Sample delay: $y$}{0.45}
  \SubFig{simple_u_sample_del}{Sample delay: $u$}{0.45}
  \caption{Elementary example: control-delay \& sampling delay}
  \label{fig:ex:elementary:delay}
\end{figure}
The simulations of Figure \ref{fig:ex:elementary:delay} compare and
contrast the two delays: control delay $\Delta$ and sample delay
$\Delta_s$. In particular, Figures \ref{subfig:simple_y_control_del}
and \ref{subfig:simple_u_control_del} correspond to $\Delta=0.4$ and
$\Delta_s=0$ but Figures \ref{subfig:simple_y_sample_del} and
\ref{subfig:simple_u_sample_del} correspond to $\Delta=0$ and
$\Delta_s=0.4$. The response to the setpoint is identical as
the prediction error is zero in this case; the response to the
disturbance change is similar, but not identical  as
the prediction error is not zero in this case.

\begin{figure}[htbp]
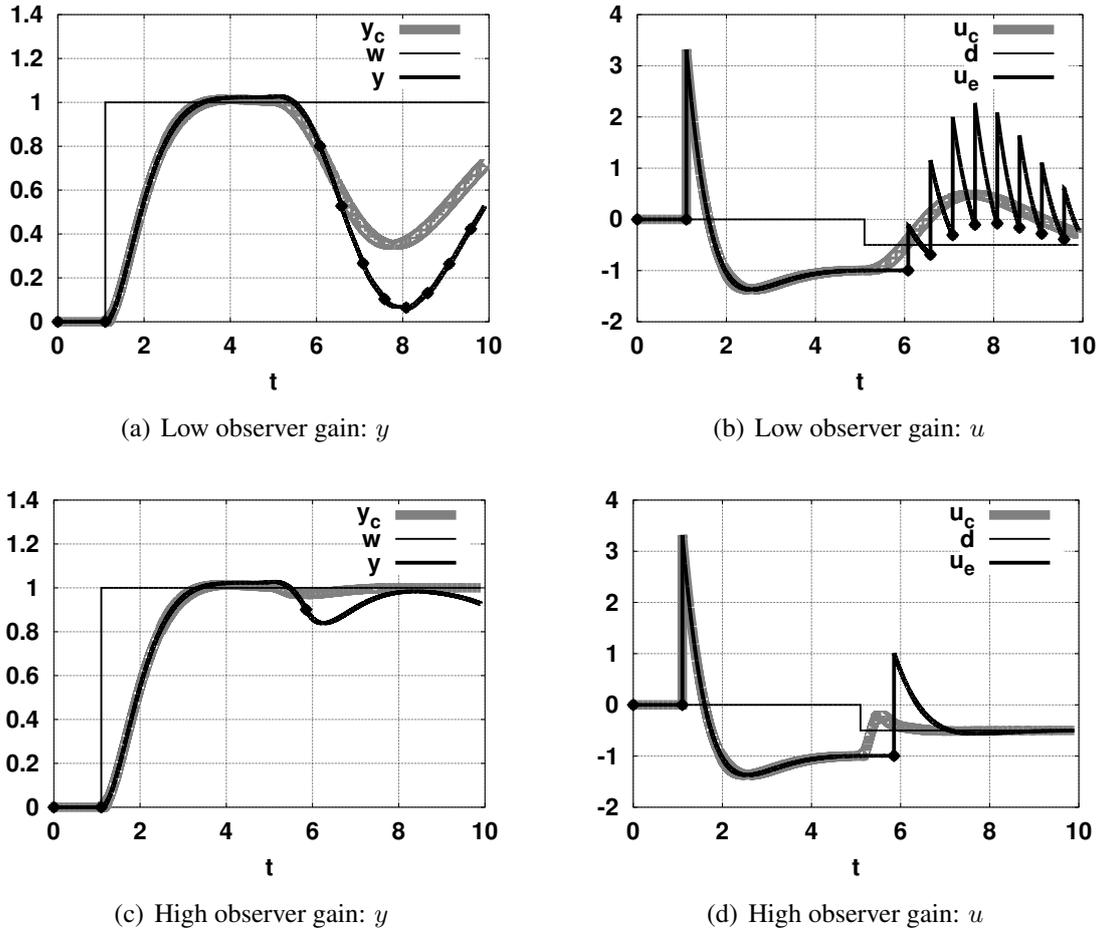

  \centering
  \SubFig{simple_y_obs_low}{Low observer gain: $y$}{0.45}
  \SubFig{simple_u_obs_low}{Low observer gain: $u$}{0.45}\\
  \SubFig{simple_y_obs_hi}{High observer gain: $y$}{0.45}
  \SubFig{simple_u_obs_hi}{High observer gain: $u$}{0.45}
  \caption{Elementary example: low \& high observer gain}
  \label{fig:ex:elementary:obs_gain}
\end{figure}
The state observer of Equation (\ref{eq:obs}) is needed to deduce
unknown states in general and the state corresponding to the unknown
disturbance in particular. As discussed in the textbooks
\citep{KwaSiv72,GooGraSal01}, the choice of observer gain gives a
trade-off between measurement noise and disturbance responses. The gain
used in the simulations of Figure \ref{fig:ex:elementary:timed} can be
regarded as medium; Figure \ref{fig:ex:elementary:obs_gain} looks at
low and high gains. As there is no measurement noise in this case, the
low gain observer gives a poor disturbance response whilst the high
gain gives an improved disturbance response.

\begin{figure}[htbp]
  \centering
  \SubFig{simple_y_occ_low}{Low occlusion time: $y$}{0.45}
  \SubFig{simple_u_occ_low}{Low occlusion time: $u$}{0.45}\\
  \SubFig{simple_y_occ_hi}{High occlusion time: $y$}{0.45}
  \SubFig{simple_u_occ_hi}{High occlusion time: $u$}{0.45}
  \caption{Elementary example: low \& high occlusion time}
  \label{fig:ex:elementary:occ}
\end{figure}
The simulations presented in Figure \ref{fig:ex:elementary:occ}
investigate the intermittent observer of Section \ref{sec:ic_obs}.  In
particular, the measurement of the system output $y$ is assumed to be
\emph{occluded} for a period $\Delta_{oo}$ following a sample. Figures
Figures \ref{subfig:simple_y_occ_low} and \ref{subfig:simple_u_occ_low}
show simulation with $\Delta_{oo}=0.1$ and  Figures
Figures \ref{subfig:simple_y_occ_hi} and \ref{subfig:simple_u_occ_hi}
show simulation with $\Delta_{oo}=0.5$. It can be seen that occlusion
has little effect on performance for the lower value, but performance
is poor for the larger value.

\subsection{The Psychological Refractory Period and
  Intermittent-equivalent setpoint}
\label{sec:ex:prp}
As noted in Section \ref{sec:equivalent_setpoint}, the intermittent
sampling of the setpoint $w$ leads to the concept of the
intermittent-equivalent setpoint: the setpoint that is actually used
within the intermittent controller. Moreover, as noted in Section
\ref{sec:ic_time}, there is a minimum intermittent interval
$\Delta_{min}$. As discussed by \citet{GawLorLakGol11}, $\Delta_{min}$
is related to the \emph{psychological refractory period} (PRP)
\citep{Tel31} which explains the experimental results of \citet{Vin48}
where a second reaction time may be longer that the first. These ideas
are explored by simulation in Figures
\ref{fig:ex:PRP:basic}--~\ref{fig:ex:PRP:del}.
In all cases, the system is given by:
\begin{xalignat}{2}
  G_0(s) &= \frac{1}{s}  &\text{Simple integrator}\label{eq:prp:G_0}\\
\end{xalignat}
The corresponding state-space system (\ref{eq:sys}) is:
\begin{align}
  \A &= 0,\; \B = \C = 1\label{eq:prp:ABC}
 \end{align}
All signals are zero except the signal $w_0$ is defined as:
 \begin{xalignat}{2}
  w_0(t) &= 1 & 0.5 \le t \le 1.5,\; 2.0 \le t \le 2.5,\; 3.0 \le t
  \le 3.2,\; 4.0 \le t \le 4.1
\label{eq:w_0}
\end{xalignat}
and the filtered setpoint $w$ is obtained by passing $w$ through the
low-pass filter $G_w(s)$ where:
\begin{equation}
  \label{eq:G_w}
  G_w(s) = \frac{1}{1 + sT_f}
\end{equation}

Except where stated, the intermittent control parameters are:
 \begin{xalignat}{2}
   \Delta_{min} &= 0.5 & \text{Min. intermittent interval}~(\ref{eq:PRP})\notag\\
   q_t &= 0.1 & \text{Threshold} (\ref{eq:ED_e})~\notag\\
   \Delta &= 0 & \text{Control delay} (\ref{eq:pred})~\notag\\
   \Delta_s &= 0 & \text{Sampling delay} (\ref{eq:t_s_i})~\notag
\end{xalignat}

\begin{figure}[htbp]
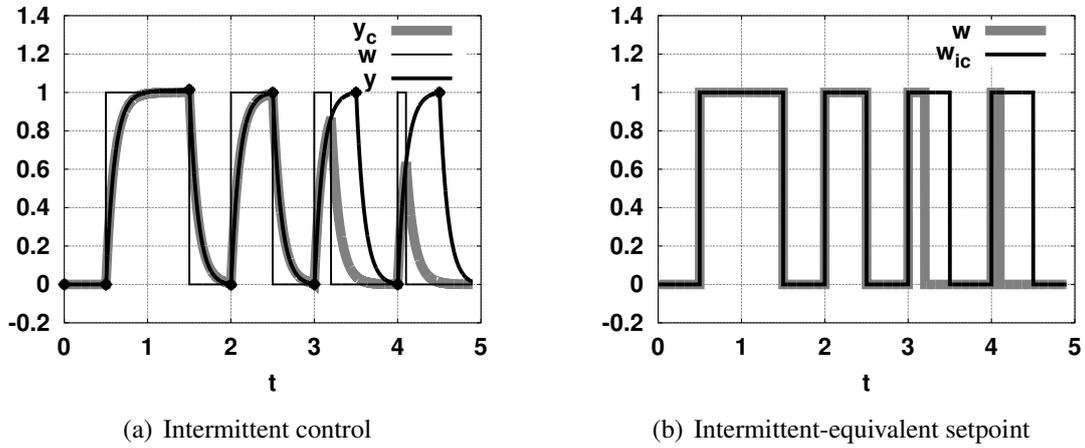

  \centering
  \SubFig{prp_y_basic}{Intermittent control}{0.45}
  \SubFig{prp_w_basic}{Intermittent-equivalent setpoint}{0.45}
  \caption{Psychological Refractory Period: square setpoint}
  \label{fig:ex:PRP:basic}
\end{figure}
Figure \ref{subfig:prp_y_basic} corresponds to the unfiltered setpoint
with $T_f=0$ and $w=w_0$ where $w_0$ is given by (\ref{eq:w_0}). For
the first two (wider) pulses, events ($\bullet$) occur at each
setpoint change; but the second two (narrower) pulses, the trailing
edges occur at a time less that $\Delta_{min}=0.5$ from the leading
edges and thus the events corresponding to the trailing edges are
delayed until $\Delta_{min}$ has elapsed. Thus the the second two
(narrower) pulse lead to outputs as if the pulses were $\Delta_{min}$
wide. Figure \ref{subfig:prp_w_basic} shows the
intermittent-equivalent setpoint $w_{ic}$ superimposed on the actual
setpoint $w$.

\begin{figure}[htbp]
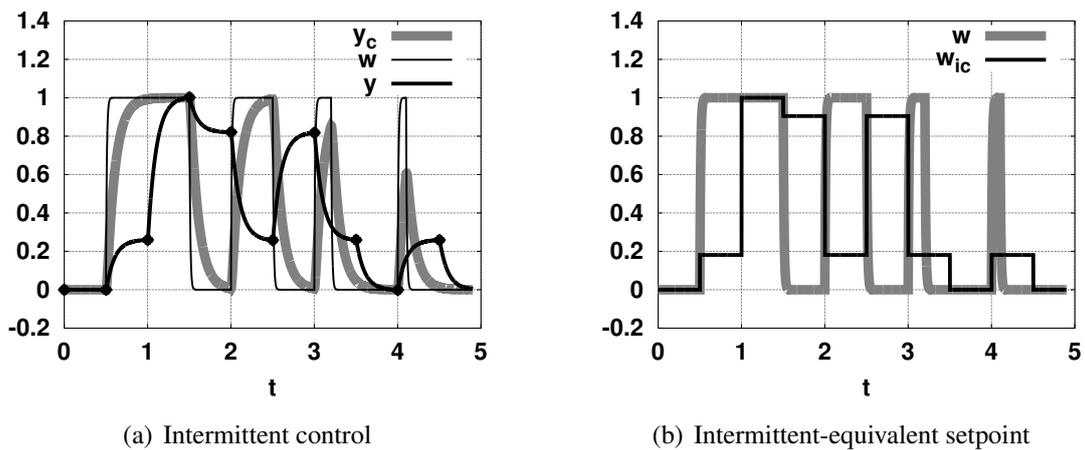

  \centering
  \SubFig{prp_y_filt}{Intermittent control}{0.45}
  \SubFig{prp_w_filt}{Intermittent-equivalent setpoint}{0.45}
  \caption{Psychological Refractory Period: filtered setpoint}
  \label{fig:ex:PRP:filt}
\end{figure}
Figure \ref{subfig:prp_y_filt} corresponds to the filtered setpoint
with $T_f=0.01$ and $w=w_0$ where $w_0$ is given by (\ref{eq:w_0}). At
the event times, the setpoint has not yet reached its final value and
thus the initial response is too small which is then corrected; Figure
\ref{subfig:prp_w_filt} shows the intermittent-equivalent setpoint
$w_{ic}$ superimposed on the actual setpoint $w$.

\begin{figure}[htbp]
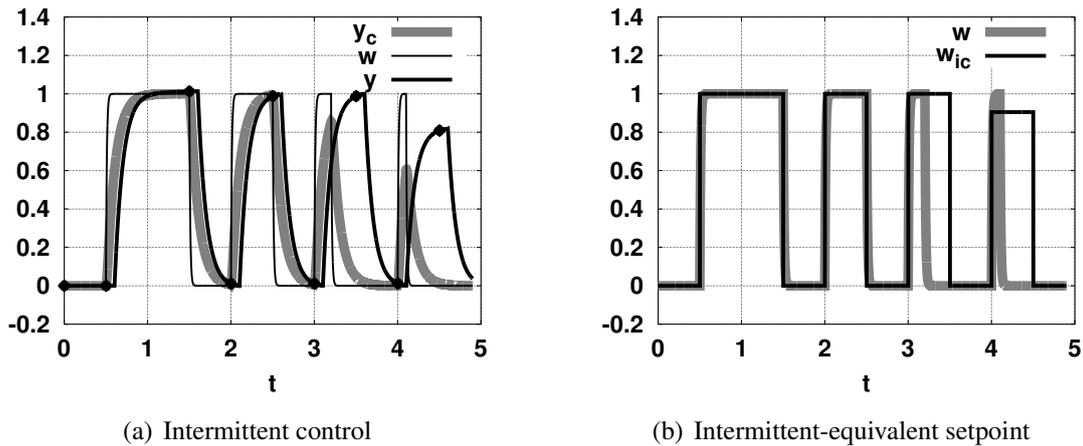

  \centering
  \SubFig{prp_y_del}{Intermittent control}{0.45}
  \SubFig{prp_w_del}{Intermittent-equivalent setpoint}{0.45}
  \caption{Psychological Refractory Period: sampling delay}
  \label{fig:ex:PRP:del}
\end{figure}
The unsatisfactory behaviour can be improved by delaying the sample
time by $\Delta_s$ as discussed in Section \ref{sec:ic_time}. Figure
\ref{subfig:prp_y_del} corresponds to Figure \ref{subfig:prp_y_filt}
exept that $\Delta_s=0.1$. Except for the short delay of
$\Delta_s=0.1$, the behavior of the first three pulses is now similar
to that of Figure \ref{subfig:prp_y_basic}. The fourth (shortest)
pulse gives, however, a reduced amplitude output; this is because the
sample occurs on the trailing edge of the pulse. This behavior has
been observed by \citet{Vin48} as is related to the \emph{Amplitude
  Transition Function} of \citet{BarGle88b}. Figure
\ref{subfig:prp_w_del} shows the intermittent-equivalent setpoint
$w_{ic}$ superimposed on the actual setpoint $w$.  This phenomena is
further investigated in Section \ref{sec:ex:atf}.

\subsection{The Amplitude Transition Function}
\label{sec:ex:atf}
\begin{figure}[htbp]
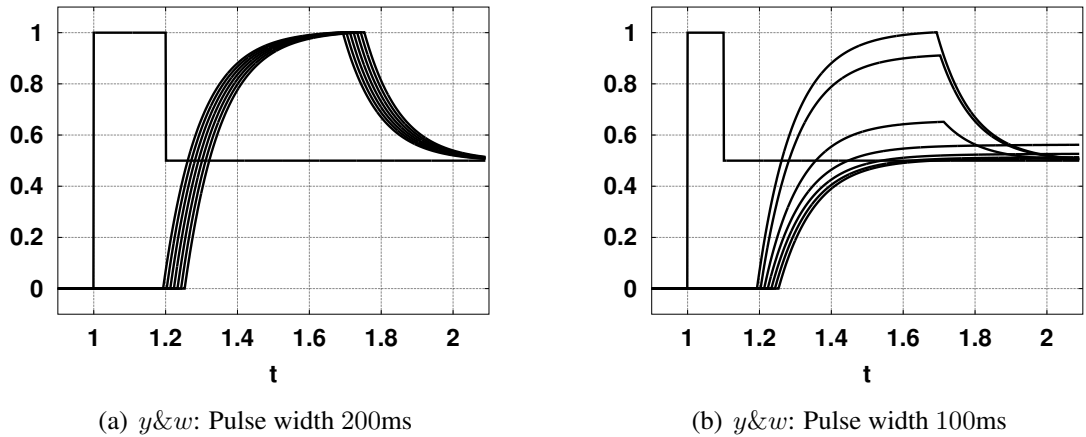

  \centering
  \SubFig{ATF_200}{$y \& w$: Pulse width $200$ms}{0.45}
  \SubFig{ATF_100}{$y \& w$: Pulse width $100$ms}{0.45}
  \caption[Amplitude Transition Function]{Amplitude Transition
    Function.}
  \label{fig:atf}
\end{figure}
This section exands on the observation in Section \ref{sec:ex:prp},
Figure \ref{fig:ex:PRP:del}, that the combination of sampling delay
and a bandwifth limited setpoint can lead to narrow pulses being
``missed''.  
It turns out that the physiological equivalent of this behaviour is
the so called \emph{Amplitude Transition Function} (ATF) described by
\citet{BarGle88b}. Instead of the symmetric pulse discussed in the PRP
context in Section \ref{sec:ex:prp}, the ATF concept is based on
asymmetric pulses where the step down is less than the step up leading
to a non-zero final value. An example of an asymetric pulse appears in
Figure \ref{fig:atf}    
The simulations in this section use the same system as in
Section \ref{sec:ex:prp} Equations (\ref{eq:prp:G_0}) and
(\ref{eq:prp:ABC}), but the setpoint $w_0$ of Equation (\ref{eq:w_0})
is replaced by:
 \begin{equation}
  w_0(t) =
  \begin{cases}
    0 & t<1\\
    1 & 1 \le t \le 1 + \Delta_{p}\\
    0.5 & t> 1 + \Delta_{p}
  \end{cases}
\label{eq:atf:w_0}
\end{equation}
where $\Delta_{p}$ is the \emph{pulse-width}.

The system was simulated for two pulse widths: $\Delta_p=200$ms
(Figure \ref{subfig:ATF_200}) and $\Delta_p=100$ms (Figure
\ref{subfig:ATF_100}). In each case, following Equation
(\ref{eq:atf:w_0}), the pulse was asymmetric going from 0 to 1 and
back to 0.5.

At each pulse width, the system was simulated with event delay
$\Delta_s = 90,100, \dots, 150$ms and the control delay was set to
$100$ms.
Figure \ref{subfig:ATF_200} shows the ``usual'' behaviour, the 200ms
pulse is expanded to $\Delta_{ol}=500$ms and delayed by
$\Delta+\Delta_s$. In contrast, Figure \ref{subfig:ATF_200} shows the
``Amplitude Transition Function'' behaviour: because the sampling is
occurring on the downwards side of the pulse, the amplitude is reduced
with increasing $\Delta_s$. Figure \ref{subfig:ATF_200} is closely
related to Figure 2 of \citet{BarGle88b}.

\section{Constrained design}
\label{sec:constrained}
The design approach outlined in Sections \ref{sec:continuous} and
\ref{sec:ic} assumes that system inputs and outputs can take any
value. In practice, this is not always the case and so
\emph{constraints} on both system inputs and outputs must be taken
into account. There are at least three classes of contraints of
interest in the context of intermittent control:
\begin{enumerate}
\item Constraints on the steady-state behaviour of a system. These are
  particularly relevant in the context of multi-input ($n_u>1$) and
  multi-output ($n_y>1$) systems. This issue is discussed in Section
  \ref{sec:constrained_steady-state} and illustrated by example in
  Section \ref{sec:ex:stand}.
\item Amplitude constraints on the dynamical behaviour of a system. This is a topic
  that is much dicussed in the Model Predictive Control literature --
  for example \citep{Raw00,Mac02,Wan09}. In the context of
  intermittent control, constraints have been considered in the
  single-input single-output context by \citet{GawWan09}; the
  corresponding multivariable case is considered in Section
  \ref{sec:mpc} and illustrated in Section \ref{sec:ex:cms}.
\item Power constraints on the dynamical behaviour of a system. This
  topic has been discussed by \citet{GawWagNeiWan12}.
\end{enumerate}

\subsection{Constrained Steady-State Design}
\label{sec:constrained_steady-state}
 Section \ref{sec:steady-state} considers the steady state design of
 the continuous controller underlying intermittent control. In
 particular, Equation (\ref{eq:x_ss}) gives a linear algebraic
 equation giving the steady-state system state $\xss$ and
 corresponding control signal $\uss$ yeilding a particular
 steady-state output $\yss$. Although in the single-input
 single-output case considered by \citet[Equation 13]{GawLorLakGol11}
 the solution is unique, as discussed in Section
 \ref{sec:steady-state} the multi-input, multi-output case gives rise
 to more possibilities.
 In particular, it is not possible to exactly solve Equation
 (\ref{eq:x_ss}) in the over-determined case where $n_{ss}>n_u$, but a
 least-squares solution exists.
 In the constrained case, this solution must satisfy two sets of
 constraints: an equality constraint ensuring that the equilibrium
 condition (\ref{eq:equilibrium}) holds and inequality constraints to
 reject physically impossible solutions.

In this context, the $n_{ss}\times n_{ss}$ \emph{weighting matrix}
$\Q_{ss}$ can be used to vary the relative importance each element of
$\yss$. In particular, define:
\begin{align}
  \SsQ &=   
  \begin{bmatrix}
    \A & \B\\
    \Q_{ss}\C_{ss} & \Z_{n_{ss} \times n_u}
  \end{bmatrix}\label{eq:Sq}\\
  \Xss &=   
  \begin{bmatrix}
    \xss \\ \uss
  \end{bmatrix}\\
  \yss &= 
  \begin{bmatrix}
    -\B_d \dss\\ \Q_{ss} \yss
  \end{bmatrix}\\
  \text{and }
  \Yh_{ss} &= 
\begin{bmatrix}
    -\B_d \dss\\ \Q_{ss} \yh_{ss}
  \end{bmatrix} = 
\SsQ \Xh_{ss}
\end{align}

This gives rise to the least-squares cost function:
\begin{align}
  J_{ss} &= \left(\yss - \Yh_{ss}\right)^T \left(\yss - \Yh_{ss}\right) \notag\\
  &= \left(\yss-\SsQ\Xh_{ss} \right)^T \left(\yss- \SsQ\Xh_{ss} \right)\label{eq:J_ss}
 \end{align}
Differentiating with respect to $\Xh_{ss}$ gives
the weighted least-squares solution of (\ref{eq:x_ss}):
\begin{align}
  \SsQ^T \left(\yss - \SsQ\Xh_{ss}\right ) &= 0\\
  \text{or } \Xh_{ss} &= \left ( \SsQ^T \SsQ \right )^{-1}\SsQ^T\yss
\end{align}

As ${\xx}_{ss}$ corresponds to a steady state solution corresponding to
Equation (\ref{eq:equilibrium}), the solution of the least-squares
problem is subject to the equality constraint:
\begin{equation}\label{eq:equality}
   \begin{bmatrix}
    \A & \B 
  \end{bmatrix}
  \Xh_{ss} = \A\xss + \B \uss = -\B_d \dss
\end{equation}

Furthermore, suppose that the solution must be such that the components of
$\Y$ corresponding to $\yss$ are bounded above and below:
\begin{equation}
  \label{eq:cons}
  \yh_{min} \le \yh = \C_{ss}\Xh \le \yh_{max}
\end{equation}
Inequality (\ref{eq:cons}) can be rewritten as:
\begin{equation}
  \label{eq:cons_1}
  \begin{bmatrix}
    -\C_{ss} \\\C_{ss}
  \end{bmatrix} \Xh \le 
 \begin{bmatrix}
    -\yh_{min} \\ \yh_{max}
  \end{bmatrix}
\end{equation}
The quadratic cost function (\ref{eq:J_ss}) together with the linear
equality constraint (\ref{eq:equality}) and the linear inequality
constraint (\ref{eq:cons_1}) forms a \emph{quadratic program} (QP)
which has well-established numerical algorithms available for its
solution \citep{Fle87}.

An example of constrained steady-state optimisation is given in
Section \ref{sec:dist}.

\subsection{Constrained Dynamical Design}
\label{sec:mpc}
Model-based predictive control (MPC) \citep{Raw00,Mac02,Wan09} combines a
quadratic cost function with linear constraints to provide optimal
control subject to (hard) constraints on both state and control
signal; this combination of quadratic cost and \emph{linear}
constraints can be solved using \emph{quadratic programming}
(QP) \citep{Fle87,BoyVan04}. Almost all MPC algorithms have a
discrete-time framework.
As a move towards a continuous-time formulation of intermittent
control, the intermittent approach to MPC was introduced
\citep{RonArsGaw99} to reduce on-line computational demand whilst
retaining continuous-time like behaviour
\citep{GawWan07,GawWan09,GawLorLakGol11}. This section introduces and
illustrates this material\footnote{Hard constraints on input
  \emph{power flow} are considered by \citet{GawWagNeiWan12} -- these
  lead to \emph{quadratically-constrained quadratic programming} (QCQP)
  \citep{BoyVan04}.}.

Using the feedback control comprising the system matched hold
(\ref{eq:hold}), its initialisation (\ref{eq:x_w_smh}), and
feedback (\ref{eq:sfb_smh}) may cause state or input constraints to be
violated over the intermittent interval. The key idea introduced by
\citet{CheGaw06} and exploited by \citet{GawWan09} is to replace the SMH
initialisation (at time $t=t_i$ (\ref{eq:Delta_i})) of Equation (\ref{eq:x_w_smh}) by:
\begin{align}
  \xh(0) & = \begin{cases}
  \xop(t_i-\del) - \xss w(t_i) & \text{when constraints not violated}\\
  \UU_i & \text{otherwise}
\end{cases}
\end{align}
where $\UU_i$ is the result of the on-line optimisation to be discussed
in Section \ref{sec:opt}. 

The first step is to construct a set of equations describing the
evolution of the system state $\xx$ and the generalised hold state $\xx_h$
as a function of the initial states and assuming that disturbances are
zero.

The differential equation \eqref{eq:X} has the explicit solution
\begin{align}
  \X(\tau) &= \E(\tau)\X_i\label{eq:X(tau)}\\
  \text{where } \E(\tau) &= e^{\A_{xu}\tau}\label{eq:Etau}
\end{align}
where $\tau$ is the intermittent continuous-time variable based on
$t_i$. 

\subsubsection{Constraints}
\label{sec:constraints}
The vector $\X$ \eqref{eq:X} contains the system state and the state
of the generalised hold; equation \eqref{eq:X(tau)} explicitly give
$\X$ in terms of the system state $\xx_i(t_i)$ and the hold state
$\xh(t_i)=\UU_i$ at time $t_i$. Therefore any constraint expressed at a
future time $\tau$ as a linear combination of $\X$ can be re-expressed
in terms of $\xh$ and $\UU_i$. In particular if the constraint at time
$\tau$ is expressed as:
\begin{equation}
  \label{eq:constraint}
  \Gamma_\tau \X(\tau) \le \gamma_\tau
\end{equation}
where $\Gamma_\tau$ is a $2n$-dimensional row vector and $\gamma_\tau$
a scalar then the constraint can be re expressed using \eqref{eq:X(tau)}
in terms of the intermittent control vector $\UU_i$ as:
\begin{equation}
  \label{eq:U_const}
  \Gamma_\tau E_u(\tau) \UU_i \le \gamma_\tau - \Gamma_\tau E_x(\tau)\xx_i
\end{equation}
where $E$ has been partitioned into the two $2n \times n$ sub-matrices $E_x$
and $E_u$ as:
\begin{equation}
  \label{eq:E_u}
  E(\tau) =
  \begin{pmatrix}
    E_x(\tau) & E_u(\tau)
  \end{pmatrix}
\end{equation}
If there are $n_c$ such constraints, they can be combined as:
\begin{equation}
  \label{eq:U_consts}
  \Gamma \UU_i \le \gamma - \Gamma_x \xx_i
\end{equation}
where each row of $\Gamma$ is $\Gamma_\tau E_u(\tau)$, each row of
$\Gamma_x$ is $\Gamma_\tau E_x(\tau)$
and each (scalar) row of $\gamma$ is $\gamma_\tau$.

Following standard MPC practice, constraints beyond the intermittent
interval can be included by assuming that the the control strategy will
be open-loop in the future.

\subsubsection{Optimisation}
\label{sec:opt}
Following, for example, \citet{CheGaw06}, a modified version of the
infinite-horizon LQR cost \eqref{eq:LQ} is used:
\begin{equation}
  \label{eq:J_ic}
  J_{ic} =  \int_0^{\tau_1} \xx(\tau)^T \Q \xx(\tau) +
  u(\tau)\R u(\tau) \, d\tau 
+ \xx(\tau_1)^T \P \xx(\tau_1)
\end{equation}
where the weighting matrices $\Q$ and $\R$ are as used in
\eqref{eq:LQ} 
and $\P$ is the positive-definite solution of the algebraic Riccati
equation (ARE):
\begin{equation}
  \label{eq:are}
  \A^T \P + \P\A - \P\B\R^{-1}\B^T\P + \Q = 0
\end{equation}

There are an number of differences between our approach to minimising
$J_{ic}$ \eqref{eq:J_ic}  and the LQR approach to minimising $J_{LQR}$
\eqref{eq:LQ}.
\begin{enumerate}
\item Following the standard MPC approach \citep{Mac02}, this is a
  \emph{receding-horizon} optimisation in the time frame of $\tau$ not
  $t$.
\item The integral is over a finite time $\tau_1$.
\item A terminal cost is added based on the steady-state ARE
  \eqref{eq:are}. In the discrete-time context, this idea is due to
  \citet{RawMus93}.
\item The minimisation is with respect to the intermittent control
  vector $\UU_i$ generating the the control signal $u$ \eqref{eq:u}
  through the generalised hold \eqref{eq:hold}.
\end{enumerate}

Using $\X$ from \eqref{eq:X(tau)}, \eqref{eq:J_ic} can be rewritten as
\begin{align}
  J_{ic} &= \int_0^{\tau_1} \X(\tau)^T \Q_{xu} \X(\tau) \,
  d\tau  +  \X(\tau_1)^T \P_{xu} \X(\tau_1)\label{eq:J_ic_1}\\
  \text{where } \Q_{xu} &=
  \begin{pmatrix}
    \Q & 0_{n\times n}\\
    0_{n\times n} & \xx_{uo} \R \xx_{uo}^T
  \end{pmatrix}\\
  \text{and } \P_{xu} &=
  \begin{pmatrix}
    \P & 0_{n\times n}\\
    0_{n\times n} & 0_{n\times n} 
  \end{pmatrix}
\end{align}
Using \eqref{eq:X(tau)}, equation \eqref{eq:J_ic_1} can be rewritten as:
\begin{align}
  J_{ic} &=  \X_i^T J_{XX} \X_i\label{eq:J_ic_2}\\
  \text{where } J_{XX} &= J_{1} +  e^{\A_{xu}^T\tau_1} \P_{xu}  e^{\A_{xu}\tau_1} \\
\text{and } J_{1} &= \int_0^{\tau_1} 
  e^{\A_{xu}^T\tau} \Q_{xu}  e^{\A_{xu}\tau} \, d\tau 
\end{align}

The $2n \times 2n$ matrix $J_{XX}$ can be partitioned into four $n
\times n$ matrices as:
\begin{equation}
  \label{eq:J_UU}
  J_{XX} =
  \begin{pmatrix}
    J_{xx} & J_{xU}\\
    J_{Ux} & J_{UU}
  \end{pmatrix}
\end{equation}
% The value of $\UU_i$ minimising \eqref{eq:J_ic_1} is then given by
% \eqref{eq:U_i} with
% \begin{equation}
%   \label{eq:K_opt}
%   \K = J_{UU}^{-1}J_{Ux}
% \end{equation}

% Given the inverse of $J_{UU}$ in \eqref{eq:K_opt}, it is clearly
% important that $J_{UU}$ is well-conditioned to avoid numerical
% problems. 

% \subsubsection{Quadratic Programme}
% \label{sec:QP}

\begin{lemma}[Constrained optimisation]\label{lem:QP}
 The minimisation of the cost function $J_{ic}$ of Equation \ref{eq:J_ic}
  subject to the constraints \eqref{eq:U_consts} is
  equivalent to the solution of the quadratic programme for the
  optimum value of $\UU_i$:
  \begin{equation}
    \label{eq:QP}
      \min_{U_i} \left \{ \UU_i^T J_{UU} \UU_i + \xx_i^T J_{Ux} \UU_i\right \}
  \end{equation}
$\text{ subject to }   \Gamma \UU_i \le \gamma - \Gamma_x \xx_i$
  where $J_{UU}$ and $J_{Ux}$ are given by
  \eqref{eq:J_UU} and $\Gamma$, $\Gamma_x$ and  $\gamma$ as described
  in Section \ref{sec:constraints}.
\end{lemma}
\begin{proof}
  See \citep{CheGaw06}.
\end{proof}
\paragraph{Remarks.}
\begin{enumerate}
\item This optimisation is dependant on the system state $\xx$ and
  therefore must be accomplished at every intermittent interval
  $\Delta_i$.
\item The computation time is reflected in the time delay $\Delta$.
\item As discussed by \citet{CheGaw06}, the relation between the cost
  function (\ref{eq:QP}) and the LQ cost function (\ref{eq:LQ}) means
  that the solution of the the QP is the same as the LQ solution when
  constraints are not violated.
\end{enumerate}

\section{Example: constrained control of mass-spring system}
\label{sec:ex:cms}

\begin{figure}[htbp]
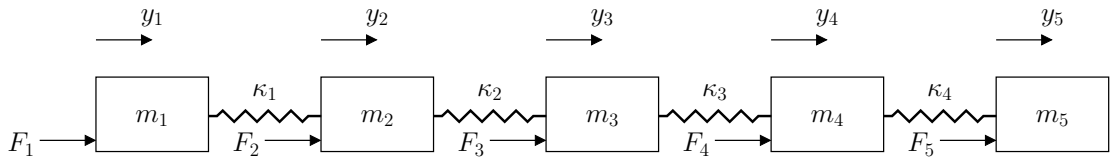

  \centering
  \Fig{CMS}{0.9}
  \caption[Coupled mass-spring system]{Coupled mass-spring system. The
    five masses $m_1$--$m_5$ all have unit mass and the four springs
    $\kappa_1$--$\kappa_5$ all have unit stiffness. The mass positions
    are denoted by $y_1$--$y_5$, velocities by $v_1$--$v_5$ the
    applied forces by $F_1$--$F_5$.}
  \label{fig:cms}
\end{figure}
\begin{figure}[htbp]
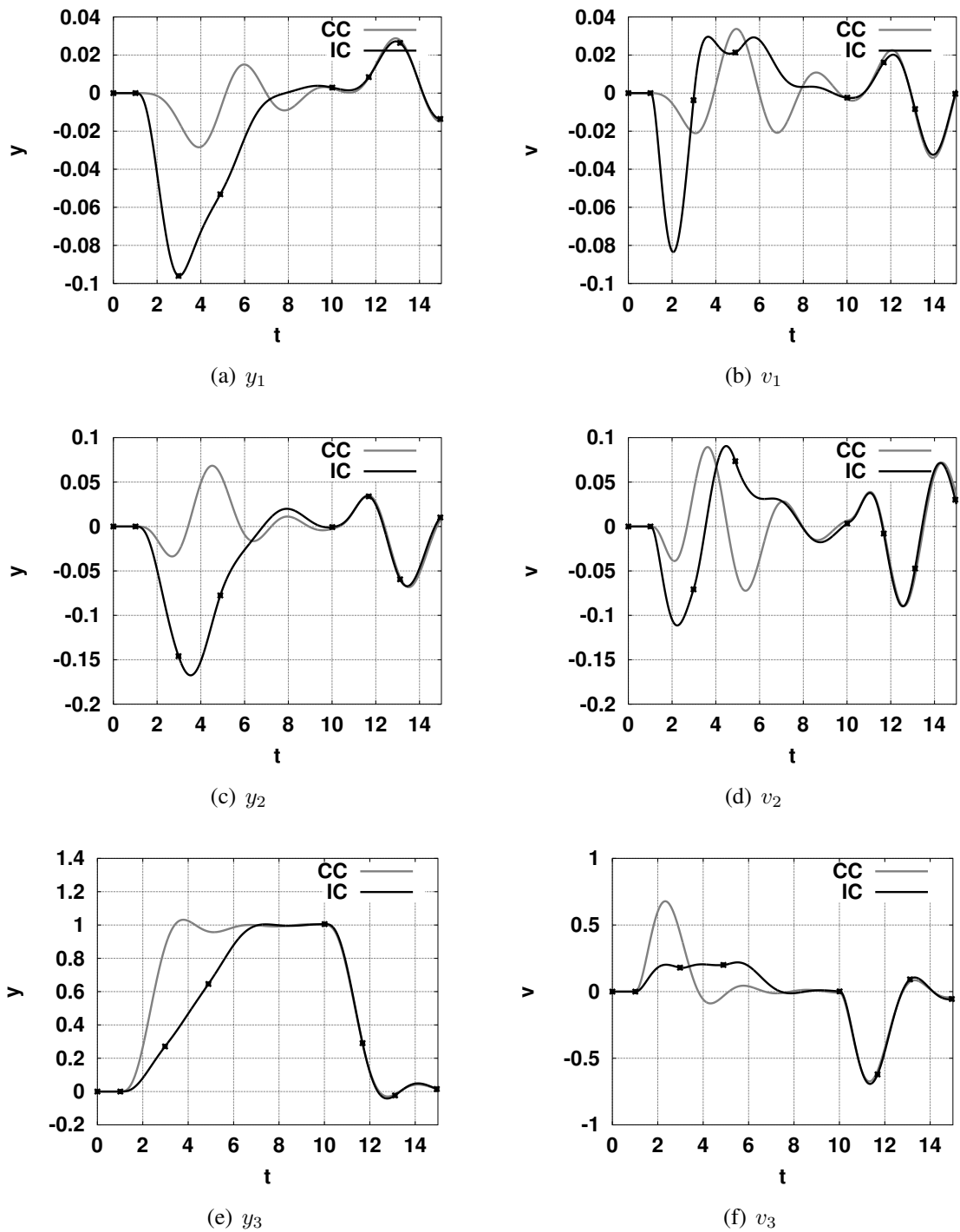

  \centering
  \SubFig{cms_1_y}{$y_1$}{0.45}  
  \SubFig{cms_1_v}{$v_1$}{0.45}\\
  \SubFig{cms_2_y}{$y_2$}{0.45}  
  \SubFig{cms_2_v}{$v_2$}{0.45}\\
  \SubFig{cms_3_y}{$y_3$}{0.45}  
  \SubFig{cms_3_v}{$v_3$}{0.45}
  \caption[Constrained control of mass-spring system]{Constrained
    control of mass-spring system. The left-hand column shows the
    positions of masses 1--3 and the right-hand column the
    corresponding velocities. The grey line corresponds to the
    simulation of the underlying \emph{unconstrained} continuous system and the black
    lines to intermittent control; the $\bullet$ correspond to the
    intermittent sampling times $t_i$.}
  \label{fig:mpc}
\end{figure}

Figure \ref{fig:cms} shows a coupled mass-spring system. The five
masses $m_1$--$m_5$ all have unit mass and the four springs
$\kappa_1$--$\kappa_5$ all have unit stiffness. The mass positions are
denoted by $y_1$--$y_5$, velocities by $v_1$--$v_5$ the applied forces
by $F_1$--$F_5$. In addition it is assumed that the five forces $F_i$ are
generated from the five control signals $u_i$ by simple integrators thus:
\begin{equation}
  \label{eq:F_i}
  \dot{F}_i = u_i, \; i=1 \dots 5 
\end{equation}
This system has fifteen states ($n_x=15$), five inputs ($n_u=5$) and
five outputs ($n_y=5$).

To examine the effect of constraints, consider the case where it is
required that the velocity of the centre mass ($i=3$) is constrained
above by
\begin{equation}
  \label{eq:v_3}
   v_3<0.2
\end{equation}
but unconstrained below. As noted in Section \ref{sec:constraints},
the constraints are at discrete values of intersample time $\tau$. In
this case, fifty points where chosen at $\tau = 0.1, 0.2, \dots
5.0$. The precise choice of these points is not critical.

In addition, the system setpoint is given by
\begin{equation}
  \label{eq:ex:cms:w}
  w_i(t) =
  \begin{cases}
    1 & i=3 \text{ and } 1 \le t <10\\
    0 & \text{ otherwise }
  \end{cases}
\end{equation}

Figure \ref{fig:mpc} shows the results of simulating the coupled
mass-spring system of Figure \ref{fig:cms} with constrained
intermittent control with constraint given by (\ref{eq:v_3}) and
setpoint by (\ref{eq:ex:cms:w}).  Figure \ref{subfig:cms_1_y} shows
the position of the first mass and \ref{subfig:cms_1_v} the
corresponding velocity; Figure \ref{subfig:cms_2_y} shows the position
of the second mass and \ref{subfig:cms_2_v} the corresponding
velocity; Figure \ref{subfig:cms_3_y} shows the position of the third
mass and \ref{subfig:cms_3_v} the corresponding velocity. The fourth
and fifth masses are not shown.
In each case, the corresponding simulation result for the underlying
continuous (unconstrained) simulation is also shown.

Note that on the forward motion of mass three, the velocity (Figure
\ref{subfig:cms_3_v}) is constrained and this is reflected in the
constant slope of the corresponding position (Figure
\ref{subfig:cms_3_y}). However, the backward motion is unconstrained
and closely approximates that corresponding to the unconstrained
continuous controller. The other masses (which have a zero setpoint)
deviate more from zero whilst mass three is  constrained, but are similar to the
unconstrained case when mass three is not constrained.

% \subsection{Power constraints}
% \label{sec:power}
% The vector $\X$, defined in  \eqref{eq:X}, contains the system state and the state
% of the generalised hold; equation \eqref{eq:X(tau)} explicitly gives
% $\X(\tau)$ in terms of the initial value $\X_i$ at time $t_i$. 

% Hence a constraint on the input power at time $\tau$ can be expressed as:
% \begin{align}
% p(\tau) &= \uu^T(\tau)\vv(\tau)
% = \X_i^T\Gu(\tau)\Gv^T(\tau) \X_i \le p_{max} \label{eq:U_consts_pwr} \\ 
% \text{where } \Gu &= \gu \E(\tau)\\
% \text{and } \Gv &= \gv \E(\tau)
% \end{align}

% \begin{lemma}[Power-constrained optimisation ]\label{lem:nlopt}
% The minimisation of the cost function $J_{ic}$ of Equation
% \ref{eq:J_ic} subject to the constraints \eqref{eq:U_consts_pwr} is
% equivalent to the solution of the following quadratically-constrained
% quadratic program (QCQP):
% \begin{equation}
%   \label{eq:nlopt}
%     \min_{U_i} \left \{ \UU_i^T J_{UU} \UU_i + 2\xx_i^T J_{Ux}
%       \UU_i  + \xx_i^T J_{xx}\xx_i \right \}
% \end{equation}
% \end{lemma}
% \begin{proof}
% See \citet{GawWagNeiWan12}.
% \end{proof}
\section{Examples: human standing}
\label{sec:ex:stand}
Human control strategies in the context of quiet standing have been
investigated over many years by a number of authors. Early work, for
example
\citep{Pet02,LakCapLor03,BotCasMorSan05,LorMagLak05},
was based on a single inverted pendulum, single-input model of the
system. More recently, it has been shown
\citep{PinSwiSoe08,GunGriSie09,GunMulBli11,GunMulBli12}
that a multiple segment multiple input model is required to model
unconstrained quiet standing and this clearly has implications for the
corresponding human control system.
Intermittent control has been suggested as the basic algorithm
\citet{GawLorLakGol11}, \citet{GawLeeHalODw13} and
\citet{GawLorGolLak14} and related algorithms have been analysed by
\citet{Ins06,SteIns06}, \citet{AsaTasNomCasMor09} and
\citet{KowGleBro12}.

This section uses a linear three-segment model to illustrate key
features of the contrained multivariable intermittent control
described in Sections \ref{sec:ic} and \ref{sec:constrained}. 
Section \ref{sec:three} describes the three-link model, Section \ref{sec:muscle}
looks at a heirachical approach to muscle-level control, Section
\ref{sec:quiet} looks at an intermittent explanation of quiet standing
and Sections \ref{sec:track} and \ref{sec:dist} discuss tracking and
disturbance rejection respectively.

\subsection{A three-segment model}
\label{sec:three}
This section uses the linearised version of the three link, three
joint model of posture given by \citet{AleFroHor05}. The upper, middle
and lower links are indicated by subscripts $u$, $m$ and $l$ respectively.
The linearised equations correspond to:
\begin{equation}\label{eq:links}
  \M \ddot{\Th} - \G \Th = \N \tor
\end{equation}
where $\Th$ is the vector of link angles given by:
\begin{equation}
  \label{eq:theta}
  \Th =
  \begin{bmatrix}
    \theta_l\\ \theta_m\\ \theta_u
  \end{bmatrix}
\end{equation}
and $\tor$ the vector of joint torques.

the mass matrix $\M$ is given by
\begin{align}
  \M &=b
  \begin{bmatrix}
    m_{ll} &  m_{lm} &  m_{lu}\\
    m_{ml} &  m_{mm} &  m_{mu}\\
    m_{ul} &  m_{um} &  m_{uu}
  \end{bmatrix}\\
  \text{where } 
  m_{ll} &= m_lc_l^2 + (m_m+m_u)l_l^2 + I_l\\
  m_{mm} &= m_mc_m^2 + m_ul_m^2 + I_m\\
  m_{uu} &= m_uc_u^2 + I_u\\
  m_{ml} &= m_{lm} = m_mc_ml_l + m_ul_ll_m\\
  m_{ul} &= m_{lu} = m_uc_ul_l\\
  m_{um} &= m_{mu} = m_uc_ul_m
\end{align}
the gravity matrix $\G$ by
\begin{align}
    \G &= g
  \begin{bmatrix}
    g_{ll} &  0 &  0\\
    0 &  g_{mm} &  0\\
    0 &  0 &  g_{uu}
  \end{bmatrix}\label{eq:gravity_matrix}\\
  \text{where }
  g_{ll} & = m_lc_l + (m_m + m_u)l_l\\
  g_{mm} &= m_mc_m + m_ul_m\\
  g_{uu} &= m_uc_u 
\end{align}
and the input matrix $\N$ by
\begin{equation}
  \label{eq:N}
  \N =
  \begin{bmatrix}
    1 & -1 &  0\\
    0 &  1 & -1\\
    0 &  0 &  1
  \end{bmatrix}
\end{equation}

The joint angles $\phi_l \dots \phi_u$ can be written in terms of the
link angles as:
\begin{align}
  \phi_l &= \theta_l\\
  \phi_m &= \theta_m - \theta_l\\
  \phi_u &= \theta_u - \theta_m
\end{align}
or more compactly as:
\begin{align}
  \Ph &= \N^T \Th\\
  \text{where }
  \Ph &=
  \begin{bmatrix}
    \phi_l\\\phi_m\\\phi_u
  \end{bmatrix}\label{eq:phi}
\end{align}
The values for the link lengths $l$, CoM location $c$, masses $m$ and
moments of inertia (about CoM) were taken from Figure 4.1 and Table
4.1 of \citet{Win09}.

The model of Equation (\ref{eq:links}) can be rewritten as:
\begin{align}
 \ddt{\xx_0} &= \A_0 \xx_0 + \B_0 \tor\\
  \xx_0 &= 
  \begin{bmatrix}
    \dot{\Th} \\ \Th
  \end{bmatrix}
\end{align}
and
\begin{align}
  \A_0 &= 
  \begin{bmatrix}
    \Z_{3\times 3} & -\M^{-1}\G\\
    \I_{3\times 3} & \Z_{3\times 3}
  \end{bmatrix}\\
  \B_0 &= 
  \begin{bmatrix}
     \M^{-1}\N \\ \Z_{3\times 3}
  \end{bmatrix}
\end{align}
The eigenvalues of $\A_0$ are: $\pm 2.62$, $\pm 6.54$ and $\pm
20.4$. The positive eigenvalues indicate that this system is (without
control) unstable. 

More sophisticated models would include nonlinear geometric and
damping effects; but this model provides the basi for illustating the
properties of constrained intermittent control.

\subsection{Muscle model \& hierarchical control}
\label{sec:muscle}
As discussed by \citet{LakCapLor03} and \citet{LorMagLak05}, the
single-inverted pendulum model of balance control uses a muscle model
comprising a spring and a contractile element. In this context, the
effect of the spring is to counteract gravity and thus effectively slow
down the toppling speed on the pendulum. This toppling speed is
directly related to the maximum real part of the system eigenvalues.
This is important as it reduces the control bandwidth necessary to
stabilise the unstable inverted pendulum system
\citep{Ste03,LorGawLak06}.
%
% Similarly, the Hill muscle model (as described in \cite{ShaWis05})
% comprises spring, dampers and a force element.

% \begin{align}
%   \ddt{\tor_k} &= \K_{\phi} \left( \phi_0 - \phi \right ) 
%    &= \K_{\phi} \left( \phi_0 - \N^T \theta \right ) 
% \end{align}
%
The situation is more complicated in the multiple link case as, unlike
the single inverted pendulum case, the joint angles are distinct from
the link angles. From Equation (\ref{eq:gravity_matrix}), the gravity
matrix is diagonal in link space; on the other hand, as the muscle
springs act at the joints, the corresponding stiffness matrix is
diagonal in joint space and therefore cannot cancel the gravity
matrix in all configurations.

The spring model used here is the multi-link extension of the model of
\citet[Figure 1]{LorMagLak05} and is given by:
\begin{align}
  \tor_k &= \K_{\phi} \left( \Ph_0 - \Ph \right )\label{eq:tor_k}\\ 
   \text{where } 
  \K_{\phi} &=
   \begin{bmatrix}
     k_1 & 0 & 0\\
     0 & k_2 & 0\\
     0 & 0 & k_3
   \end{bmatrix}
\text{and }
 \Ph_0 &=
  \begin{bmatrix}
    \phi_{l0}\\\phi_{m0}\\\phi_{u0}
  \end{bmatrix}\label{eq:phi_0}
\end{align}
$\tor_k$ is the vector of spring torques at each joint, $\Ph$ contains
the joint angles (\ref{eq:phi}) and $k_1 \dots k_3$ are the spring
stiffnesses at each joint.  It is convenient to choose the control
signal $\uu$ to be:
\begin{equation}
  \label{eq:uu}
  \uu = \ddt{\Ph_0}
\end{equation}
and thus Equation (\ref{eq:tor_k}) can be rewritten as:
\begin{align}
  \ddt{\tor_k} &= \K_{\phi} \left( \uu - \ddt{\phi} \right )\label{eq:dtor_k}\\ 
   &= \K_{\phi} \left( \uu - \N^T \ddt{\theta} \right )
 \end{align}

 Setting $\tor=\tor_k + \tor_d$ where $\tor_d$ is a disturbance
 torque, the composite system formed from the link dynamics
 (\ref{eq:links}) and the spring dynamics (\ref{eq:dtor_k}) is given
 by Equation (\ref{eq:sys}) where:
\begin{align}
  \xx &= 
  \begin{bmatrix}
    \xx_0 \\ \tor
  \end{bmatrix} =
  \begin{bmatrix}
    \dot{\Th} \\ \Th \\ \tor
  \end{bmatrix}\label{eq:x}\\
  \yy &= \Th\\
  \dd &= \tor_d
\end{align}
and
\begin{align}
  \A  &= 
  \begin{bmatrix}
    \A_0 & \B_0\\
    -\K_{\phi}\N^T & \Z_{3\times 6}
  \end{bmatrix}
= 
  \begin{bmatrix}
    \Z_{3\times 3} & -\M^{-1}\G & \M^{-1}\N\\
    \I_{3\times 3} & \Z_{3\times 3} & \Z_{3\times 3}\\
    -\K_{\phi}\N^T & \Z_{3\times 3} & \Z_{3\times 3}
  \end{bmatrix}\label{eq:A}\\
  \B &= 
  \begin{bmatrix}
    \Z_{3\times 3}\\ \Z_{3\times 3}\\ \K_{\phi}
  \end{bmatrix}\\
  \B_d &= 
  \begin{bmatrix}
    \B_0  \\ \Z_{3\times 3}
  \end{bmatrix}
= 
  \begin{bmatrix}
     \M^{-1}\N \\ \Z_{3\times 3}  \\ \Z_{3\times 3}
  \end{bmatrix}\\
  \C &= 
  \begin{bmatrix}
   \Z_{3\times 3} & \I_{3\times 3}  & \Z_{3\times 3}
  \end{bmatrix}
\end{align}
There are, of course, many other state-space representations with the
same input-output properties, but this particular state space
representation has two useful features: firstly, the velocity control
input of Equation (\ref{eq:uu}) induces an integrator in each of the
three inputs and secondly the state explicitly contains the joint
torque due to the springs. The former feature simplifies control
design in the presence of input disturbances with constant components
and the latter feature allows spring preloading (in anticipation of a
disturbance) to be modelled as a state initial condition. These
features are used in the example of Section \ref{sec:dist}.

\begin{figure}[htbp]
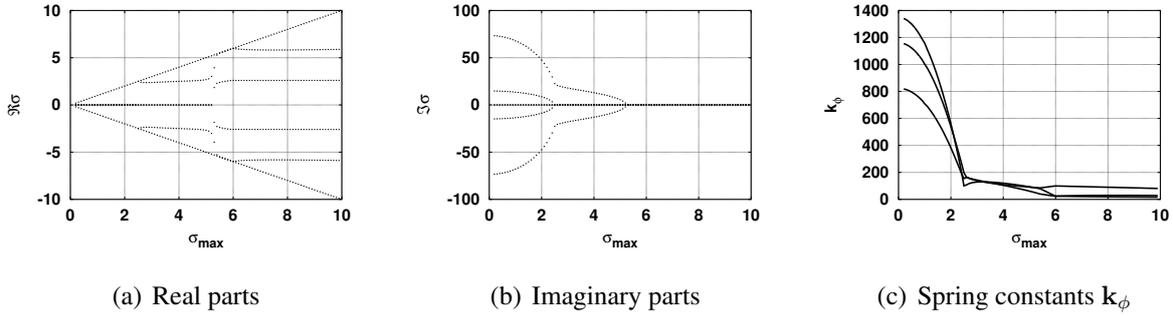

  \centering
  \SubFig{Three_eigs_real}{Real parts}{\FigSize}
  \SubFig{Three_eigs_imag}{Imaginary parts}{\FigSize}
  \SubFig{Three_K}{Spring constants $\kk_\phi$}{\FigSize}
  \caption[Choosing the spring constants.]{Choosing the spring
    constants. (a) The real parts of the non-zero eigenvalues plotted
    against $\sigma_{max}$, the specified maximum real part of all
    eigenvalues resulting from (\ref{eq:sqp_1}) \&
    (\ref{eq:sqp_2}). (b) The imaginary parts corresponding to
    (a). (c) The spring constants $\kk_\phi$. }
  \label{fig:eigs}
\end{figure}
It has been argued \citep{Hog84a} that humans use muscle co-activation
of antagonist muscles to manipulate the passive muscle stiffness and
thus $\K_{\phi}$. As mentioned above, the choice of $\K_{\phi}$ in the
single-link case \citep[Figure 1]{LorMagLak05} directly affects the
toppling speed via the maximum real part of the system
eigenvalues. Hence we argue that such muscle co-activation could be
used to choose the maximum real part of the system
eigenvalues and thus manipulate the required closed-loop control
bandwidth. However,  muscle co-activation requires the flow of energy
and so it makes sense to choose the minimal stiffness consistent with
the required maximum real part of the system
eigenvalues. Defining:
\begin{equation}
  \label{eq:k_phi}
  \kk_{phi} =
  \begin{bmatrix}
    k_1\\k_2\\k_3
  \end{bmatrix}
\end{equation}
this can be expressed mathematically as:
\begin{align}
  \min_{\kk_\phi} ||\K_\phi||\label{eq:sqp_1}\\
  \text{subject to } \max \left[ \Re \sigma_i \right ] < \sigma_{max} \label{eq:sqp_2}
\end{align}
where $\sigma_i$ is the $i$th eigenvalue of $\A$. This is a quadratic
optimisation with non-linear constraints which can be solved by sequential
quadratic programming (SQP) \citep{Fle87}.

In the single-link case, increasing spring stiffness from zero decreases the
value of the positive eigenvalue until it reaches zero, after that
point the two eigenvalues form a complex-conjugate pair with zero real
part. The three link case corresponds to three eigenvalue
pairs. Figure \ref{fig:eigs} shows how the real and imaginary parts
of these  six eigenvalues vary with the constraint $\sigma_{max}$
together with the spring constants $\kk_\phi$. Note that the spring
constants and imaginary parts rise rapidly when the maximum real
eigenvalue is reduced to below about 2.3.

Joint damping can be modelled by the equation:
\begin{align}
  \tor_c &= -\C_{\phi} \ddt{\Ph} = -\C_{\phi} \N^T \ddt{\Th} \label{eq:tor_c}\\ 
   \text{where }  \C_{\phi} &=
   \begin{bmatrix}
     c_1 & 0 & 0\\
     0 & c_2 & 0\\
     0 & 0 & c_3
   \end{bmatrix}
\end{align}
Setting $\tor=\tor_k+\tor_c$ the matrix $\A$ of Equation (\ref{eq:A})
is replaced by:
\begin{align}
   \A  &= \begin{bmatrix}
    -\M^{-1}\N\C_{\phi}\N^T & -\M^{-1}\G & \M^{-1}\N\\
    \I_{3\times 3} & \Z_{3\times 3} & \Z_{3\times 3}\\
    -\K_{\phi}\N^T & \Z_{3\times 3} & \Z_{3\times 3}
  \end{bmatrix}\label{eq:A_1}
\end{align}

\subsection{Quiet standing}
\label{sec:quiet}

\begin{figure}[htbp]
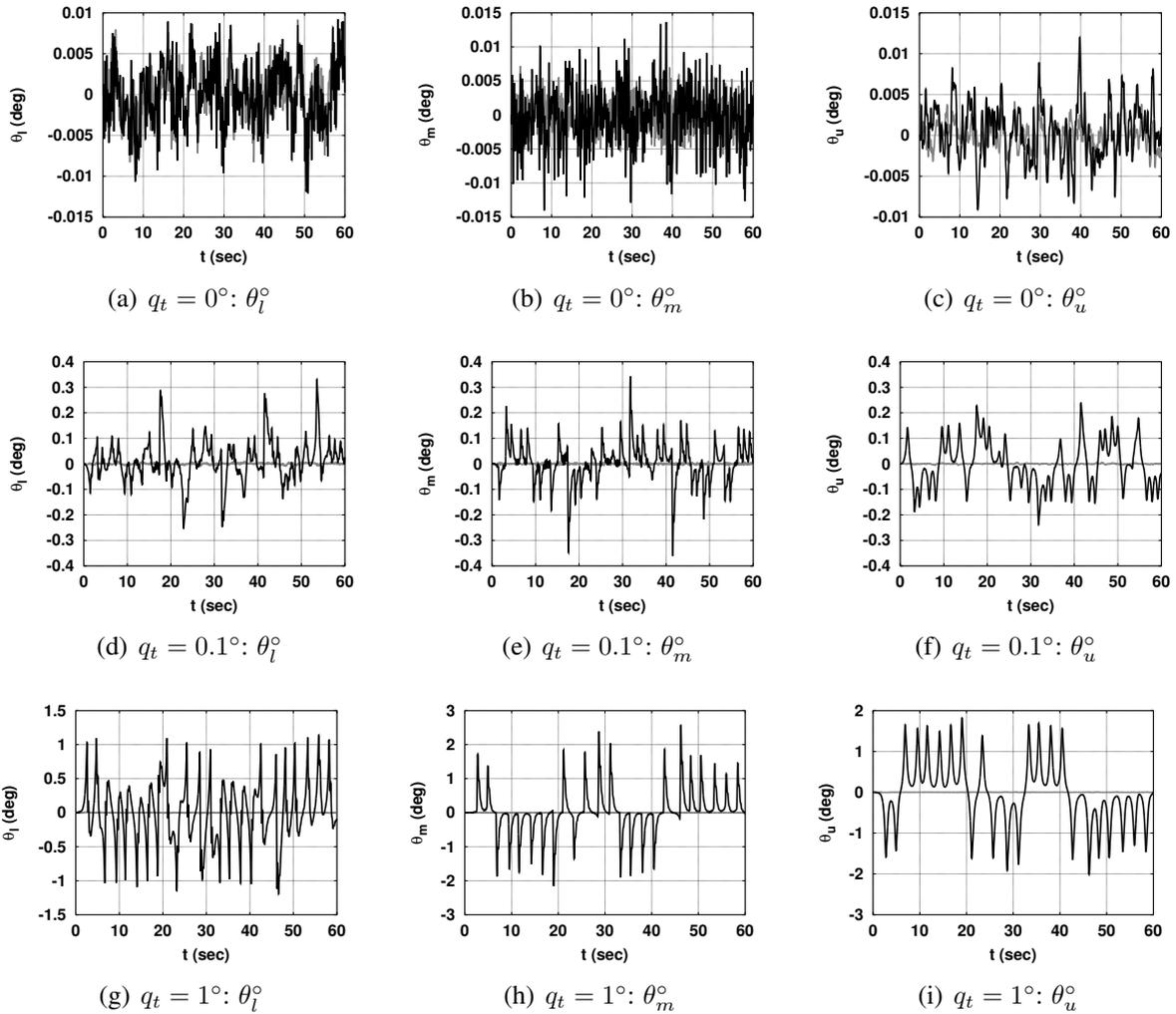

  \centering
  \FIG{0}{0}{0}{}
  \FIG{10}{0.1}{0}{}
  \FIG{100}{1}{0}{}
  \caption[Quiet standing]{Quiet standing with torque disturbance at
    the lower joint. Each row shows the lower ($\Th_l$), middle
    ($\Th_m$) and upper ($\Th_u$) link angles (deg) plotted against
    time $t$ (sec) for a different thresholds $q_t$. Larger thresholds
    give larger, and more regular sway angles.}
  \label{fig:quiet}
\end{figure}

\begin{figure}[htbp]
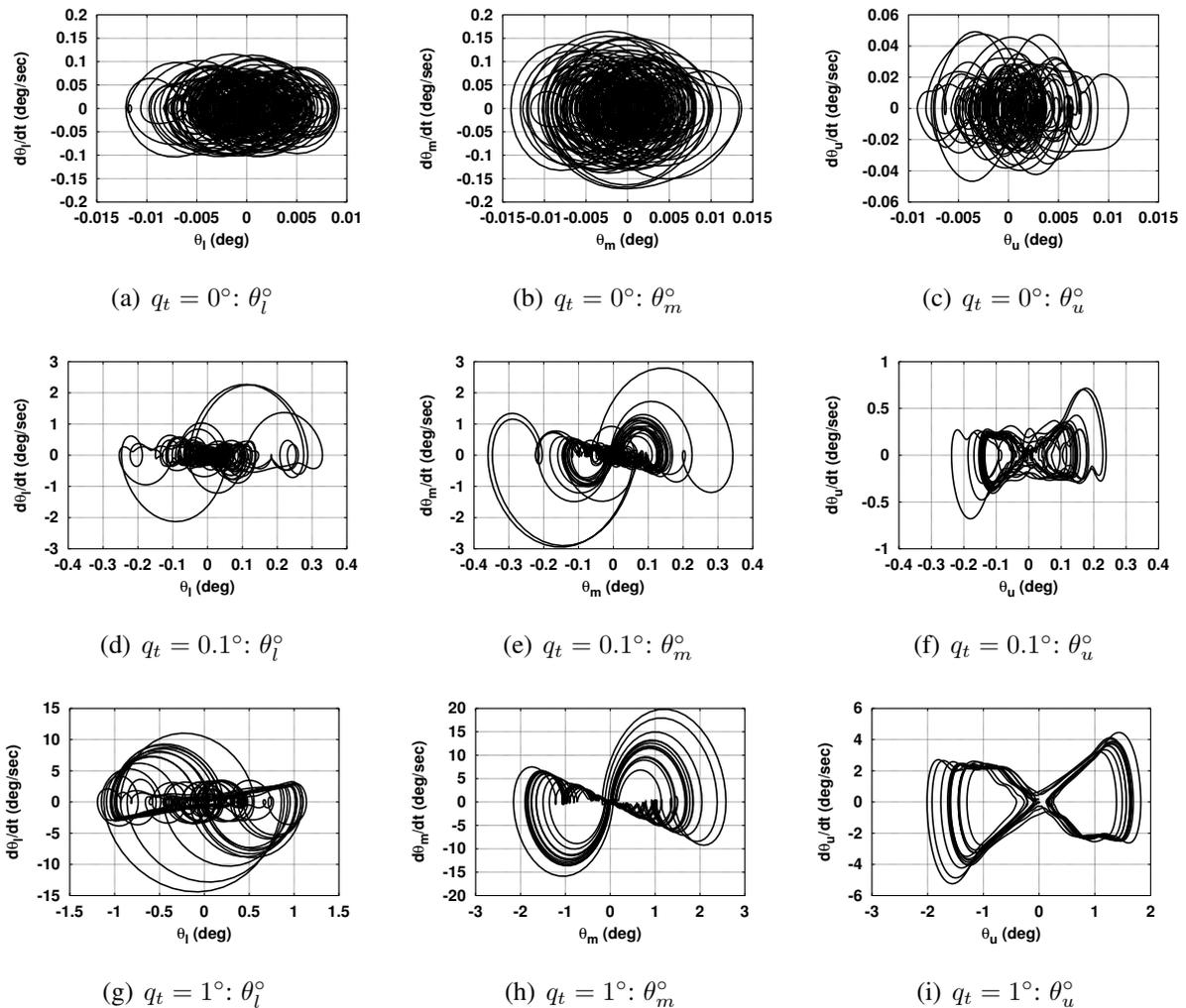

  \centering
  \FIG{0}{0}{0}{_pp}
  \FIG{10}{0.1}{0}{_pp}
  \FIG{100}{1}{0}{_pp}
  \caption[Quiet standing: phase-plane]{Quiet standing:
    phase-plane. Each plot corresponds to that of Figure
    \ref{fig:quiet} but the angular velocity is plotted against
    angle. Again, the increase in sway angle and angular velocity with
  threshold is evident.}
  \label{fig:quiet_pp}
\end{figure}

In the case of quiet standing, there is no setpoint tracking and no
constant disturbance and thus $w=0$ and $\xx_{ss}$ is not
computed. The spring constants were computed as in Section
\ref{sec:muscle} with $\sigma_{max}=3$. The corresponding non-zero
eigenvalues of $\A$ are $\pm 3$, $\pm 2.38$, $\pm j 17.9$
The intermittent controller of Section \ref{sec:ic}, based on the
continuous-time controller of Section \ref{sec:condes} was simulated
using the following control design parameters:
\begin{align}
  \Q_c &=
  \begin{bmatrix}
    q_v\I_{3\times3} & \Z_{3\times3} & \Z_{3\times3}\\
    \Z_{3\times3} & q_p\I_{3\times3} & \Z_{3\times3}\\
    \Z_{3\times3} & \Z_{3\times3} & q_\tor\I_{3\times3}
  \end{bmatrix}\\
  \text{where } q_v &= 0,\; q_p = 1,\; q_\tor = 10\\
  \R_c &= \K_{phi}^2 
\end{align}
The corresponding closed-loop poles are $-3.45 \pm 18.3$, $-3.95 \pm
3.82$, $-2.87 \pm 0.590$, $-5.31$, $-3.28$ and $-2.48$.
The intermittent control parameters (Section \ref{sec:ic}) were time
delay $\del$ and minimum intermittent interval $\Delta_{min}$ were
chosen as:
\begin{align}
  \del &= 0.1\text{s}\label{eq:del}\\
  \Delta_{min} &= 0.25\text{s}\label{eq:Delta_min}
\end{align}
These parameters are used in all of the following simulations.

A multisine disturbance with standard deviation $0.01\text{Nm}$ was
added to the control signal at the lower (ankle) joint. With reference
to Equation (\ref{eq:ED_e}), the threshold was set on the three
segment angles so that the threshold surface (in the 9D state-space)
was defined as:
\begin{align}
  \Th^T\Th &= \xx^T\Q_t\xx = q_t^2  \label{eq:surf}\\
  \text{where } \Q_t\ &=
  \begin{bmatrix}
    \Z_{3\times3} & \Z_{3\times3} & \Z_{3\times3}\\
    \Z_{3\times3} & \I_{3\times3} & \Z_{3\times3}\\
    \Z_{3\times3} & \Z_{3\times3} & \Z_{3\times3}
  \end{bmatrix} 
\end{align}
Three simulations of both IC and CC were performed  with event
threshold $q_t=0^\circ$,
$q_t=0.1^\circ$ and $q_t=1^\circ$ and the resultant link angles $\Th$ are
plotted against time in Figure \ref{fig:quiet}; the black lines show
the IC simulations and the grey lines the CC simulations. 
The three-segment model together with the spring model has 9 states.
Figure \ref{fig:quiet_pp} shows three cross sections though this space
(by plotting segment angular velocity against segment angle) for the
three thresholds.

As expected, the small threshold gives smaller displacements from
vertical; but the disturbance is more apparent.
The large threshold gives largely self-driven behaviour.
% %
This behaviour is discussed in more detail by \citet{GawLorGolLak14}.

\subsection{Tracking}
\label{sec:track}
\begin{figure}[htbp]
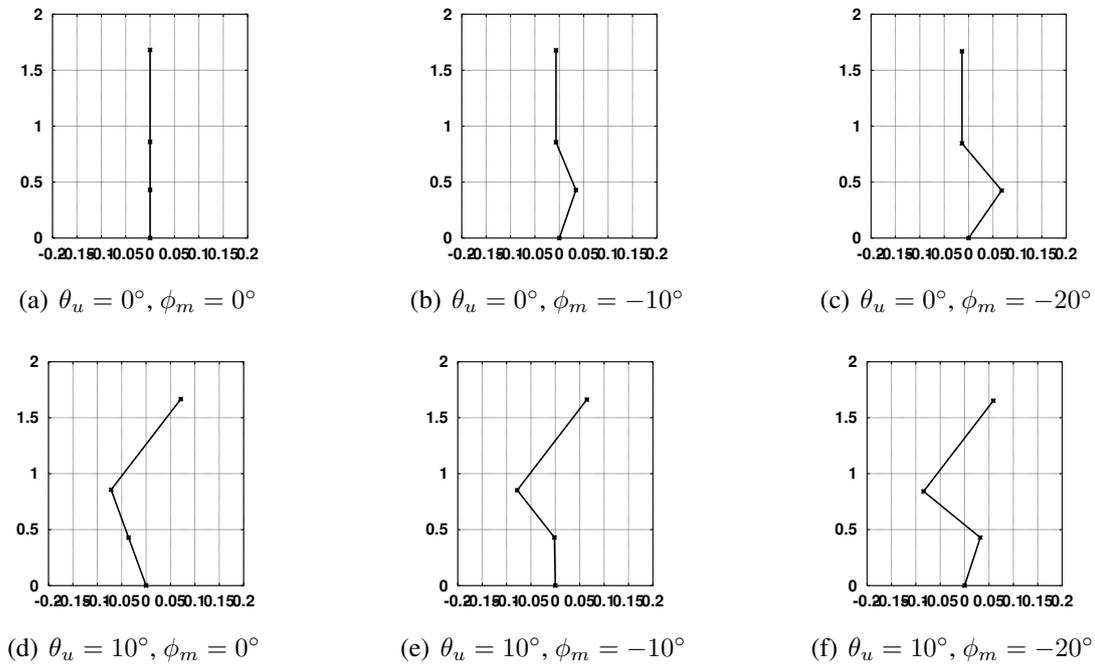

  \centering
  \SubFig{links_0_0}{$\theta_u=0^\circ$, $\phi_m=0^\circ$}{\FigSize}
  \SubFig{links_0_10}{$\theta_u=0^\circ$, $\phi_m=-10^\circ$}{\FigSize}
  \SubFig{links_0_20}{$\theta_u=0^\circ$, $\phi_m=-20^\circ$}{\FigSize}
  \SubFig{links_10_0}{$\theta_u=10^\circ$, $\phi_m=0^\circ$}{\FigSize}
  \SubFig{links_10_10}{$\theta_u=10^\circ$, $\phi_m=-10^\circ$}{\FigSize}
  \SubFig{links_10_20}{$\theta_u=10^\circ$, $\phi_m=-20^\circ$}{\FigSize}
  \caption[Equilibria: Link Configuration]{Equilibria: Link
    Configuration. (a),(b),(c). In each configuration, the upper link is set at
    $\theta_u=0^\circ$ and the posture is balanced (no ankle torque); the knee
    angle is set to three possible values. (d),(e),(f) as (a),(b),(c)
    but $\theta_u=10^\circ$ }
  \label{fig:equi_link}
\end{figure}

\begin{figure}[htbp]
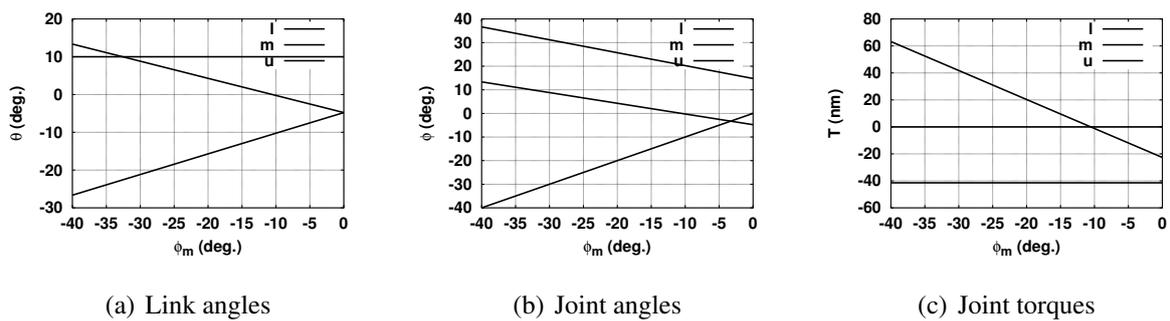

  \centering
  \SubFig{ss_link}{Link angles}{\FigSize}
  \SubFig{ss_joint}{Joint angles}{\FigSize}
  \SubFig{ss_torque}{Joint torques}{\FigSize}
  \caption[Equilibria: angles and torques.]{Equilibria: angles and
    torques. (a) The equilibrium link angles $\Th$ are plotted against
    the fixed knee angle $\phi_m$ with balanced posture and the upper
    link at an angle of $\theta_u=10^\circ$. (b) \& (c) As (a) but
    with joint angles $\Ph$ and joint torques $\tor$
    respectively. Note that the ankle torque $T_l$ is zero (balanced
    posture) and the waist torque $T_u$ balances the fixed
    $\theta_u=10^\circ$.}
  \label{fig:equi_angle}
\end{figure}

\begin{figure}[htbp]
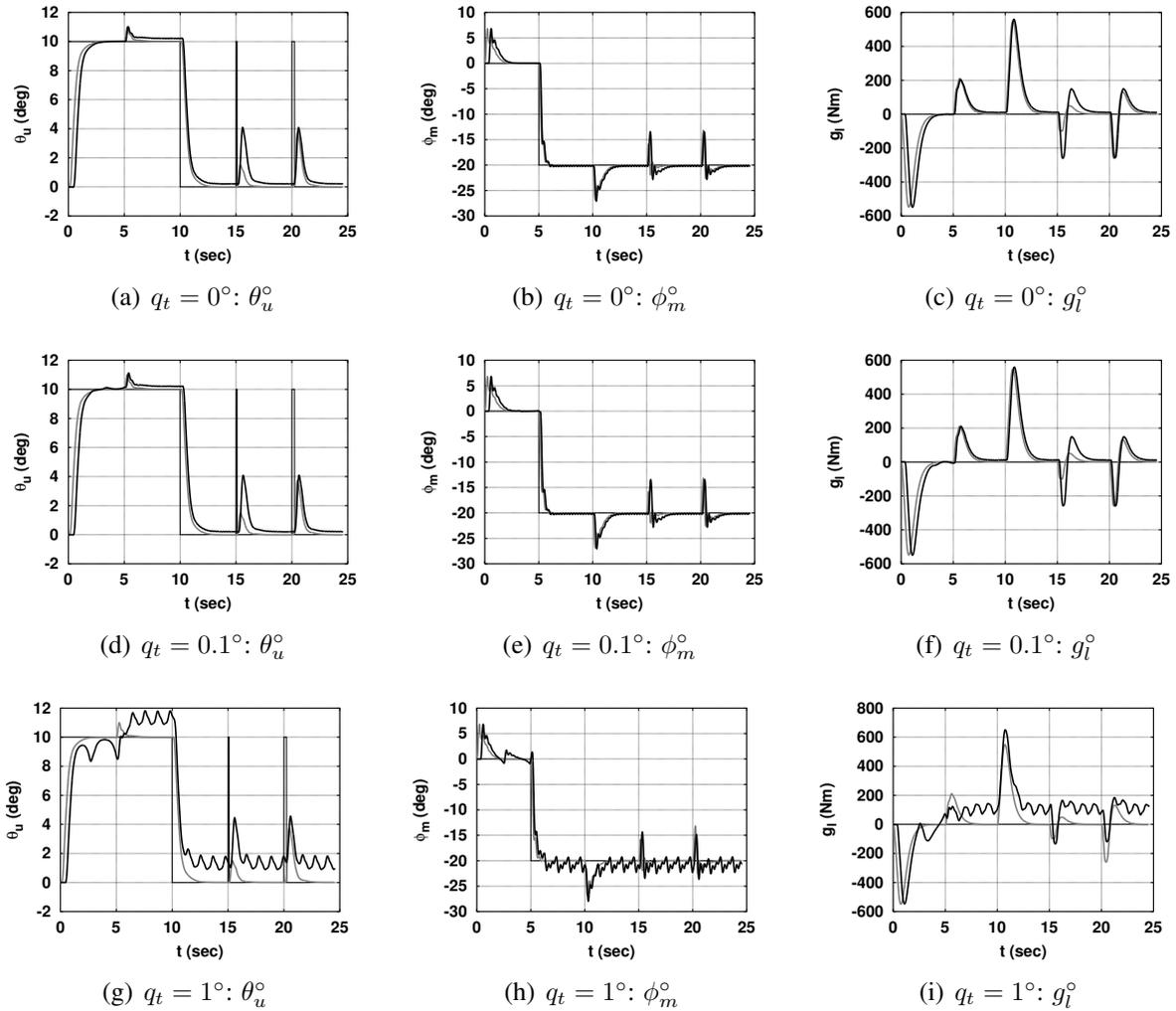

  \centering
  \FIGw{0}{0}{1}{1}
  \FIGw{10}{0.1}{1}{1}
  \FIGw{100}{1}{1}{1}
  \caption[Tracking: controlled knee joint]{Tracking: controlled knee
    joint. The equilibrium design (Section \ref{sec:steady-state}) sets
    the upper link angle $\Th_u$ to $10^\circ$ for $t<10$ and to to
    $0^\circ$ for $10 \le t<15$ and sets the gravity torque at the
    ankle joint to zero; it also sets the knee angle $\phi_m$ to zero
    for $t<5$ and to $-20^\circ$ for $t \ge 5$. At time $t=15$ a pulse
    of width 0.1s is applied to the upper link angle setpoint and a a
    pulse of width 0.25s is applied at time $t=20$. Note that the
    intermittent control response is similar in each case: this
    refractory behaviour is due to event-driven control with a minimum
    intermittent interval $\Delta_{min}=0.25$ (\ref{eq:Delta_min}).}
  \label{fig:track_k}
\end{figure}

\begin{figure}[htbp]
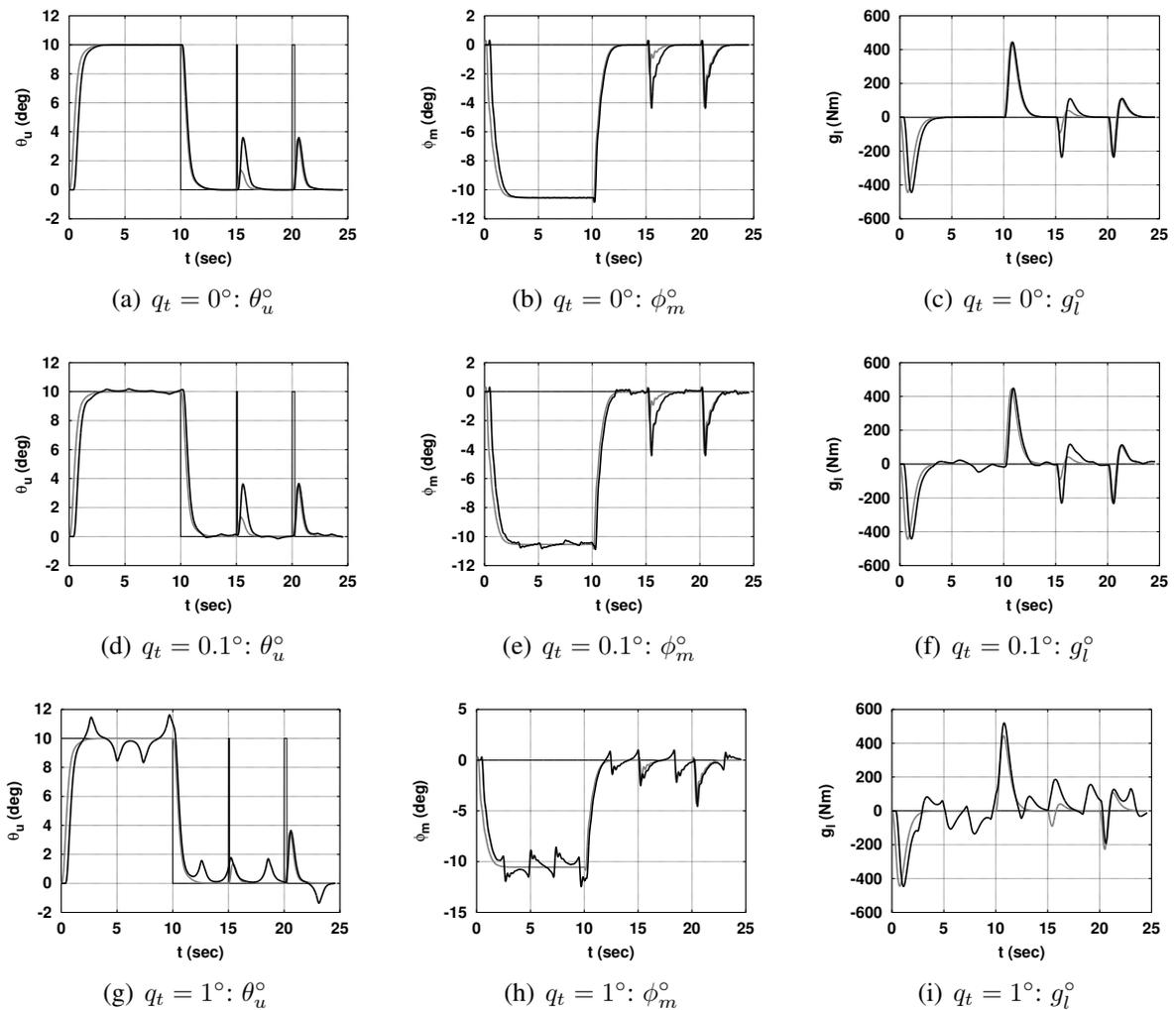

  \centering
  \FIGw{0}{0}{1}{0}
  \FIGw{10}{0.1}{1}{0}
  \FIGw{100}{1}{1}{0}
  \caption[Tracking: free knee joint]{Tracking: free knee joint. This
    Figure corresponds to Figure \ref{fig:track_k} except that the knee
  joint angle $\phi_m$ is not constrained.}
  \label{fig:track}
\end{figure}

As discussed in Section \ref{sec:steady-state}, the equilibrium state
$\xx_{ss}$ has to be designed for tracking purposes. As there are
three inputs, it is possible to satisfy up to three steady-state
conditions. Three possible steady-state conditions are:
\begin{enumerate}
\item The upper link should follow a setpoint:
  \begin{equation}
    \label{eq:th_ss}
    \Th_u = w_u
  \end{equation}
\item The the component of ankle torque due to gravity should be zero:
  \begin{equation}
    \label{eq:tau_ss}
    \tor_1 =
    \begin{bmatrix}
      1&0&0
    \end{bmatrix}
    \N^{-1}\G\Th = 0
  \end{equation}
\item The knee angle should follow a set point:
  \begin{equation}
    \label{eq:knee}
    \phi_m = \theta_m - \theta_l = w_m 
  \end{equation}
\end{enumerate}
These conditions correspond to:
\begin{align}
  \C_{ss} &=
  \begin{bmatrix}
    0 & 0 & 0 &  0 & 0 & 1 &  0 & 0 & 0\\
    0 & 0 & 0 &  27.14627 &  22.51155 &  23.74238  &  0 & 0 & 0\\
    0 & 0 & 0 &  -1 & 1 & 0 &  0 & 0 & 0
  \end{bmatrix}\\
  \yy_{ss} &= \I_{3 \times 3}\\
  \ww(t) &=
  \begin{bmatrix}
    w_u(t) \\ 0 \\ w_m(t) \label{eq:w}
  \end{bmatrix}
\end{align}
This choice is examined in Figures \ref{fig:equi_link} and
\ref{fig:equi_angle} by choosing the knee angle
$\phi_m=\theta_m-\theta_l$. Figure \ref{fig:equi_angle} shows how the
link and joint angles, and the corresponding torques, vary with
$\phi_m$. Figure \ref{fig:equi_link} shows a picture of the three
links for three values of $\phi_m$. In each case, note that the upper
link and the corresponding hip torque remain constant due to the first
condition and that each configuration appears balanced due to
condition 2.

The simulations shown in Figure \ref{fig:track_k} shows the tracking
of a setpoint $\ww(t)$ (\ref{eq:w}) using the three conditions for
determining the steady-state. In this example, the individual setpoint
components of Equation (\ref{eq:w}) are:
\begin{align}
   w_u(t) &=
   \begin{cases}
     10^\circ & 0<t\le 10\\
     0^\circ  & 10 < t \le 15\\
     0^\circ  & 10 < t \le 15\\
     10^\circ  & 15 < t \le 15.1\\
     0^\circ  & 15.1 < t \le 20\\
     10^\circ  & 20 < t \le 20.25\\
     0^\circ  & t> 20.25
   \end{cases}\label{eq:w_u}\\
  w_m(t) &=
   \begin{cases}
     0^\circ & 0<t\le 5\\
     -20^\circ  & t>5
   \end{cases}
 \end{align}
 
As a further example, only the first two conditions for
determining the steady-state are used; the knee is not included. 
These conditions correspond to:
\begin{align}
  \C_{ss} &=
  \begin{bmatrix}
    0 & 0 & 0 &  0 & 0 & 1 &  0 & 0 & 0\\
    0 & 0 & 0 &  27.14627 &  22.51155 &  23.74238  &  0 & 0 & 0
  \end{bmatrix}\\
  \yy_{ss} &= \I_{2 \times 2}\\
  \ww(t) &=
  \begin{bmatrix}
    w_u(t) \\ 0\label{eq:w2}
  \end{bmatrix}
\end{align}
The under-determined equation (\ref{eq:x_ss}) is solved using the pseudo
inverse.
% %
% \tbd{sec:track}{The octave $\backslash$ operator is used. Exactly what is going
%   on here needs to be investigated further.}
% %
The simulations shown in Figure \ref{fig:track} shows the tracking
of a setpoint $\ww(t)$ (\ref{eq:w}) using the first two conditions for
determining the steady-state. In this example, the individual setpoint
component $w_u$ is given by (\ref{eq:w_u}).
Comparing Figure \ref{fig:track}(b), (e) \& (h) with Figure
\ref{fig:track_k}(b), (e) \& (h), it can be seen that the knee angle
is no longer explicitly controlled.

\subsection{Disturbance rejection}
\label{sec:dist}
\begin{figure}[htbp]
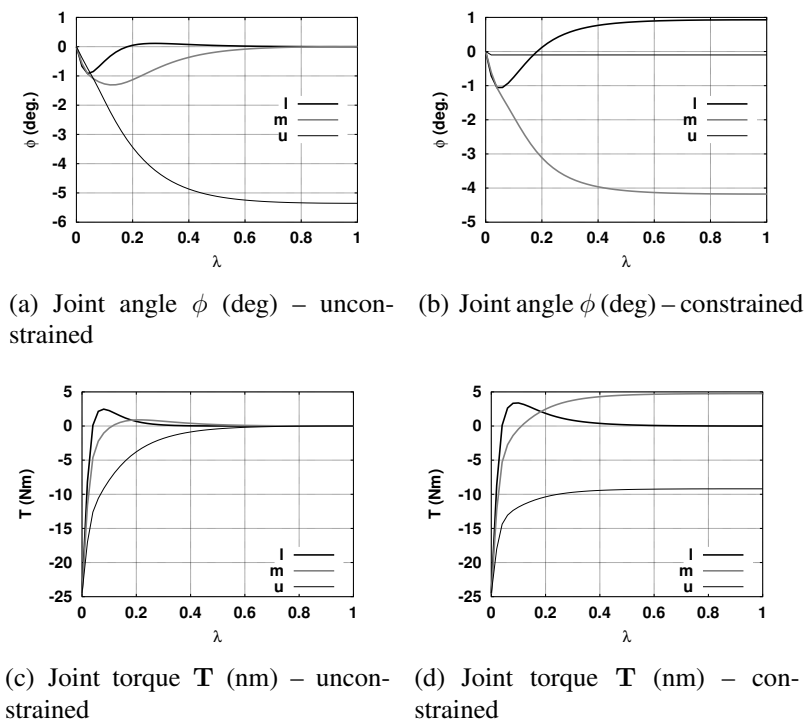

  \centering
  %%\SubFig{lift_ss_link_0}{Link angle $\theta$ (deg)}{\FigSize}
  %%\SubFig{lift_ss_link_1}{Link angle $\theta$ (deg)}{\FigSize}
  \SubFig{lift_ss_joint_0}{Joint angle $\phi$ (deg) -- unconstrained}{\FigSize}
  \SubFig{lift_ss_joint_1}{Joint angle $\phi$ (deg) -- constrained}{\FigSize}\\
  \SubFig{lift_ss_torque_0}{Joint torque  $\tor$ (nm) -- unconstrained}{\FigSize}
  \SubFig{lift_ss_torque_1}{Joint torque  $\tor$ (nm) -- constrained}{\FigSize}
  \caption[Equilibria]{Equilibria. For a constant disturbance torque
    $\tor_d$ acting on the upper link, the plots show how link angle
    and joint torque vary with the weighting factor
    $\lambda$. $\lambda=0$ gives an upright posture (zero link and
    joint angles -- Figure \ref{subfig:lift_links_0_0}) and
    $\lambda=1$ gives a balanced posture (zero joint torques -- Figure
    \ref{subfig:lift_links_10_0}). $\lambda=0.5$ gives an intermediate
    posture (Figure \ref{subfig:lift_links_5_0}).  The left column is
    the unconstrained case, the right column is the constrained case
    where hip angle $\phi_u>-0.1^\circ$ and knee angle $\phi_m<0$; the
    former constraint becomes active as $\lambda$ increases and the
    knee and hip joint torques are no longer zero at $\lambda=1$.}
  \label{fig:equilibria}
\end{figure}

\begin{figure}[htbp]
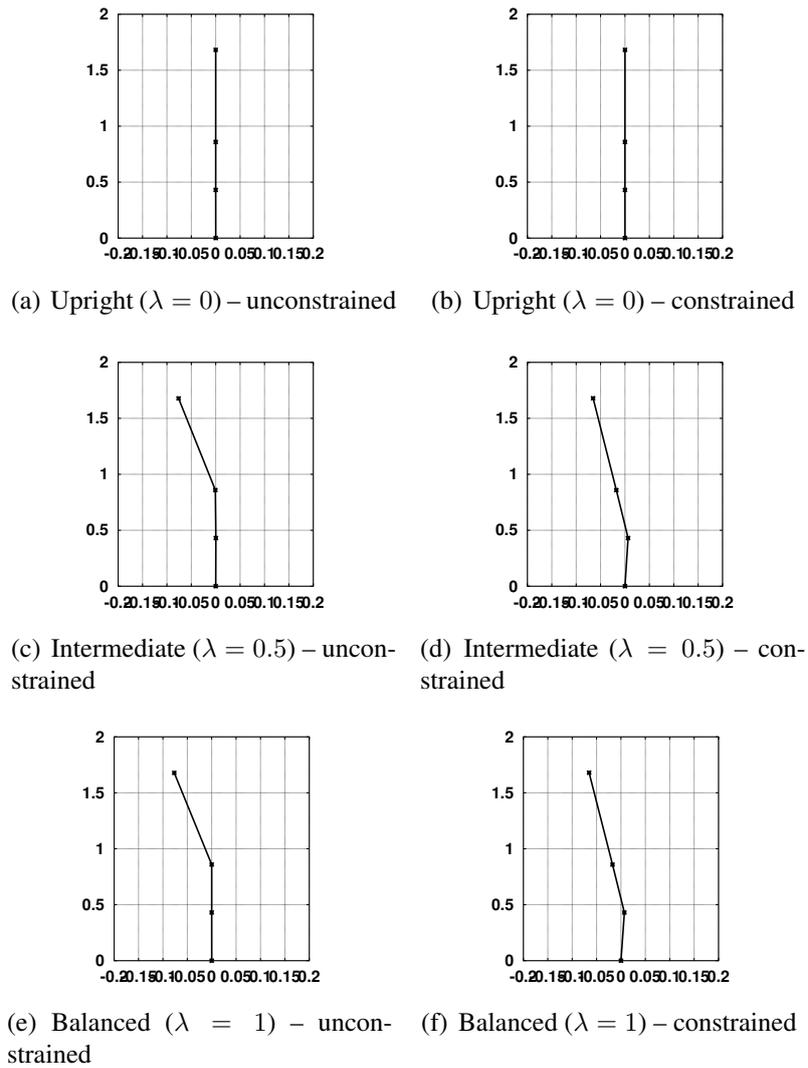

  \centering
    \SubFig{lift_links_0_0}{Upright ($\lambda=0$) -- unconstrained}{\FigSize}
    \SubFig{lift_links_0_1}{Upright ($\lambda=0$) -- constrained}{\FigSize}\\
    \SubFig{lift_links_5_0}{Intermediate ($\lambda=0.5$) -- unconstrained}{\FigSize}
    \SubFig{lift_links_5_1}{Intermediate ($\lambda=0.5$) -- constrained}{\FigSize}\\
    \SubFig{lift_links_10_0}{Balanced ($\lambda=1$) -- unconstrained}{\FigSize}
    \SubFig{lift_links_10_1}{Balanced ($\lambda=1$) -- constrained}{\FigSize}
    \caption[Equilibria: link configurations]{Equilibria: link
      configurations. (a),(c),(e) unconstrained; (b),(d),(f)
      constrained where hip angle $\phi_u>-0.1^\circ$ and knee angle
      $\phi_m<0$.}
  \label{fig:link_config}
\end{figure}

\begin{figure}[htbp]
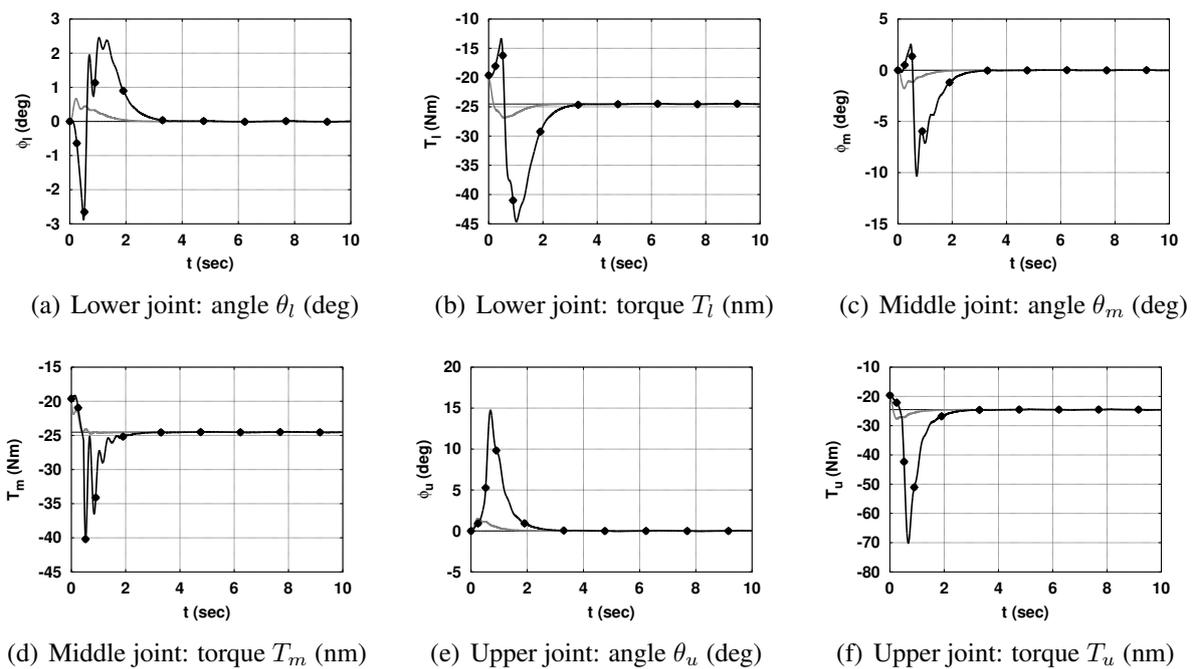

  \centering
  \PoleFig{1}{0}{0}{1}{80}
  \caption[Simulation: upright posture]{Pole-lifting simulation --
    constrained: upright posture ($\lambda=0$). The spring preload
    $\kappa$ (\ref{eq:preload}) is
    80\%. Note that the steady state link angles are zero and the
    steady-state torques are all $-T_d$ to balance $T_d=g m_p
    l_p = 24.5$Nm (\ref{eq:tau_d}). The dots correspond to the sample
    times $t_i$. The intervals $\Delta_i$ (\ref{eq:Delta_i}) are
    irregular and greater than the minimum $\Delta_{min}$.  }
  \label{fig:upright_con}
\end{figure}

\begin{figure}[htbp]
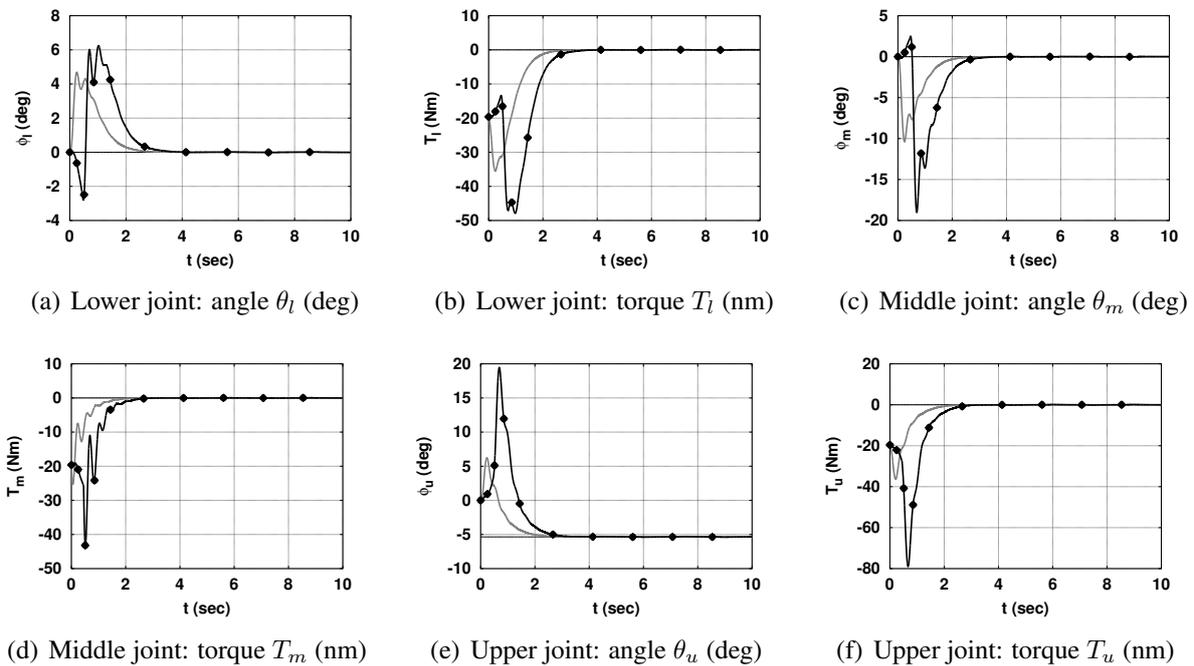

  \centering
  \PoleFig{1}{100}{0}{0}{80}
  \caption[Simulation: balanced posture -- unconstrained]{Pole-lifting
    simulation -- unconstrained steady-state: balanced posture
    ($\lambda=1$). The spring preload $\kappa$ (\ref{eq:preload}) is
    80\%. The steady state joint torques are zero.}
  \label{fig:balanced_uncon}
\end{figure}

\begin{figure}[htbp]
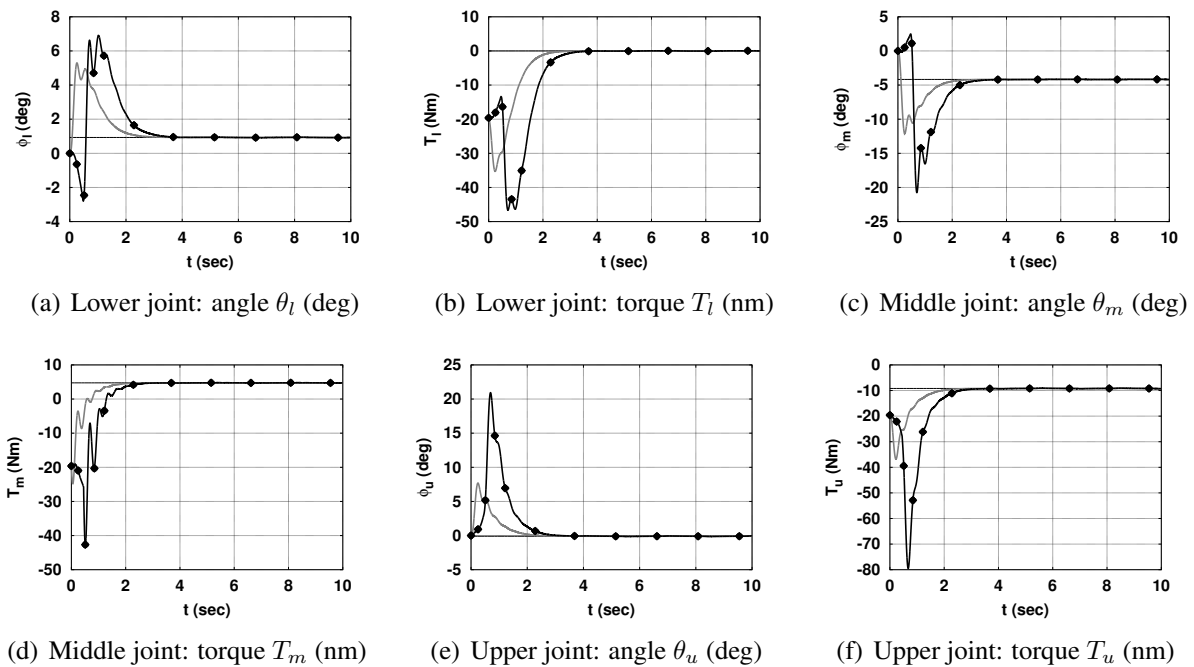

  \centering
  \PoleFig{1}{100}{0}{1}{80}
  \caption[Simulation: balanced posture-- constrained]{Pole-lifting
    simulation -- constrained steady-state: balanced posture ($\lambda=1$). The
    spring preload $\kappa$ (\ref{eq:preload}) is 80\%. Due to the
    constraints, the steady state joint torques $\tor_u$ and $\tor_m$
    are not zero, but the upper (hip) joint is now constrained.}
  \label{fig:balanced_con}
\end{figure}

% \begin{figure}[htbp]
%   \centering
%   \PoleFig{0}{100}{0}{1}{80}
%   \caption[Simulation: balanced posture-- constrained, zero
%   threshold]{Pole-lifting simulation -- constrained, zero
%     threshold. This is the same as Figure \ref{fig:balanced_con}
%     except that the threshold is $q_t=0$ instead of $q_t=0.01$; this
%     results in regular sampling with interval $\Delta_{i} =
%     \Delta_{min}=0.25$. }
%   \label{fig:balanced_con_0}
% \end{figure}

% \begin{figure}[htbp]
%   \centering
%   \PoleFig{10}{100}{0}{1}{80}
%   \caption[Simulation: balanced posture-- constrained, large
%   threshold]{Pole-lifting simulation -- constrained, large
%     threshold. This is the same as Figure \ref{fig:balanced_con} except
%     that the threshold is $q_t=0.1$ instead of $q_t=0.01$ }
%   \label{fig:balanced_con_100}
% \end{figure}

% \begin{figure}[htbp]
%   \centering
%   \PoleFig{1}{100}{1}{1}{80}
%   \caption[Simulation: balanced posture]{Pole-lifting simulation --
%     constrained with observer. This is the same as Figure
%     \ref{fig:balanced_con} except that an observer is used.  }
%   \label{fig:balanced_con_o}
% \end{figure}

\begin{figure}[htbp]
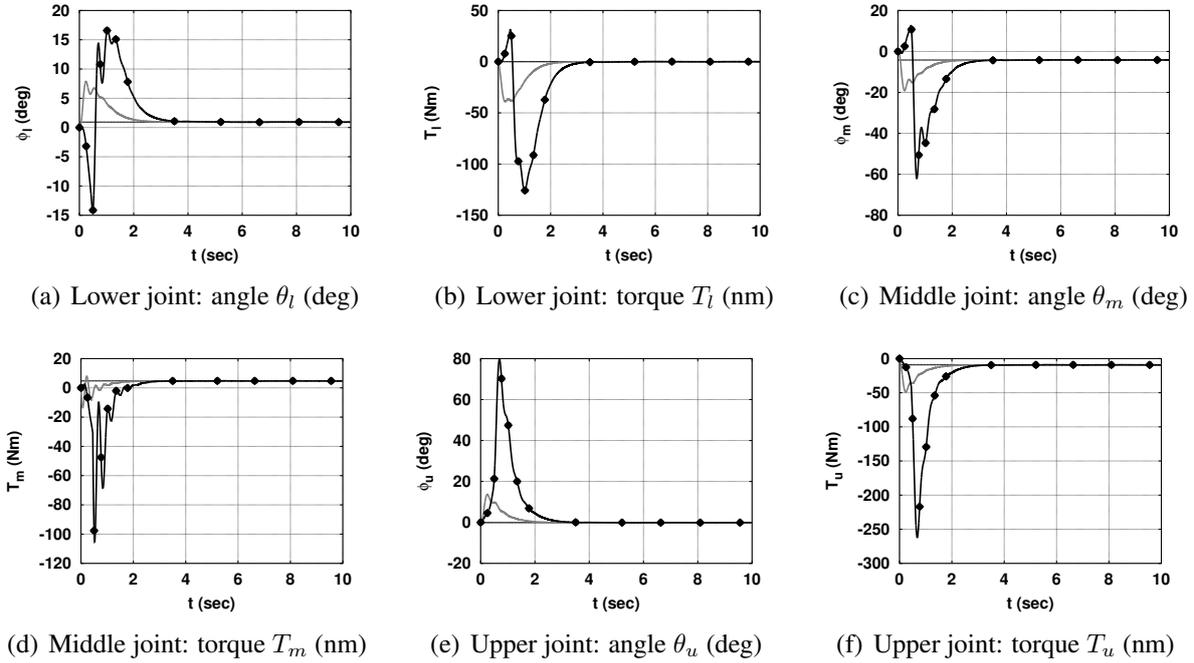

  \centering
  \PoleFig{1}{100}{0}{1}{0}
  \caption[Simulation: balanced posture]{Pole-lifting simulation --
    constrained steady-state but with 0\% preload. This is the same as Figure
    \ref{fig:balanced_con} except that the spring preload is 0\%.}
  \label{fig:balanced_con_pre_0}
\end{figure}

% \begin{figure}[htbp]
%   \centering
%   \PoleFig{1}{100}{0}{1}{50}
%   \caption[Simulation: balanced posture]{Pole-lifting simulation --
%     constrained with 50\% preload. This is the same as Figure
%     \ref{fig:balanced_con} except that the spring preload is 50\%.}
%   \label{fig:balanced_con_pre_50}
% \end{figure}

% \begin{figure}[htbp]
%   \centering
%   \PoleFig{1}{100}{0}{1}{100}
%   \caption[Simulation: balanced posture]{Pole-lifting simulation --
%     constrained with 100\% preload. This is the same as Figure
%     \ref{fig:balanced_con} except that the spring preload is 100\%.}
%   \label{fig:balanced_con_pre_100}
% \end{figure}

Detailed modelling of a human lifting and holding a heavy pole would
require complicated dynamical equations. This section looks at a
simple approximation to the case where a heavy pole of mass $m_p$ is
held at a fixed distance $l_p$ to the the body. In particular, the
effect is modelled by
\begin{enumerate}
\item adding a torque disturbance $T_d$ to the upper link where
  \begin{equation}
    \label{eq:tau_d}
    T_d = g m_p l_p
  \end{equation}
\item adding a mass $m_p$ to the upper link.
\end{enumerate}

In terms of the system Equation (\ref{eq:sys}) and the three link
model of Equation (\ref{eq:links}), the disturbance $\dd$ is given by
\begin{align}
  \dd &= N^{-1}
  \begin{bmatrix}
    0\\0\\T_d
  \end{bmatrix}\\
  &=   
  \begin{bmatrix}
    1\\1\\1
  \end{bmatrix} T_d
\end{align}

As discussed in Section \ref{sec:muscle}, it is possible to preload
the joint spring to give an initial torque. In this context, this is
done by initialising the system state $\xx$ of Equation (\ref{eq:x}) as:
\begin{equation}
  \label{eq:preload}
  \xx(0) =   
  \begin{bmatrix}
    \dot{\Th}(0) \\ \Th(0) \\ \tor(0)
  \end{bmatrix} =
  \begin{bmatrix}
    \Z_{3\times1} \\ \Z_{3\times1} \\ \kappa \dd
  \end{bmatrix}
\end{equation}
$\kappa$ will be referred to as the \emph{spring preload} and will be
expressed as a percentage: thus $\kappa = 0.8$ will be referred to as
80\% preload.

There are many postures appropriate to this situation, two of which are:
\begin{description}
\item[upright]: all joint \emph{angles} are zero and the pole is balanced by
  appropriate joint torques (Figures \ref{subfig:lift_links_0_0} \& \ref{subfig:lift_links_0_1});
\item[balanced]: all joint \emph{torques} are zero and the pole is
  balanced by appropriate joint (and thus link) angles (Figures
  \ref{subfig:lift_links_10_0} \& \ref{subfig:lift_links_10_1}).
\end{description}
In terms of Equation (\ref{eq:y_ss}), the upright posture is specified
by choosing:
\begin{align}
  \C_{ss} &=
  \begin{bmatrix}
    \Z_{3\times3} & \I_{3\times3} & \Z_{3\times3}
  \end{bmatrix}\\
  \yy_{ss} &= \Z_{3\times1}
\end{align}
and the balanced posture is specified
by choosing:
\begin{align}
  \C_{ss} &=
  \begin{bmatrix}
    \Z_{3\times3} & \Z_{3\times3} & \I_{3\times3}
  \end{bmatrix}\\
  \yy_{ss} &= \Z_{3\times1}
\end{align}
A combination of both can be specified by  choosing:
% \begin{align}
%   \C_{ss} &=
%   \begin{bmatrix}
%    \Z_{3\times3} & (1-\lambda)\I_{3\times3} & \Z_{3\times3}\\
%     \Z_{3\times3} & \Z_{3\times3} & \frac{\lambda}{\tor_d} \I_{3\times3}
%   \end{bmatrix}\label{eqn_C_ss_lambda}\\
%   \yy_{ss} &= \Z_{6\times1}\label{eqn_y_ss_lambda}
% \end{align}
\begin{align}
  \C_{ss} &=
  \begin{bmatrix}
    \Z_{3\times3} & \I_{3\times3} & \Z_{3\times3}\\
    \Z_{3\times3} & \Z_{3\times3} & \I_{3\times3}
  \end{bmatrix}\label{eqn_C_ss_lambda}\\
  \yy_{ss} &= \Z_{6\times1}\label{eqn_y_ss_lambda}\\
  \Q_{ss} &=
  \begin{bmatrix}
    (1-\lambda) \I_{3\times3} & \Z_{3\times3}\\
    \Z_{3\times3} & \frac{\lambda}{\tor_d} \I_{3\times3}
  \end{bmatrix}\label{eqn_Q_ss_lambda}
\end{align}
The parameter $0 \le\lambda\le 1$ weights the two postures and
division by $\tor_d$ renders the equations dimensionless. 

When $\C_{ss}$ is given by (\ref{eqn_C_ss_lambda}), $n_{ss}=6$. As
$n_u=3$, $n_{ss}>n_u$ and so, as discussed in Section
\ref{sec:steady-state}, the set of equations (\ref{eq:x_ss}) is over
determined and the approach of Section \ref{sec:constrained} is used.
Two situations are examined: unconstrained and constrained with hip
angle and knee angle subject to the inequality constraints:
\begin{align}
  \phi_u&>-0.1^\circ\\
  \phi_m&<0
\end{align}
In each case, the equality constraint (\ref{eq:equality}) is imposed.
Figure \ref{fig:equilibria} shows how the equilibrium joint angle $\Ph$
and torque $\tor$ vary with $\lambda$ for the two cases.
As illustrated in Figure \ref{fig:link_config}
the two extreme cases $\lambda=0$ and  $\lambda=1$ correspond to the
upright and balanced postures; other values give intermediate postures.

Figures \ref{fig:upright_con} and \ref{fig:balanced_con} show simulation
results for the two extreme cases of $\lambda$ for the unconstrained
case with $m_p=5$kg and $l_p=0.5$m. In
each case, the initial link angles are all zero
($\Th_l=\Th_m=\Th_u=0$) and the disturbance torque $\tor_d=g$ is applied
at $t=0$. Apart from the equilibrium vector $\xx_{ss}$, the control
parameters are the same in each case.

In the case of Figure \ref{fig:upright_con}, the steady-state torques are
$\tor_l=\tor_m=\tor_u=-\tor_d$ to balance $\tor_d=g m_p l_p$; in the
case of Figure \ref{fig:balanced_con}, the links balance the applied
torque by setting $\Th_l=\Th_m=0$ and $\Th_u =
\dfrac{-\del}{gc_um_u}$.
% %
% \tbd{sec:lift}{Again, the modelling of the muscle spring using reduced
% gravity needs to be investigated.}
% %

%%\tbd{sec:dist}{Possibly need more on observer design and
%%      simulation} 

%%\tbd{sec:dist}{Rationalise the figures and describe in text.} 

% \section{Distributed intermittent control}
% \label{sec:dic}

% \section{Examples: mechanical system}
% \label{sec:ex:mech}

\section{Intermittency induces Variability}
\label{sec:variability}
\begin{figure}[htbp]
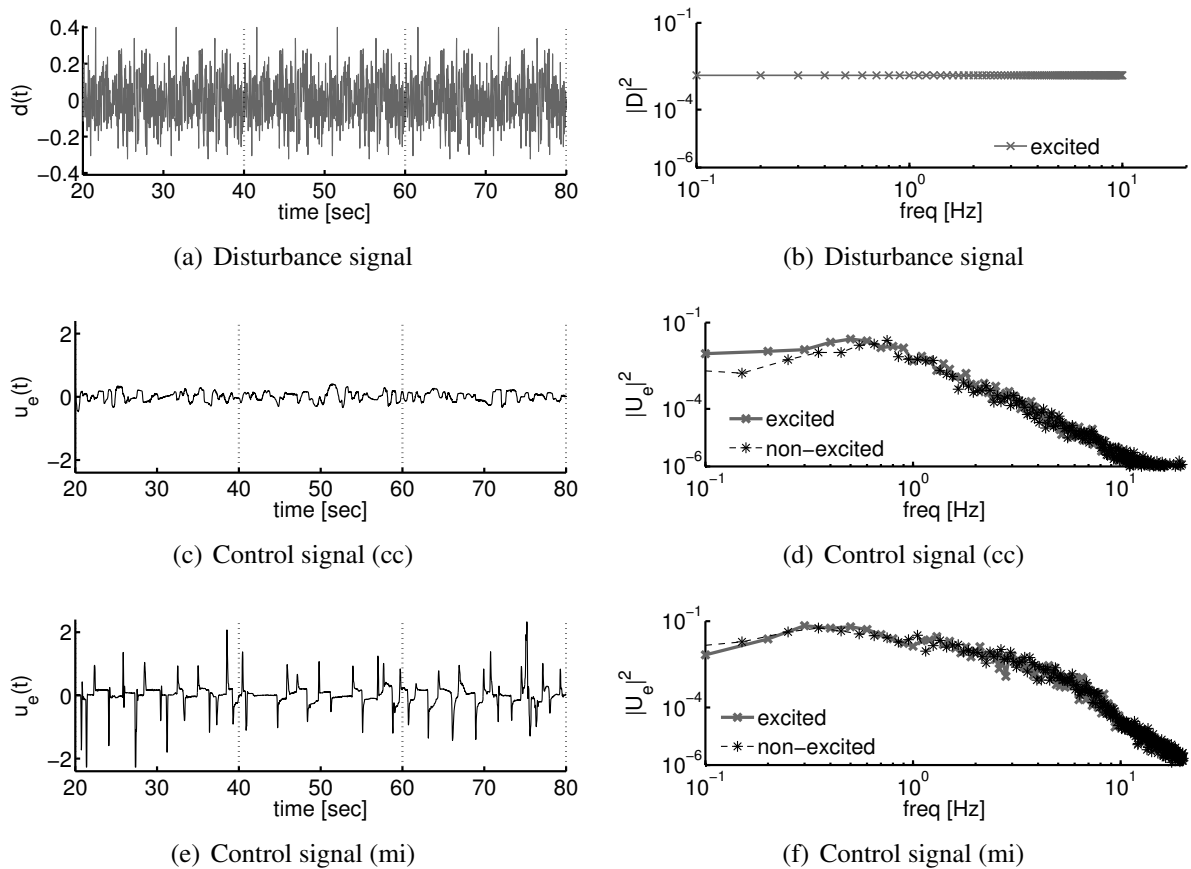

  \centering
  \SubFig{time_20sec_cc_d}{Disturbance signal}{0.47}%
  \SubFig{freq_20sec_cc_d}{Disturbance signal}{0.47}%
  \SubFig{time_20sec_cc_ue}{Control signal (cc)}{0.47}%
  \SubFig{freq_20sec_cc_ue}{Control signal (cc)}{0.47}%
  \SubFig{time_20sec_ol_ue}{Control signal (mi)}{0.47}%
  \SubFig{freq_20sec_ol_ue}{Control signal (mi)}{0.47}%
  \caption[Experimental data showing variability]{Experimental data
    showing variability during human motor control. Plots on the left
    side show time-domain signals, plots on the right side depict the
    corresponding frequency domain signals. ``cc'' -- keep close to
    centre (position control), ``mi'' -- minimal intervention
    (minimise control).}
  \label{fig:variability_example}
\end{figure}
Variability is an important characteristic of human motor control:
when repeatedly exposed to identical excitation, the response of the
human operator is different for each repetition. This is illustrated
in Figure~\ref{fig:variability_example}:
Figure~\ref{subfig:time_20sec_cc_d} shows a periodic input disturbance
(periodicity 10s), while Figures~\ref{subfig:time_20sec_cc_ue}
and~\ref{subfig:time_20sec_ol_ue} show the corresponding output signal
of a human controller for different control aims. It is clear that in
both cases, the control signal is different for each 10s period of
identical disturbance.

In the frequency domain, variability is represented by the observation
that, when the system is excited at a range of discrete frequencies
(as shown in Figure~\ref{subfig:freq_20sec_cc_d}, the disturbance
signal contains frequency components at $[0.1,\,0.2,\dots,10]$Hz), the
output response contains information at both the excited and the
non-excited frequencies (Figures~\ref{subfig:freq_20sec_cc_ue}
and~\ref{subfig:freq_20sec_ol_ue}). The response at the non-excited
frequencies (at which the excitation signal is zero) is termed the
remnant.

Variability is usually explained by appropriately constructed motor-
and observation noise which is added to a linear continuous-time model
of the human controller (signals $v_u$ and $v_y$ in
Figure~\ref{fig:MSS_arch}) \citep{LevBarKle69,KleBarLev70}.  While
this is currently the prominent model in human control, its
physiological basis is not fully established. This has led to the idea
that the remnant signal might be based on structure rather than
randomness \citep{NewDeuSos06}.

Intermittent control includes a sampling process, which is generally
based on thresholds associated with a trigger (see
Figure~\ref{fig:MIC_arch}). This non-uniform sampling process leads to
a time-varying response of the controller. It has been suggested that
the remnant can be explained by event-driven intermittent control
without the need for added noise \citep{MamGolGawLor11,GawGolMamLorLak13}, and that
this sampling process introduces variability \citep{GawLeeHalODw13}.

In this section we will discuss how intermittency can provide an
explanation for variability which is based on the controller structure
and does not require a random process. Experimental data from a
visual-manual control task will be used as an illustrative example.

\subsection{Experimental setup}
\label{sec:experimental-setup}

In this section, experimental data from a visual-manual control task
are used in which the participant were asked to use a sensitive,
contactless, uniaxial joystick to sustain control of an unstable 2nd
order system whose output was displayed as a dot on a oscilloscope
\citep{LorGolLakGaw10}. The controlled system represented an inverted
pendulum with a dynamic response similar to that of a human standing
(Load 2 of Table 1 in \citep{LorLakGaw09}),
\begin{equation}
  \label{eq:sys-balance}
  \begin{cases}
      \ddt{{\xx}}(t) &= \left[ \begin{array}{c c}
      -0.0372\; &  1.231\\
      1  &       0
    \end{array} \right] {\xx}(t) + \left[ \begin{array}{c}
      6.977 \\
      0
    \end{array} \right] (\uu(t) - \ddd(t))\\
      \yy(t) &= \left[ \begin{array}{c c}
      0 \;&  1
    \end{array} \right] {\xx}(t)\\
      \yy_o(t) &= \left[ \begin{array}{c c}
      1 \;&  1
    \end{array} \right] {\xx}(t)
  \end{cases}
\end{equation}
The external disturbance signal, $d(t)$, applied to the load input,
was a multi-sine consisting of $N_f=100$ discrete frequencies $\omega_k$, with resolution
$\omega_0=2\pi f_0$, $f_0=0.1\text{Hz}$ \citep{PinSch01}
\begin{equation}
  \label{eq:multisine}
  d(t)=\sum_{k=1}^{N_f}a_k\,cos\left(\omega_kt + \phi_k \right)
\qquad\text{with}\;\omega_k=2\pi\,k f_0
\end{equation}
The signal $d(t)$ is periodic with $T_0=1/f_0=10\text{s}$. To obtain
an unpredictable excitation, the phases $\phi_k$ are random values
taken from a uniform distribution on the open interval $(0,2\pi)$,
while $a_k=1$ for all $k$ to ensure that all frequency are equally
excited.

We considered two control priorities using the instructions ``keep the
dot as close to the centre as possible'' (``cc'', prioritising
position), and ``while keeping the dot on screen, wait as long as
possible before intervening'' (``mi'', minimising intervention).

\subsection{Identification of the linear time-invariable (LTI) response}
\label{sec:ident-non-vari}

Using previously established methods discussed in
Section~\ref{sec:iic} and by \citet{GolMamLorGaw12}, the design
parameters (i.e. LQ design weightings and mean time-delay, $\Delta$)
for an optimal, continuous-time linear predictive controller (PC)
(Figure~\ref{fig:MSS_arch}) are identified by fitting the complex
frequency response function relating $d$ to $u_e$ at the excited
frequencies. The linear fit to the experimental data is shown in
Figures~\ref{subfig:subj9_cc_r-1515_baseline_pc}
and~\ref{subfig:subj9_ol_r-1515_baseline_pc} for the two different
experimental instructions (``cc'' and ``mi''). Note that the PC only
fits the excited frequency components; its response at the non-excited
frequencies (bottom plots) is zero.
\begin{figure}[htbp]
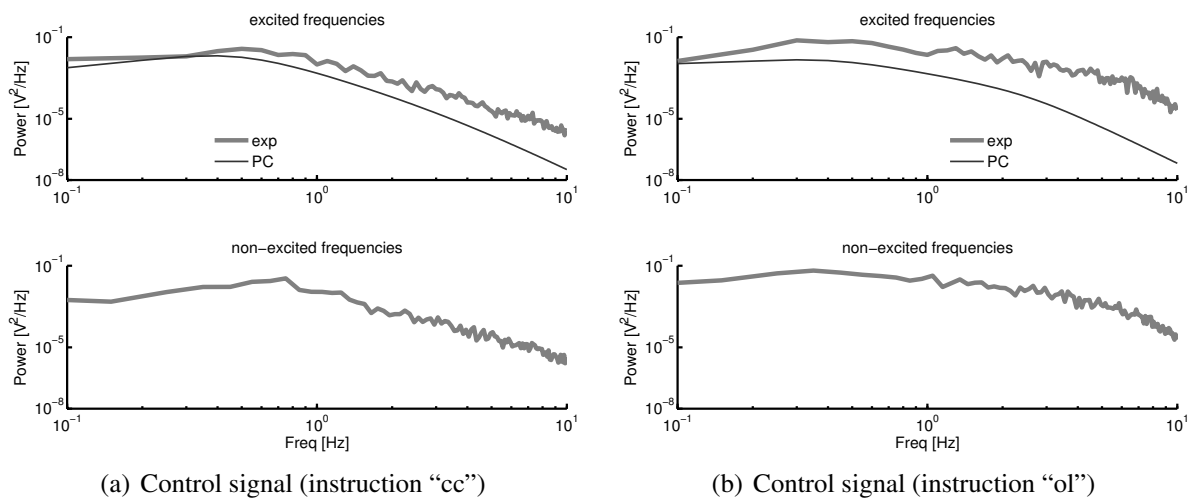

  \centering
  \SubFig{subj9_cc_r-1515_baseline_pc}{Control signal (instruction ``cc'')}{0.47}%
  \SubFig{subj9_ol_r-1515_baseline_pc}{Control signal (instruction ``ol'')}{0.47}\\
  \caption[Example individual result]{Example individual result for identification of LTI
    response for two experimental instructions. The continuous
    predictive controller can fit the excited frequencies (top graphs), but can not
    explain the experimental response at non-excited frequencies
    (bottom graphs).}
  \label{fig:result-individual}
\end{figure}

\subsection{Identification of the remnant response}
\label{sec:remnant}
The controller design parameters (i.e. the LQ design weightings)
obtained when fitting the LTI response, are used as the basis to model
the response at the non-excited (remnant) frequencies. First, the
standard approach of adding noise to a continuous PC is
demonstrated. Following this, it is shown that event driven IC can
approximate the experimental remnant response, by adjusting the
threshold parameters associated with the event trigger.

\subsubsection{Variability by adding noise}
\label{sec:pc-with-added}
For the PC, noise can be injected either as observation noise, $v_y$,
or as noise added to the input, $v_u$. The noise spectrum is obtained
by considering the measured response $u_e$ at non-excited frequencies
and, using the corresponding loop transfer function (see
Section~\ref{sec:pred-cont-contr}), calculating the noise input ($v_u$
or $v_y$) required to generate this. The calculated noise signal is
then interpolated at the excited frequencies.

Results for added input noise ($v_u$) are shown in
Figure~\ref{fig:noise}. As expected, the fit at the non-excited
frequencies is nearly perfect
(Figures~\ref{subfig:subj9_cc_r-1515_remnant_pc} and
\ref{subfig:subj9_ol_r-1515_remnant_pc}, bottom panels). Notably, the
added input noise also improves the fit at the excited frequencies
(Figures~\ref{subfig:subj9_cc_r-1515_remnant_pc} and
\ref{subfig:subj9_ol_r-1515_remnant_pc}, top panels).
\begin{figure}[htbp]
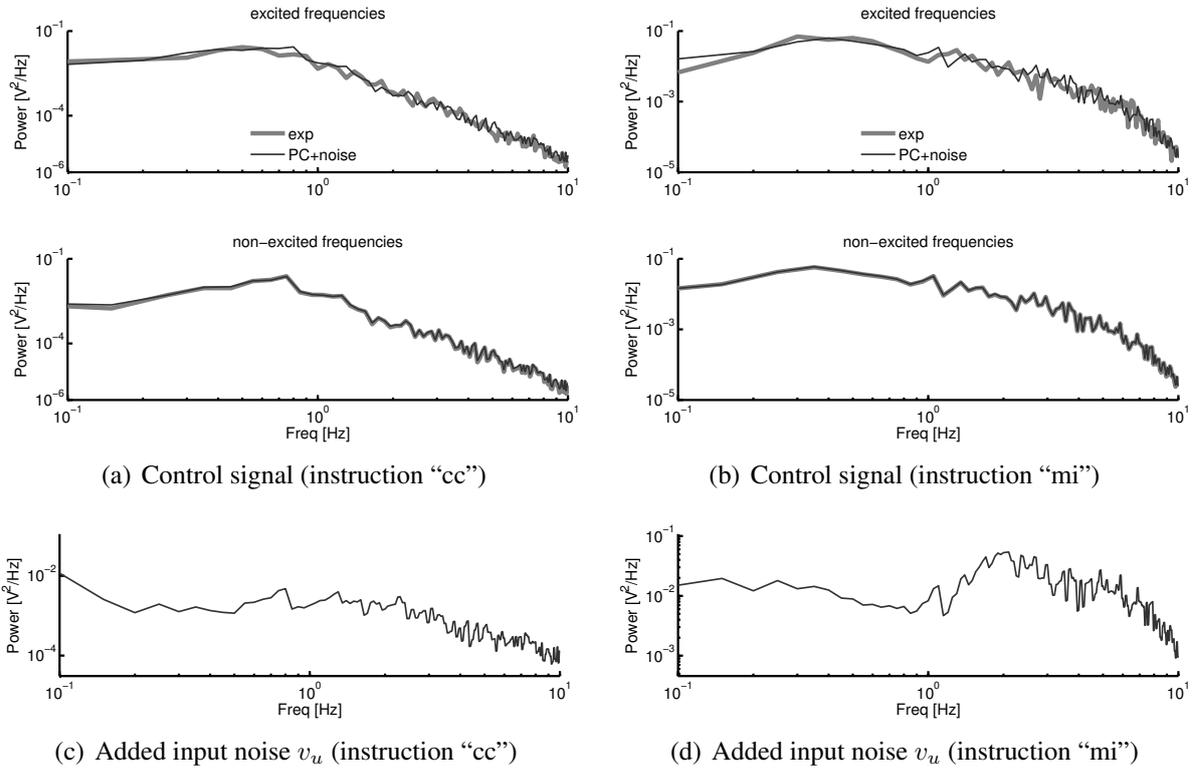

  \centering
  \SubFig{subj9_cc_r-1515_remnant_pc}{Control signal (instruction ``cc'')}{0.47}%
  \SubFig{subj9_ol_r-1515_remnant_pc}{Control signal (instruction ``mi'')}{0.47}\\
  \SubFig{subj9_cc_r-1515_noise}{Added input noise $v_u$ (instruction ``cc'')}{0.47} 
  \SubFig{subj9_ol_r-1515_noise}{Added input noise $v_u$ (instruction ``mi'')}{0.47} 
  \caption[Variability as a result of coloured input
  noise]{Variability as a result of coloured input noise added to a
    predictive continuous controller. The top graphs show the
    resulting fit to experimental data at excited and non-excited
    frequencies. The bottom graphs show the input noise added.}
  \label{fig:noise}
\end{figure}

The spectra of the input noise $v_u$ are shown in
Figures~\ref{subfig:subj9_cc_r-1515_noise}
and~\ref{subfig:subj9_ol_r-1515_noise}. It can be observed that the
noise spectra are dependent on the instructions given (``cc'' or
``mi''), with no obvious physiological basis to explain this
difference.

\subsubsection{Variability by intermittency}
\label{sec:vari-interm}

As an alternative explanation, a noise-free event driven intermittent
controller is considered (cf. Figure~\ref{fig:MIC_arch}). The same
design parameters as for the PC are used, with the time-delay set to a
minimal value of $\Delta^{min}=0.1$sec and a corresponding minimal
intermittent interval, $\Delta^{\text{min}}_{ol}=0.1$sec. 

Variations in the loop-delay are now the result of the event
thresholds, cf. equation~\eqref{eq:ED_e}. In particular, we consider
the first two elements of the state prediction error $e_{hp}$,
corresponding to the velocity ($e_{hp}^{v}$) and position
($e_{hp}^{p}$) states, and define an ellipsoidal event detection
surface given by
\begin{equation}
  \label{eq:thresholds}
  \left(\frac{e_{hp}^{p}}{\thp}\right)^2 + \left(\frac{e_{hp}^{v}}{\thv}\right)^2>1
\end{equation}
where $\thp$ and $\thv$ are the thresholds associated with the
corresponding states.

To find the threshold values which resulted in simulation which best
approximates the experimental remnant, both thresholds were varied
between 0 (corresponding to clock-driven IC) and 3, and the threshold
combination that resulted in the best least-squares fit at all
frequencies (excited and non-excited) was selected as the optimum. The
resulting fit is shown in
Figures~\ref{subfig:subj9_cc_r-1515_remnant_ic} and
\ref{subfig:subj9_ol_r-1515_remnant_ic}. For both instructions, the
event driven IC can both, explain the remnant signal and improve the
fit at excited frequencies. 
\begin{figure}[htbp]
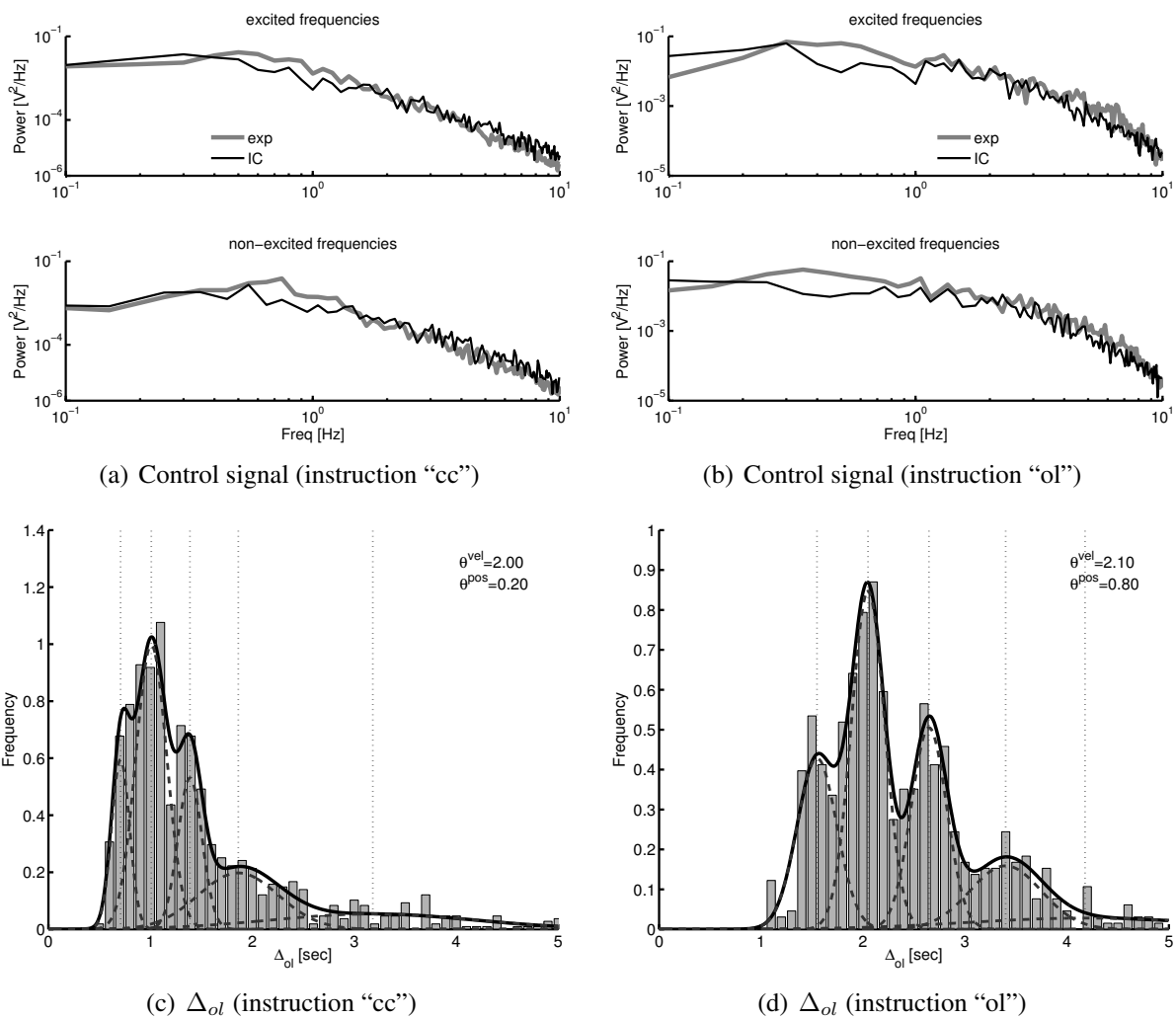

  \centering
  \SubFig{subj9_cc_r-1515_remnant_ic}{Control signal (instruction ``cc'')}{0.47}%
  \SubFig{subj9_ol_r-1515_remnant_ic}{Control signal (instruction ``ol'')}{0.47}\\
  \SubFig{subj9_cc_r-1515_delta_ol_comp}{$\Delta_{ol}$ (instruction ``cc'')}{0.47} 
  \SubFig{subj9_ol_r-1515_delta_ol_comp}{$\Delta_{ol}$ (instruction ``ol'')}{0.47} 
  \caption[Variability resulting from event-driven IC]{Variability
    resulting from event-driven IC. The top graphs show the resulting
    fit to experimental data at excited and non-excited
    frequencies. The bottom graphs show the distribution of
    intermittent intervals, together with the optimal threshold
    values.}
 \label{fig:result-individual-ic}
\end{figure}

The corresponding thresholds (for ``cc'': $\thv=2.0,\,\thp=0.2$, for
``mi'': $\thv=2.1,\,\thp=0.8$) reflect the control priorities for each
instruction: for ``cc'' position control should be prioritised,
resulting in a small value for the position threshold, while the
velocity is relatively unimportant. For ``mi'' the control
intervention should be minimal, which is associated with large
thresholds on both, velocity and position.

Figures~\ref{subfig:subj9_cc_r-1515_delta_ol_comp}
and~\ref{subfig:subj9_ol_r-1515_delta_ol_comp} show the distributions
of the open loop intervals for each condition, together with an
approximation by a series of weighted Gaussian distributions
\citep{McLPee00}. For position control (``cc''), open loop intervals
are clustered around a modal interval of approximately 1s, with all
$\Dol>0.5$s. For the minimal intervention condition (``mi''), the open
loop intervals are clustered around a modal interval of approximately
2s, and all $\Dol>1$s. This corresponds to the expected behaviour of
the human operator where more frequent updates of the intermittent
control trajectory are associated with the more demanding position
control instruction, while the instruction to minimise intervention
results in longer intermittent intervals. Thus the identified
thresholds not only result in IC models which approximate the response
at excited and non-excited frequency, but also reflect the underlying
control aims.

\subsection{Conclusion}
\label{sec:conclusion-variability}
The hypothesis that variability is the result of a continuous control
process with added noise (PC with added noise), requires that the
remnant is explained by a non-parametric input noise component. In
comparison, IC introduces variability as a result of a small number of
threshold parameters which are clearly related to underlying control
aims.

\section[Identification: the underlying continuous
system]{Identification of intermittent control: the underlying
  continuous system}
\label{sec:iic}
This section, together with Section \ref{sec:iic:intermittency},
addresses the question of how intermittency can be identified when
observing closed loop control. In this Section it is discussed how
intermittent control can masquerade as continuous control, and how the
underlying continuous system can be
identified. Section~\ref{sec:iic:intermittency} addresses the question
how intermittency can be detected in experimental data.

% \subsection{Identifying the underlying continuous system}
% \label{sec:iic:continuous}

% In this section we consider the task of sustained control of a load
% and ask the question how this process can be described by a control
% loop using a continuous or an intermittent controller.

System identification provides one approach to hypothesis testing and
has been used by \citet{JohMagAke88} and \citet{Pet02} to test the
non-predictive hypothesis and by \citet{GawLorLak09} to test the
non-predictive and predictive hypotheses.
Given time domain data from an sustained control task which is excited
by an external disturbance signal, a two stage approach to controller
estimation can be used in order to perform the parameter estimation in
the frequency domain: firstly, the frequency response function is
estimated from measured data, and secondly, a parametric model is
fitted to the frequency response using non-linear optimisation
\citep{PinSch01,PinSchRol08}.  This approach has two advantages:
firstly, computationally expensive analysis of long time-domain data
sets can be reduced by estimation in the frequency domain, and
secondly, advantageous properties of a periodic input signal (as
advocated by \citet{PinSchRol08}) can be exploited.

In this section, first the derivation of the underlying frequency
responses for a predictive continuous time controller and for the
intermittent, clock-driven controller
(i.e. $\Delta_{ol}=\text{const}$) is discussed. The method is limited
to clock-driven IC since frequency analysis tools are readily
available only for this case \citep{Gaw09}. The two stage
identification procedure is then outlined, followed by example results
from a visual-manual control task.

The material in this section is partially based on \citet{GolMamLorGaw12}.

\subsection{Closed-loop frequency response}
\label{sec:FR}
As a prerequisite for system identification in the frequency domain,
this section looks at the frequency response of closed-loop system
corresponding to the underlying predictive continuous design method as
well as that of the intermittent controller with a fixed intermittent
interval \citep{Gaw09}.

\subsubsection{Predictive continuous control}
\label{sec:pred-cont-contr}
The system equations~\eqref{eq:sys} can be rewritten in transfer function form as
\begin{align}
  y_o(s) &= G(s) u(s)\notag\\ 
   %&= e^{-s\Delta}G(s) u_0(s)\\
\text{with }
G(s) &= C \left [ sI - A \right ]^{-1} B\label{eq:Gs}
\end{align}
where $I$ is the $n\times n$ unit matrix and $s$ denotes the complex
Laplace operator.

Transforming equations~\eqref{eq:obs} and~\eqref{eq:pred} into the
Laplace domain, assuming that disturbances $v_u$, $v_u$ and $w$ are zero:
\begin{alignat}{2}
  \xo(s)  &= \left ( sI - A_o \right )^{-1}( B u(s) +
  Ly_o(s)) \;\; &\text{(Observer)} \label{eq:obs_s_p}\\
  {x_p}(s)  &= e^{A\Delta}\hat{x}(s) 
  + \left ( sI - A \right )^{-1} \left (I - e^{-(sI - A)\Delta} \right ) B u(s)
  \;\;&\text{(Predictor)} \label{eq:pred_s}\\
  u(s) &= - k e^{-s\Delta}{{x}_p}(s)  &\text{(Controller)}\label{eq:con_s}
\end{alignat}
where $I$ is the $n\times n$ unit matrix.

Equations (\ref{eq:obs_s_p})--(\ref{eq:con_s}) can be rewritten as:
\begin{align}
  u(s) &= -e^{-s\Delta}\left[H_y(s) y(s) + (H_1(s) + H_2(s)) u(s)\right]\\
  \text{where } H_y(s) &= ke^{A\Delta}\left ( sI - A_o \right )^{-1} L\\
  H_1(s) &= ke^{A\Delta}\left ( sI - A_o \right )^{-1} B\\
  \text{and } H_2(s) &= k\left ( sI - A \right )^{-1} 
  \left (I - e^{-(sI - A)\Delta} \right ) Be^{s\Delta}
\end{align}

It follows that the controller transfer function $H(s)$ is given by:
\begin{equation}
  \label{eq:H(s)PC}
  H(s) = \frac{H_y(s)}{1 + H_1(s) + H_2(s)}
\end{equation}
where
\begin{equation}
  \label{eq:H(s)}
  \frac{u(s)}{y_o(s)} = - e^{-s\Delta}H(s)
\end{equation}
With equations~\eqref{eq:Gs} and~\eqref{eq:H(s)}, the system loop-gain $L(s)$ and closed-loop transfer function
$T$ are given by:
\begin{align}
  L(s) &= e^{-s\Delta}G(s)H(s)\label{eq:L(s)}\\
  T(s) &= \frac{u(s)}{d(s)}=\frac{L(s)}{1+L(s)}\label{eq:T(s)L}
\end{align}
Equation~\eqref{eq:T(s)L} gives a parametrised expression relating $u(s)$ and $d(s)$.

\subsubsection{Intermittent Control}
\label{sec:FRF-ic}
The sampling operation in Figure~\ref{fig:MIC_arch} makes it harder to
derive a (continuous-time) frequency response and so the details are
omitted here. For the case were the intermittent interval is assumed
to be constant, the basic result derived by \citet{Gaw09} apply and can
be encapsulated as the following theorem\footnote{This is a simplified
  version of \cite[Theorem 1]{Gaw09} for the special case considered in this
  Section.}:

\subparagraph{Theorem} \emph{ The continuous-time
  system~\eqref{eq:sys} controlled by an intermittent controller
  with generalised hold gives a closed-loop system where the Fourier
  transform $\UU$ of the control signal $u(t)$ is given in terms of
  the Fourier transform $\Xd$ by }
\begin{equation}
  \UU   = \FF \samp{\Xd} \label{eq:UUz}\\
\end{equation}
\emph{where}
\begin{align}
  \FF &= \HH \Senz \label{eq:FF}\\
  \HH &= \frac{1}{\Delta_{ol}}\kk\left [\jw I - A_c \right
        ]^{-1} \left [ I - e^{-(\jw I - A_c)\Delta_{ol}} \right
        ] \label{eq:H}\\
  \Senz &=  [I + \GGz]^{-1}\label{eq:Sz}\\
  \GGz &= \left [e^{\jw}I - A_x \right ]^{-1}
         B_x \label{eq:G_z}\\
  \Xd &= \GG\DD \\
  \GG &= [\jw I - A]^{-1}B\label{eq:G}
\end{align}
\emph{The sampling operator is defined as}
\begin{equation}
  \label{eq:s-notation}
  \samp{\Xd} = \sum_{k=-\infty}^{\infty}  \FRf{\X^d}
\end{equation}
\emph{
  where the \emph{intermittent sampling-frequency} is given by
  $\omega_{ol} = 2\pi/\Delta_{ol}$.
} 

As discussed in \citet{Gaw09}, the presence of the sampling operator
$\samp{\Xd}$ means that the interpretation of $\FF$ is not quite the
same as that of the closed loop transfer function $T(s)$ of
(\ref{eq:T(s)L}), as the sample process generates an infinite number
of frequencies which can lead to aliasing. As shown in \citet{Gaw09},
the (bandwidth limited) observer acts as an anti-aliasing filter,
which limits the effect of $\samp{\Xd}$ to higher frequencies and
makes $\FF$ a valid approximation of $\UU$.
% However, following \cite{Gaw09}, it 
$\FF$ will therefore be treated as equivalent to
$T(j\omega)$ in the rest of this Section.

\subsection{System identification} 
\label{sec:ident}

The aim of the identification procedure is to derive an estimate for
the closed-loop transfer function of the system.
Our approach follows the two stage procedure of \citet{PinSch01} and
\citet{PinSchRol08}.  In the first step, the frequency response
transfer function is estimated based on measured input--output data,
resulting in a non-parametric estimate. In the second step, a
parametric model of the system is fitted to the estimated frequency
response using an optimisation procedure.

\subsubsection{System setup}
\label{sec:system-setup}

To illustrate the approach, we consider the visual-manual control task
described in Section~\ref{sec:experimental-setup}, where the subject
is asked to sustain control of an unstable 2nd order load using a
joystick, with the instruction to keep the load as close to the centre
as possible (``cc''). 

\subsubsection{Non-parametric estimation}
\label{sec:non-param-estim}

In the first step, a non-parametric estimate of the closed loop
frequency response function (FRF) is derived, based on observed
input--output data. The system was excited by a multi-sine disturbance
signal (equation~\eqref{eq:multisine}). The output $u(t)$ of a linear
system which is excited by $d(t)$ then only contains information at
the same discrete frequencies $\omega_k$ as the input signal. If the
system is non-linear or noise is added, the output will contain a
remnant component at non-excited frequencies as discussed in
Section~\ref{sec:variability}.  several periods was used.

The time domain signals $d(t)$ and $u(t)$ over one period $T_0$ of the
excitation signal were transformed into the frequency domain. If the
input signal has been applied over $N_p$ periods, then the
frequency-domain data for the $l$th period can be denoted as
$d^{[l]}(j\omega_k)$ and $u^{[l]}(j\omega_k)$, respectively, and the
FRF can be estimated as
\begin{equation}
  \label{eq:5}
  \hat{T}^{[l]}(j\omega_k)=\frac{u^{[l]}(j\omega_k)}{d^{[l]}(j\omega_k)},
  \qquad k=1,2,\dots,N_f
\end{equation}
where $N_f$ denotes the number of frequency components in the
excitation signal. An estimate of the FRF over all $N_p$ periods is
obtained by averaging,
\begin{equation}
  \label{eq:6}
  \hat{T}(j\omega_k)=\frac{1}{N_p}\sum_{l=1}^{N_p}\hat{T}^{[l]}(j\omega_k),
  \qquad k=1,2,\dots,N_f
\end{equation}
This approach ensures that only the periodic (deterministic) features
related to the disturbance signal are used in the identification, and
that the identification is robust with respect to remnant components.

\subsubsection{Parametric optimisation}
\label{sec:param-optim}

In the second stage of the identification procedure, a parametric
description, $\tilde{T}(j\omega_k,\theta)$, is fitted to the estimated
FRF of equation~(\ref{eq:6}). The parametric FRF approximates the
closed loop transfer function (equation~(\ref{eq:T(s)L})) which
depends in the case of predictive control, on the loop transfer
function $L(j\omega_k,\theta)$, equation~(\ref{eq:L(s)}), parametrised
by the vector $\theta$, while for the intermittent controller this is
approximated by $F(j\omega,\theta)$, equation~(\ref{eq:UUz}),
\begin{equation}
  \label{eq:12}
  \tilde{T}(j\omega_k,\theta) =
  \begin{cases}
    \frac{L(j\omega_k,\theta)}{1+L(j\omega_k,\theta)}
    & \text{for PC}\\
    \FFk &\text{for IC}
  \end{cases}
\end{equation}

We use an indirect approach to parametrise the controller, where the
controller and observer gains are derived from optimised design
parameters using the standard LQR approach of
equation~\eqref{eq:LQ}. This allows the specification of boundaries
for the design parameters which guarantee a nominally stable closed
loop system. As described in Section~\ref{sec:condes}, the feedback
gain vector $k$ can then be obtained by choosing the elements of the
matrices $Q_c$ and $R_c$ in~\eqref{eq:LQ}, and nominal stability can
be guaranteed if these matrices are positive definite. As the system
model is second order, we choose to parametrise the design using two
positive scalars, $q_v$ and $q_p$,
\begin{equation}
  \label{eq:9}
   R_c=1 \qquad
   Q_c=\left[\begin{array}{c c} q_v & 0\\0 & q_p \end{array}\right],
\text{ with } q_v,\,q_p >0
\end{equation}
related to relative weightings of the velocity ($q_v$) and position
($q_p$) states.

The observer gain vector $L$ is obtained by applying the same approach
to the dual system $[A^T,C^T,B^T]$. It was found that the results
are relatively insensitive to observer properties which was therefore
parametrised by a single positive variable, $q_o$,
\begin{equation}
  \label{eq:10}
  R_o=1 \qquad
  Q_o= q_o\, BB^T 
  \text{ with } q_o>0 
\end{equation}
where $R_o$ and $Q_o$ correspond to $R_c$ and $Q_c$ in equation~(\ref{eq:LQ})
for the dual system.

The controller can then be fully specified by the positive parameter vector
$\theta=[q_v, q_p, q_o,\Delta]$ (augmented by $\Delta_{ol}$ for
intermittent control).

The optimisation criterion $J$ is defined as the mean squared
difference between the estimated FRF and its parametric fit
\begin{equation}
  \label{eq:8}
  J(\theta)=\frac{1}{N_f}
   \sum_{k=1}^{N_f} \left[\hat{T}(j\omega_k)-\tilde{T}(j\omega_k,\theta)\right]^2
\end{equation}
This criterion favours lower frequency data since $|T(\jw)|$ tends to
be larger in this range.

The parameter vector is separated into two parts, time delay
parameters,
\begin{equation}
  \label{eq:14}
  \theta_{\Delta}=
  \begin{cases}
      [\Delta] & \text{for PC}\\ 
      [\Delta,\Delta_{ol}] & \text{for IC}
  \end{cases}
\end{equation}
and controller design parameters
\begin{equation}
    \theta_{c}=[q_v,q_p,q_o]
\end{equation}
such that $\theta = \left[\theta_{\Delta}, \theta_c\right]$.  The time
delay parameters are varied over a predefined range, with the
restriction that $\Delta_{ol}> \Delta$ for IC. For each given set of
time delay parameters, a corresponding set of optimal controller
design parameters $\theta^*_{c}$ is found which solves the constrained
optimisation problem
\begin{equation}
  \label{eq:11}
  \theta_{c}^*=\arg \min_{\theta_c}J([\theta_\Delta,\theta_c]),\qquad \theta_c>0 
\end{equation}
which was solved using the SQP algorithm \citep{NocWri06}, (MATLAB
Optimization Toolbox, Mathworks, USA).

The optimal cost function for each set of time-delay parameters,
$J^*(\theta_\Delta)$, was calculated, and the overall optimum, $J^*$
determined. For analysis, the time-delay parameters corresponding to
the optimal cost are determined, with $\Delta$ and $\Delta_{ol}$
combined for the IC to give the effective time-delay,
\begin{equation}
  \label{eq:Delta-effective}
  \Delta_e= \Delta+0.5\Delta_{ol} 
\end{equation}

\subsection{Illustrative example}
\label{sec:illustrative-results}
\begin{figure}[htbp]
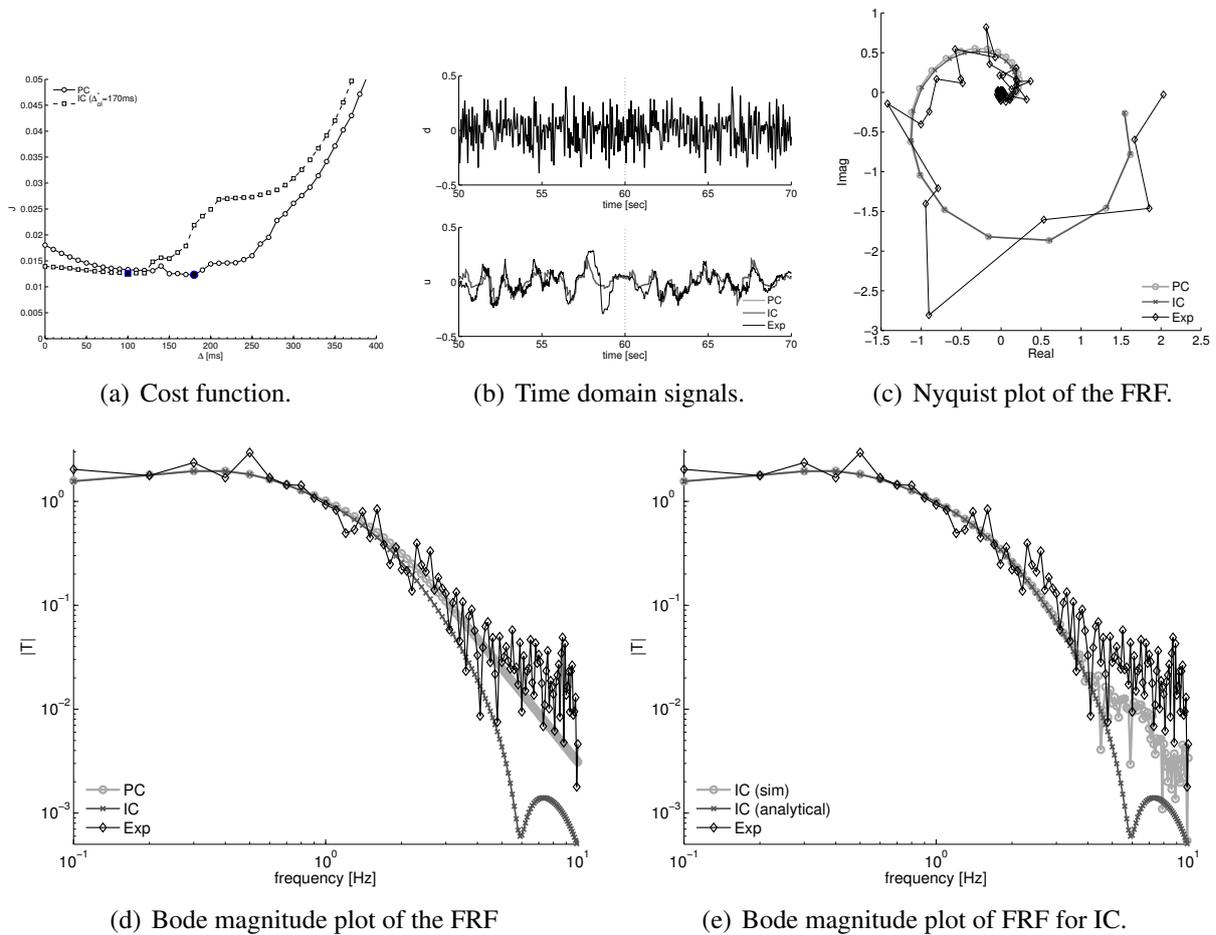

  \centering 
  \SubFig{exp-J-7-cc}{Cost function.}{0.31}%
  \SubFig{time-exp-s7}{Time domain signals.}{0.31}%
  \SubFig{nyq-exp-s7}{Nyquist plot of the FRF.}{0.31}%
  \SubFig{bode-exp-s7}{Bode magnitude plot of the FRF}{0.47}%
  \SubFig{bode-ic-exp-s7}{Bode magnitude plot of FRF for IC.}{0.47}%
  \caption[Illustrative experimental results]{Illustrative experimental results. (a) shows the optimisation
    cost (equation~\eqref{eq:8}) as a function of the time delay for
    predictive continuous control (PC) and intermittent control (IC),
    together with the value of the intermittent interval corresponding
    to the smallest cost ($\Delta^*_{ol}$). (b)-(d) show comparisons
    between predictive continuous control (PC), intermittent control
    (IC) and the experimental data (Exp). Plots for PC and IC in (c)
    and (d) are derived analytically (equation~\eqref{eq:12}). (e)
    compares the analytical FRF for the IC with the FRF derived from
    time-domain simulation data.}
  \label{fig:exp-s7}
\end{figure}

Results from identifying the experimental data from one subject are
used to illustrate the approach. 

An extract of the time domain data are shown in
Figure~\ref{subfig:time-exp-s7}. The top plot shows the multi-sine
disturbance input over two 10sec periods, and the bottom plot depicts
the corresponding measured control signal response (thin dark
line). From this response, the experimental FRF was estimated in
stage 1 of the identification (dark solid lines in Figures~\ref{subfig:nyq-exp-s7}--\ref{subfig:bode-ic-exp-s7}).

Stage 2 of the procedure aimed to find the controller design
parameters which resulted in the best fit to the experimental FRF. The
corresponding cost functions for the predictive continuous and for the
intermittent controller are shown in Figure~\ref{subfig:exp-J-7-cc},
with the minima indicated by solid markers. The estimated FRF
(equation~\eqref{eq:6}) and their parametric fits
(equation~\eqref{eq:12}) are shown in
Figure~\ref{subfig:nyq-exp-s7}. It is clear that both the PC and IC
are able to fit the experimental FRF equally well. The resulting
controller parameters (summarised in table~\ref{tab:ctr-parameters})
are very similar for both control architectures. This is confirmed by
time-domain simulations using the estimated controllers
(Figure~\ref{subfig:time-exp-s7}) where the PC and IC responses are difficult to
distinguish.
\begin{table}[htbp]
  \centering
  \begin{tabular}{l  c c}
        & \textbf{PC} & \textbf{IC} \\\hline
    $q_p$ & 0.99 & 1.07  \\
    $q_v$ & 0.00 &  0.00 \\
    $q_o$ & 258.83 & 226.30 \\
    $\Delta$ & 180ms & 95ms \\
    $\Delta_{ol}$ & -- & 170ms \\
    $\Delta_{e}$ & -- & 180ms \\\hline
  \end{tabular}
  \caption{Estimated controller design parameters}
  \label{tab:ctr-parameters}
\end{table}

Although the Nyquist plot of Figure~\ref{subfig:nyq-exp-s7} suggests
that the PC and IC responses are virtually identical, further analysis
shows that this is only the case at lower frequencies (at which most
of the signal power lies). The Bode plot of the frequency response
(Figure~\ref{subfig:bode-exp-s7}) shows that the PC and IC are
indistinguishable only for frequencies up to around 2-3Hz. This is
also the frequency range up to which the PC and IC provide a good
approximation to the experimental FRF.

The controller frequency responses shown in
Figures~\ref{subfig:nyq-exp-s7} and~\ref{subfig:bode-exp-s7} are based
on the analytically derived expressions. For the IC, the sampling
operator means that the theoretical response is only a valid
approximation at lower frequencies, with aliasing evident at higher
frequencies. A comparison of the analytical response with the response
derived from simulated time-domain data
(Figure~\ref{subfig:bode-ic-exp-s7}) shows that the simulated frequency
response of the IC at higher frequency is in fact closer to the experimental
data than the analytical response.

\subsection{Conclusions}
\label{sec:conclusions}

The results illustrate that continuous predictive and intermittent
controllers can be equally valid descriptions of a sustained control
task. Both approaches allow fitting the estimated non-parametric
frequency responses with comparable quality. This implies that
experimental data can be equally well explained using the PC and the
IC hypotheses. This result is particularly interesting as it means
that experimental results showing good fit for continuous predictive
control models, dating back to at least those of \citet{KleBarLev70},
do not rule out an intermittent explanation.  A theoretical
explanation for this result is given in \citep[Section
4.3]{GawLorLakGol11} where the masquerading property of intermittent
control is discussed: As shown there (and illustrated in the results
here), the frequency response of an intermittent controller and that
of the corresponding predictive controller are indistinguishable at
lower frequency and only diverge at higher frequencies where aliasing
occurs.  Thus, the responses of the predictive and the intermittent
controllers are difficult to distinguish, and therefore both
explanations appear to be equally valid.

\section[Identification: Detecting intermittency]{Identification of intermittent control: Detecting intermittency}
\label{sec:iic:intermittency}
%
%%\tbd{sec:iic:intermittency}{Ian to check: figs and cites in particular.}
As discussed in Section \ref{sec:ic} and by \citet{GawLorLakGol11},
the key feature distinguishing intermittent control from continuous control is the
open-loop interval $\Delta_{ol}$ of Equations (\ref{eq:Delta_i}) and (\ref{eq:PRP}). 
As noted in Sections \ref{sec:ic_time} and \ref{sec:ex:prp}, the
open-loop interval provides an explanation of the Psychological
Refractory Period (PRP) of \citet{Tel31} as discussed by \citet{Vin48}
to explain the human response to double stimuli. Thus
``intermittency'' and ``refractoriness'' are intimately related.
Within this interval, the control trajectory is open loop but is
continuously time varying according to the basis of the generalised
hold. The length of the intermittent interval gives a trade-off
between continuous control (zero intermittent interval) and
intermittency. Continuous control maximises the frequency bandwidth
and stability margins at the cost of reduced flexibility whereas
intermittent control provides time in the loop for optimisation and
selection \citep{KamGawGolLakLor13,LorKamLakGolGaw14} at the cost of
reduced frequency bandwidth and reduced stability margins. The
rationale for intermittent control is that it confers online
flexibility and adaptability. This rationale has caused many
investigators to consider whether intermittent control is an
appropriate paradigm for understanding biological motor control
\citep{Cra47a,Cra47b,Vin48,Bek62,NavSta68,NeiNei88,MiaWeiSte93a,HanBerDroSlo97,NeiNei05,LorLak02}.
However, even though intermittent control was first proposed in the
physiological literature in 1947, there has not been an adequate
methodology to discriminate intermittent from continuous control and
to identify key parameters such as the open loop interval
$\Delta_{ol}$.  Within the biological literature, four historic planks
of evidence (discontinuities, frequency constancy, coherence limit and
psychological refractory period) have provided evidence of
intermittency in human motor control \citep{LorKamLakGolGaw14}.
\begin{enumerate}
\item The existence of discontinuities within the control signal has
  been interpreted as sub-movements or serially planned control
  sequences
  \citep{NavSta68,Pou74,MiaWeiSte93a,MiaWeiSte86,HanBerDroSlo97,LorLak02},
\item Constancy in the modal rate of discontinuities, typically around
  2-3 per second, has been interpreted as evidence for a central
  process with a well defined timescale \citep{NavSta68,Pou74,LakLor06,LorGawLak06}.
\item The fact that coherence between unpredicted disturbance or
  set-point and control signal is limited to a low maximum frequency,
  typically of 1-2 Hz, below the mechanical bandwidth of the feedback
  loop has been interpreted as evidence of sampling
  \citep{NavSta68,LorLakGaw09,LorGolLakGaw10}.
\item The psychological refractory period has provided direct evidence
  of open loop intervals but only for discrete movements and serial
  reaction time (e.g. push button) tasks and has not been demonstrated
  for sustained sensori-motor control \citep{Vin48,PasJoh98,HarRotMia13}.
\end{enumerate}
Since these features can be reproduced by a continuous controller with
tuned parameters and filtered additive noise
\citep{LevBarKle69,LorKamGolGaw12}, this evidence is
circumstantial. Furthermore, there is no theoretical requirement for
regular sampling nor for discontinuities in control
trajectory. Indeed, as historically observed by \citet{Cra47a,Cra47b},
humans tend to smoothly join control trajectories following
practice. Therefore the key methodological problem is to demonstrate that
on-going control is sequentially open loop even when the control
trajectory is smooth and when frequency analysis shows no evidence of
regular sampling.

\subsection{Outline of method}
\begin{figure}[htbp]
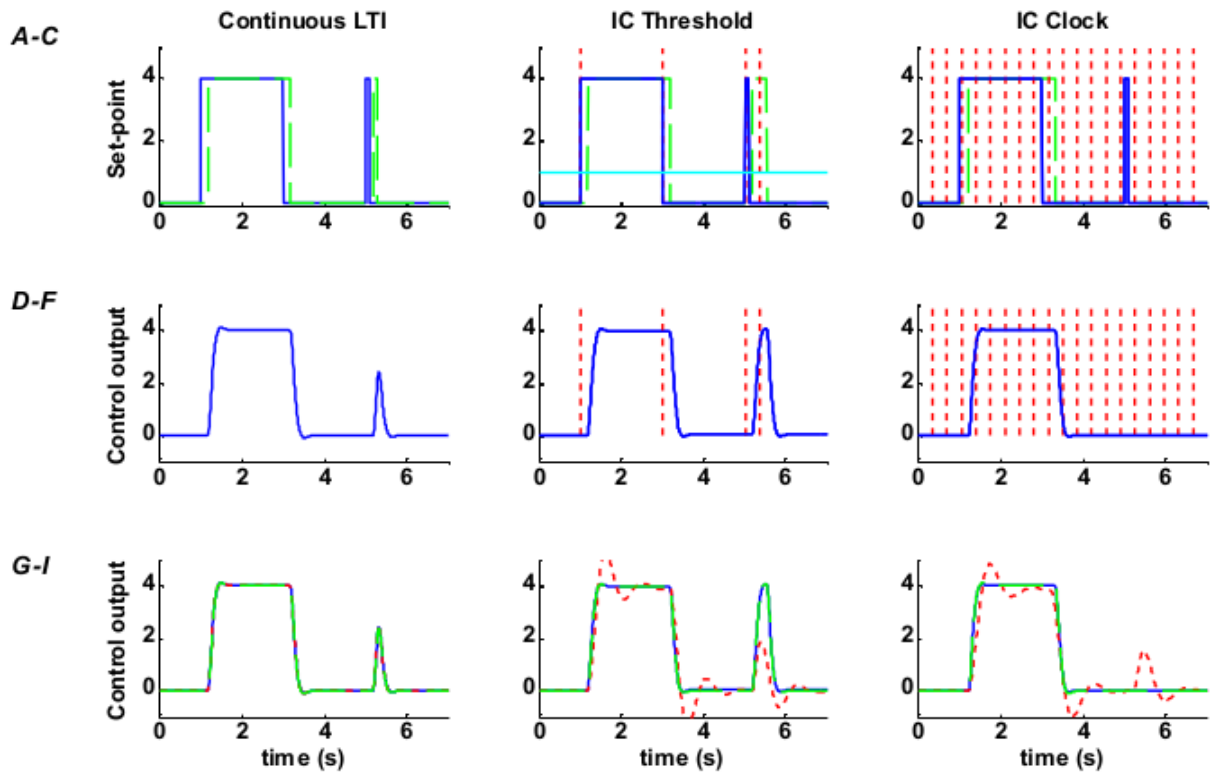

  \centering 
  \Fig{RSI_fig2}{1}
  \caption[Reconstruction of the set-point]{Reconstruction of the
    set-point. Example responses to set-point step-sequence A--C. Solid:
    two paired steps (long inter-step interval, short inter-step
    interval) are applied to the set-point of each of three models:
    continuous linear time invariant, threshold triggered intermittent
    control (unit threshold), and clock triggered (zero threshold)
    intermittent control (cols 1--3 respectively). Dashed: Set-point
    adjusted: time of each step follows preceding trigger by one model
    time delay ($\td$).  D--F. Solid: Control output (ue). Red
    vertical dashed: event trigger times.  G--I. Solid: Control output
    (ue). Dash-dotted: ARMA (LTI) fit to set-point (solid in A--C).
    Dashed: ARMA (LTI) fit to adjusted set-point (dashed in
    A--C). [\citep{LorKamGolGaw12} Copyright \textcopyright 2012 the
    authors. Used with permission.]}
\label{fig:RSI_fig2}
\end{figure}
\begin{figure}[htbp]
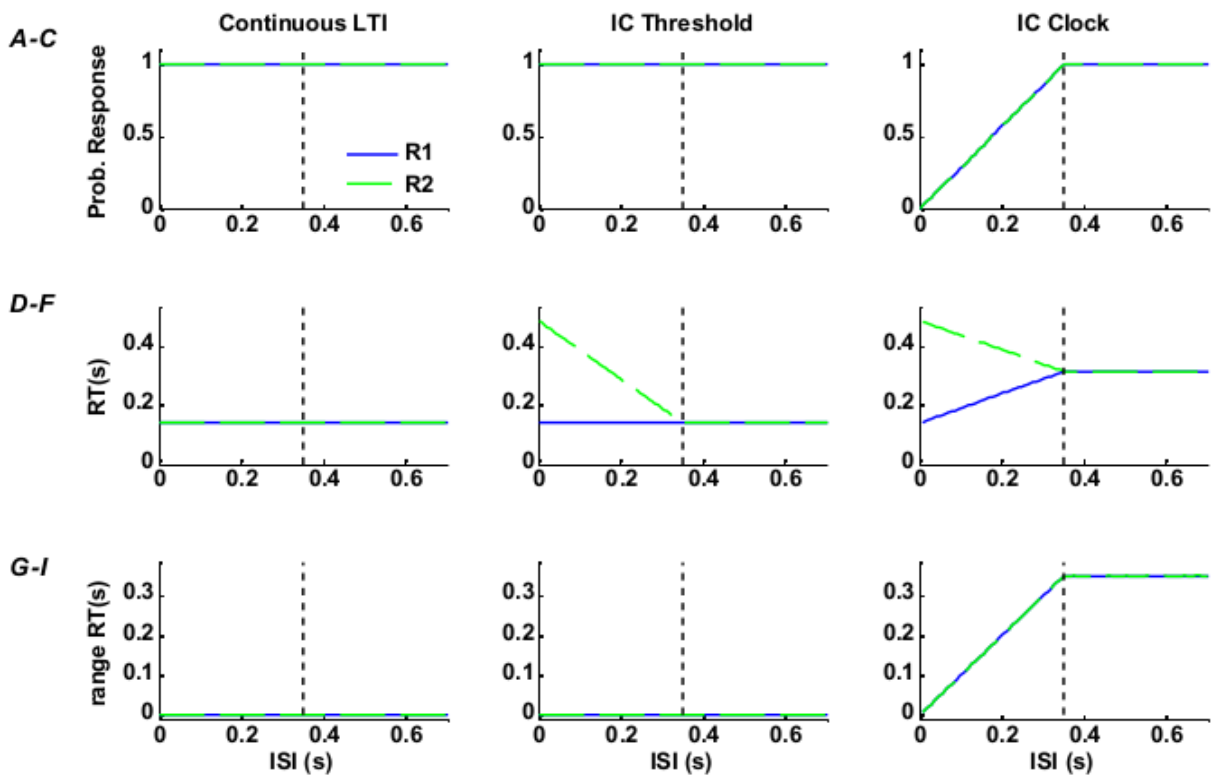

  \centering 
  \Fig{RSI_fig3}{1}
  \caption[Reconstruction of the set-point]{Reconstruction of the
    set-point. Predicted delays for varying inter-step interval
    For three models continuous linear time invariant, threshold
    triggered intermittent control and clock triggered intermittent
    control (cols 1--3 respectively) the following is shown as a
    function of inter-step interval (ISI):- A--C. The predicted
    probability of response D--F. The mean response delay G--I.  The
    range of response delays Response 1 and 2 (R1, R2) are solid and
    dashed respectively.  For these calculations the open loop
    interval ($\Delta_{ol}$) is 0.35s (vertical dashed line) and
    feedback time-delay (td) is 0.14s. [\citep{LorKamGolGaw12} Copyright \textcopyright 2012 the
    authors. Used with permission.]}
\label{fig:RSI_fig3}
\end{figure}
\begin{figure}[htbp]
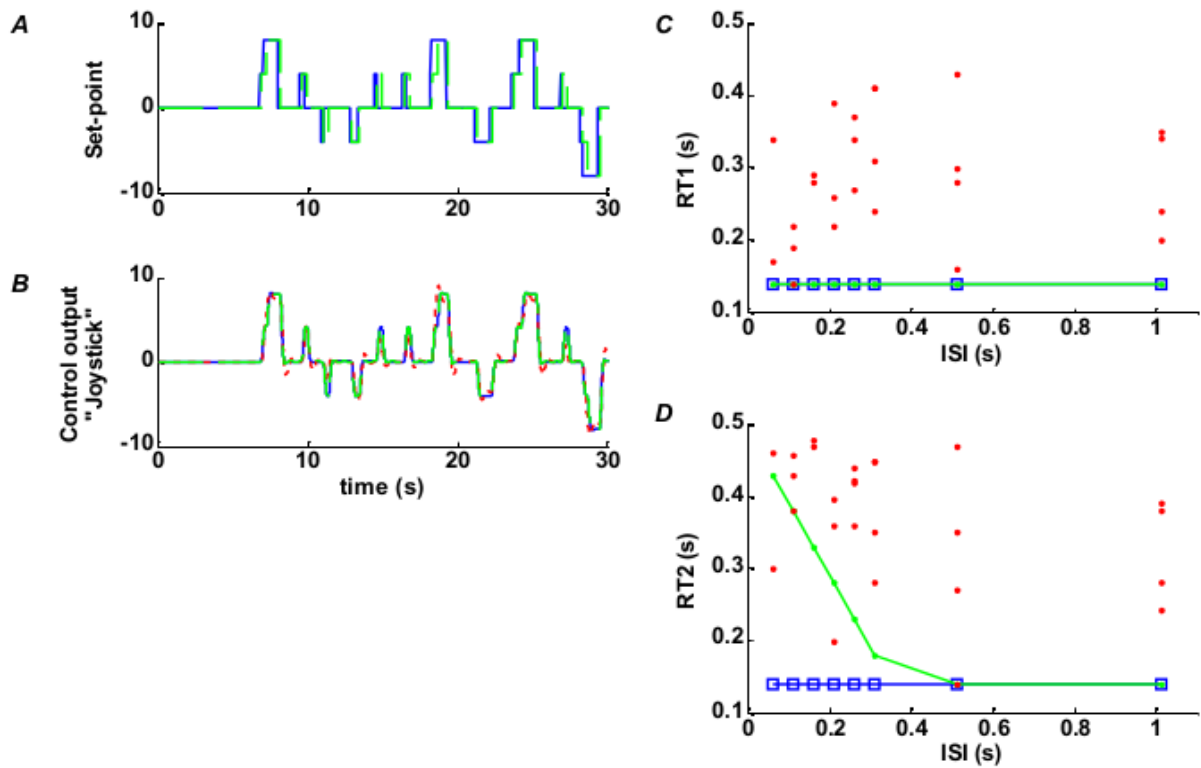

  \centering 
  \Fig{RSI_fig4}{1}
  \caption[Reconstruction of the set-point]{Reconstruction of the
    set-point. Representative Stage 1 Analysis A. Solid:
    Set-point sequence containing 8 inter-step-interval pairs with
    random direction (first 30s). Dashed: adjusted set-point from step
    1 analysis. After a double unidirectional step, set-point returns
    to zero before next pair.  B. Solid: Control output
    (ue). Dash-dotted: ARMA (LTI) fit to set-point.  Dashed: ARMA
    (LTI) fit to adjusted set-point.  C, D. Response times (RT1, RT2)
    respectively from each of three models v Inter-step-interval
    (ISI). Joined square: Continuous LTI.  Joined dot: Threshold
    intermittent control. Isolated dot: Clock intermittent control.
    The system is zero order. The open loop interval ($\Delta_{ol}$) is 0.35s
    and feedback time-delay ($\td$) is 0.14s. [\citep{LorKamGolGaw12} Copyright \textcopyright 2012 the
    authors. Used with permission.]}
  \label{fig:RSI_fig4}
\end{figure}
% \tbd{sec:iic:intermittency}{ Ian, please supply an extended caption
%    to explain Figure \ref{fig:RSI}}
%
Using the intermittent-equivalent setpoint of Sections
\ref{sec:equivalent_setpoint} and \ref{sec:ex:prp},
%(Sections 3.7, 4,2) 
we summarise a method to distinguish intermittent from continuous
control \citep{LorKamGolGaw12}. The identification experiment uses a
specially designed paired-step set-point sequence. The corresponding
data analysis uses a conventional ARMA model to relate the
theoretically derived equivalent set-point (of Section
\ref{sec:equivalent_setpoint}) to the control signal. The method
sequentially and iteratively adjusts the timing of the steps of this
equivalent set-point to optimise the linear time invariant fit. The
method has been verified using realistic simulation data and was found
to robustly distinguish not only between continuous and intermittent
control but also between event-driven intermittent and clock-driven
intermittent control \citep{LorKamGolGaw12}. This identification
method is applicable for machine and biological applications. For
application to humans the set-point sequence should be unpredictable
in the timing and direction of steps.  This method proceeds in three
stages. Stages 1 and 2 are independent of model assumptions and
quantify refractoriness, the key feature discriminating intermittent
from continuous control.

\subsubsection{Stage 1: Reconstruction of the set-point}
With reference to Figure \ref{fig:RSI_fig2}, this stage takes the known
set-point and control output signals and reconstructs the set-point
step times to form that sequence with a linear-time invariant response
which best matches the control output. This is implemented as an
optimisation process in which the fit of a general linear time series
model (zero-delay ARMA) is maximised by adjusting the trial set of
step times. The practical algorithmic steps are stated by
\citet{LorKamGolGaw12}.  The output from stage 1 is an estimate of the
time delay for each step stimulus.

\subsubsection{Stage 2: Statistical analysis of delays:}
Delays are classified according to step (1 or 2, named
reaction-time\footnote{In the physiological literature, ``delay'' is
  synonymous with ``reaction time''.} 1 (RT1) and reaction-time 2
(RT2) respectively) and inter-step-interval (ISI). A significant
difference in delay, RT2 v. RT1, is \emph{not} explained by a
linear-time-invariant model. The reaction time properties, or
refractoriness, is quantified by:
\begin{enumerate}
\item the size of ISI for which RT2 $>$ RT1. This indicates the temporal
  separation required to eliminate interference between successive
  steps and
\item the difference in delay (RT2 $-$ RT1).
\end{enumerate}

\subsubsection{Stage 3: Model based interpretation: }
\begin{figure}[htbp]
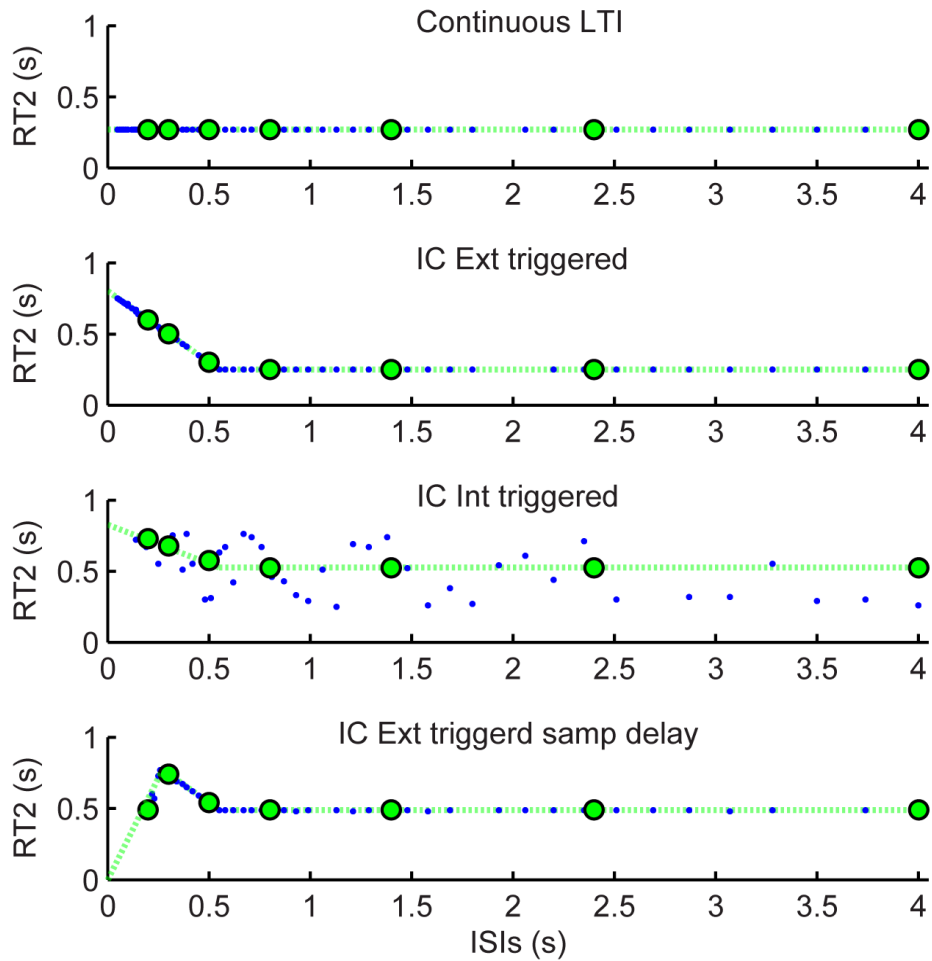

  \centering
  \Fig{FCN}{0.8}
  \caption[Model-based interpretation]{Model-based interpretation
    (stage 3). Parameter variants from the generalised IC model of
    figure~\ref{fig:MIC_arch} showing several possible relationships
    between RT2 and inter-step interval (ISI) indicative of serial
    ballistic (intermittent) and continuous control behaviour. The
    simulated system is zero order. The open-loop interval ($\Delta_{ol}$) is
    0.55s and feedback time delay ($\td$) is 0.25s. For four models: A)
    continuous LTI ($\Delta_{ol}=0$), B) externally-triggered intermittent
    control with a prediction error threshold, C) internally-triggered
    intermittent control (with zero prediction error threshold,
    triggered to saturation), and D) externally-trigger intermittent
    control supplemented with a sampling delay of 0.25s which is
    associated with the ISI at the maximum delay for RT2. The joined
    green circles represent the theoretical delays as a function of
    ISI which are confirmed by the model simulations (blue
    dots). [\citep{KamGawGolLakLor13} Copyright \textcopyright 2012 the
    authors. Used with permission.]}
  \label{fig:FCN}
\end{figure}
% \tbd{sec:iic:intermittency}{ Ian, please expand this (or Figure
%   captions) and refer to ALL relevant figures.}
For controllers following the generalised continuous (Figure
\ref{fig:MSS_arch}) and intermittent  (Figure
\ref{fig:MIC_arch}) structures, the probability of a response occurring, the mean
delay and the range of delays can be predicted for each
inter-step-interval (Figure \ref{fig:RSI_fig3} and Appendix C of
\citet{LorKamGolGaw12}). For a continuous controller (Figure
\ref{fig:MSS_arch}) all delays equal the model delay
($\td$). Intermittent control is distinguished from continuous control
by increased delays for RT2 v. RT1 for inter-step-intervals less than
the open-loop interval ($\Delta_{ol}$). Clock (zero threshold)
triggered intermittent control is distinguished from threshold
triggered intermittent control by the range of delays for RT1 and RT2
and by the increased mean delay for inter-step intervals greater than
the open-loop interval ($\Delta_{ol}$) (Figure 3). If the results of
Stage 1--2 analysis conform to these patterns (Figure
\ref{fig:RSI_fig3}), the open-loop interval ($\Delta_{ol}$) can be
estimated. Simulation also shows that the sampling delay
($\Delta_{s}$) can be identified from the ISI at which the delay RT2
is maximal (Fig \ref{fig:FCN})
\citep{KamGawGolLakLor13,KamGawGolLor13}.  Following verification by
simulation \citep{LorKamGolGaw12}, the method has been applied to
human visually guided pursuit tracking and to whole body pursuit
tracking. In both cases control has been shown to be intermittent
rather than continuous.

\subsection{Refractoriness in sustained manual control} 
\begin{figure}[htbp]
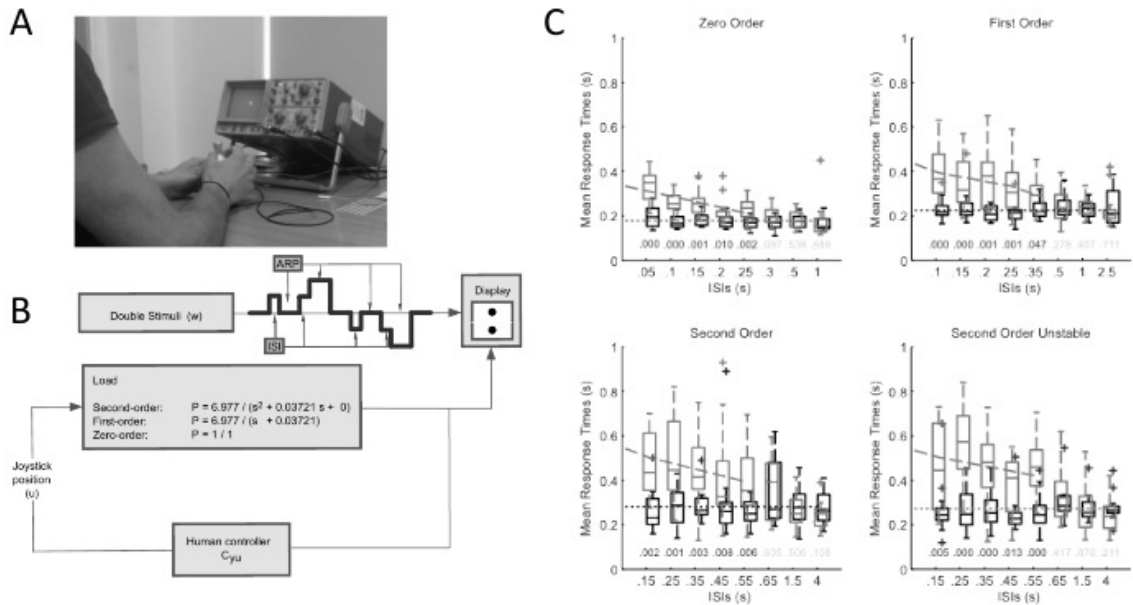

  \centering
  \Fig{ESSR_fig2}{1}
  \caption[Refractoriness in sustained manual control]{ Refractoriness
    in sustained manual control. A. Task setup. An oscilloscope showed
    real-time system output position as a small focused dot with
    negligible delay. Participants provided input to the system using
    a sensitive, uniaxial, contactless joystick. The system ran in
    Simulink Real-Time Windows Target within the MATLAB environment
    (Math-Works). B. Control system and experimental set
    up. Participants were provided with a tracking target in addition
    to system output.  The tracking signal was constructed from four
    possible patterns of step sequence (uni- and reversed directional
    step to the left or to the right). First and second stimuli are
    separated by an unpredictable inter step interval (ISI), patterns
    are separated by an unpredictable approximate recovery period
    (ARP). The participant was only aware of an unpredictable sequence
    of steps. C. Group results: The four panels: Zero Order, First
    Order, Second Order, Second Order Unstable show the inter
    participant mean first (RT1, black) and second (RT2, gray)
    response times against Inter step intervals (ISIs), p-values of
    the ANOVA's post hoc test are displayed above each ISI level
    (dark if significant, light if not). [ \citep{KamGawGolLor13}. Copyright \textcopyright 2013 the
    authors. Used with permission.]}
  \label{fig:ESSR_fig2}
\end{figure}
Using a uni-axial, sensitive, contactless joystick, participants were
asked to control four external systems (zero, first, second order
stable, second order unstable) using visual feedback to track as fast
and accurately as possible the target which changes position
discretely and unpredictably in time and direction (Figures
\ref{fig:ESSR_fig2}A\&B).
%(Fig. 2 ESSR A, B) 27.
For the zero, first and second order systems, joystick position
determines system output position, velocity and acceleration
respectively. The unstable second order system had a time-constant
equivalent to a standing human. Since the zero order system has no
dynamics requiring ongoing control, step changes in target produce
discrete responses i.e. sharp responses clearly separated from periods
of no response. The first and second order systems require sustained
ongoing control of the system output position: thus the step stimuli
test responsiveness during ongoing control. The thirteen participants
showed evidence of substantial open loop interval (refractoriness)
which increased with system order (0.2 to 0.5 s,
\ref{fig:ESSR_fig2}~C). For first and second order systems,
participants showed evidence of a sampling delay (0.2-0.25 s,
\ref{fig:ESSR_fig2}~C). This evidence of refractoriness discriminates
against continuous control.

\subsection{Refractoriness in whole body control }
\begin{figure}[htbp]
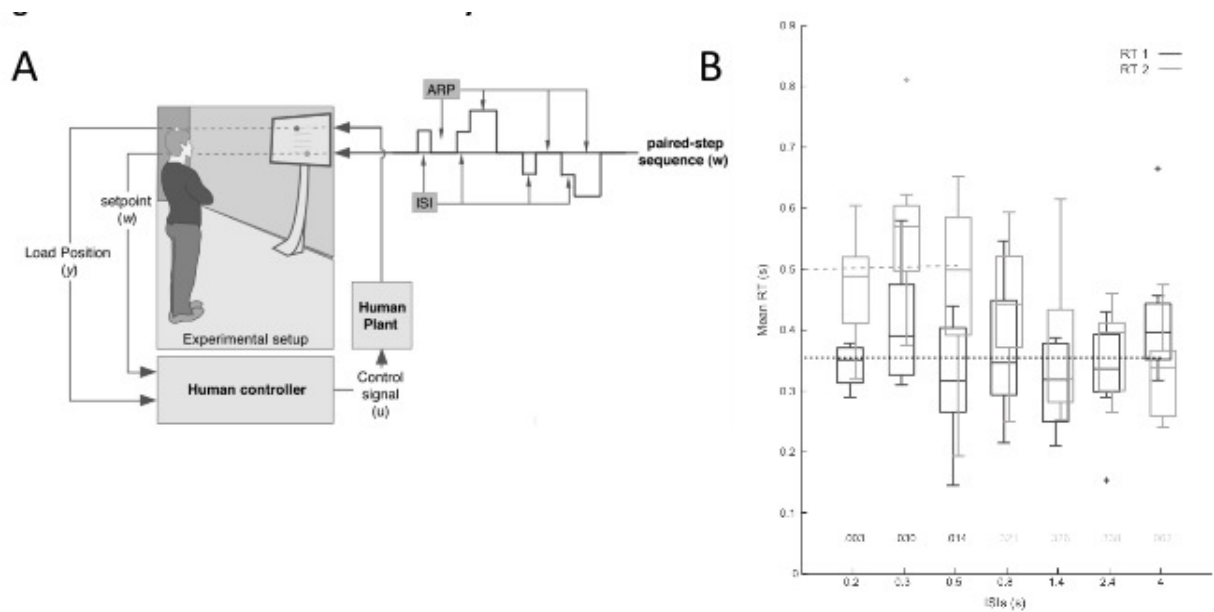

  \centering
  \Fig{ESSR_fig4}{1}
  \caption[Refractoriness in whole body control.]{Refractoriness in
    whole body control A. The participant receives visual feed-back of
    the Anterior-Posterior head position through a dot presented on an
    LCD screen mounted on a trolley. Without moving their feet,
    participants were asked to track the position of a second dot
    displayed on the screen. The four possible step sequence
    combinations (uni- and reversed-directional step up or down) of
    the pursuit target are illustrated by the solid line. First and
    second stimuli are separated by an inter-step interval (ISI). The
    participant experiences an unpredictable sequence of steps.
    B. Group results. Figure shows the inter-participant mean RT1
    (black) and RT2 (gray) against ISI combined across the eight
    participants. The P-values of the ANOVA's post-hoc test are
    display above each ISI level (black if <0.05, gray if not). The
    dotted line shows the mean RT1, the dashed line shows the
    regression linear fit between (interfered) RT2 and ISIs.
    [\citep{KamGawGolLakLor13}. Copyright \textcopyright 2013 the
    authors. Used with permission.]}
  \label{fig:ESSR_fig4}
\end{figure}
Control of the hand muscles may be more refined, specialised and more
intentional than control of the muscles serving the legs and
trunk \citep{KamGawGolLakLor13}. Using online visual feedback ($<$~100 ms delay) of a marker on
the head, participants were asked to track as fast and accurately as
possible a target which changes position discretely and unpredictably
in time and direction (Figure \ref{fig:ESSR_fig4}A). This required
head movements of 2cm along the anterior-posterior axis and while
participants were instructed not to move their feet, no other
constraints or strategies were requested. The eight participants
showed evidence of substantial open loop interval (refractoriness)
(~0.5 s) and a sampling delay (~0.3 s) (\ref{fig:ESSR_fig4}~B).
%%Fig 4B ESSR).
This result extends the evidence of intermittent control from
sustained manual control to integrated intentional control of the
whole body.

\section{Adaptive intermittent control}
\label{sec:aic}
% \subsection{Introduction}
% \label{sec:aic:intro}
%% General 
The purpose of feedback control is to reduce uncertainty
\citep{Hor63,Jac74}.  Feedback control using fixed controllers can be
very successful in this role as long as sufficient bandwidth is
available and the system does not contain time delays or other
non-minimum phase elements
\citep{Hor63,GooGraSal01,SkoPos96}. However, when the system and
actuators do contain such elements, an adaptive controller can be used
to reduce uncertainly and thus, in time, improve controller
performance.
By its nature, intermittent control is a low-bandwidth controller and
so adaptation is particularly appropriate in this context. Conversely,
intermittent control frees computing resources that can be used for
this purpose.

%% Engineering AC
As discussed in the textbooks by \citet{GooSin84} and by
\citet{AstWit89} adaptive control of engineering systems is well
established. Perhaps the simplest approach is to combine real-time
recursive parameter estimation with a simple controller design method
to give a so called ``self-tuning'' strategy
\citep{AstWit73,ClaGaw75a,ClaGaw79,ClaGaw81,AstWit89}. However, such
an approach ignores two things: the controller is based on initially
incorrect parameter estimates and the quality of the parameter
estimation is dependent on the controller properties. This can be
formalised using the concepts of \emph{caution} whereby the adaptive
controller takes account of parameter uncertainty and \emph{probing}
whereby the controller explicitly excites the system to improve
parameter estimation \citep{JacPat72,Bar81}. Adaptive controllers
which explicitly and jointly optimise controller performance and
parameter estimation have been called \emph{dual} controllers
\citep{Fel60,BarTse74}. Except in simple cases, for example that of
\citet{AstHel86}, the solution to the dual control problem is
impractically complex.

%% Bio-inspired
Since \citet{Wie65} developed the idea of cybernetics, there has been
a strong interest in applying both biologically-inspired and
engineering-inspired ideas to adaptive control in both humans and
machines; ideas arising from the biological and engineering
inspired fields have been combined in various ways.

%% Reinforcement
One such thread is \emph{reinforcement learning}
\citep{SutBar98} which continues to be developed both theoretically
and through applications \citep{KhaHerLew12}. It can be argued that
that ``reinforcement learning is direct adaptive optimal control''
\citep{SutBarWil92}.
Artificial neural networks have been applied to engineering control
systems: see the survey of \citet{HunSbaZbiGaw92} and numerous textbooks
\citep{MilSutWer90,ZbiHun95,KalHunZbi97a}. Recent work is described by
\citet{VraLew09}.
Again, there are links between Artificial Neural Network (ANN) methods
such as back-propagation and engineering parameter estimation
\citep{GawSba90}.

%%SLAM
Field robotics makes use of the concept of \emph{Simultaneous Location
  and Mapping} (SLAM) \citep{DurBai06,BaiDur06}. Roughly speaking,
location corresponds to state estimation and mapping to parameter
estimation and therefore concepts and techniques from SLAM are
appropriate to adaptive (intermittent) control. In particular, the
Extended Kalman Filter (EKF) and the Unscented Kalman Filter (UKF)
\citep{JulUhlDur00,Sch12} provide the basis for SLAM and hence adaptive
control.

In this section the simplest \emph{continuous-time} self-tuning approach
\citep{Gaw82b,Gaw87c,Gaw90} is used. As indicated in the Examples of
Sections \ref{sec:ex:balance} and \ref{sec:ex:reach}, this simple
approach has interesting behaviours; neverthless, it would be
interesting to investigate  more sophisticated
approaches based on, for example, the EKF and UKF.

%% Continuous v. discrete.

%% Need for IC

%% Adaptive citations
% \citet{AstWit73}
% \citet{ClaGaw75a}
% \citet{GarAnt12}

% %% Unscented citations
% \citet{CeaGoo11}
% \citet{ChoPosNes12}
% \citet{ShiSau08}
% \citet{ManKupRen12}
% \citet{Dau05}
% \citet{Eve09}
% \citet{JulUhlDur00}
% \citet{VosTimKur04}
% \citet{LarHauPed11}

\subsection{System Model}
\label{sec:aic:model}
Parameter estimation is much simplified if the system can be  transformed
into \emph{linear-in-the-parameters} (LIP) form; the resultant model
can be viewed as a \emph{non-minimal state-space} (NMSS) representation of the
system. The NMSS approach is given in discrete-time form by
\citet{YouBehWanCho87} and continuous-time form by \citet{TayChoYou98}.
Although a purely state-space approach to the NMSS representation is
possible \citep{GawWanYou07}, a polynomial approach is simpler and is
presented here.

The linear-time invariant system considered in this Chapter are given in
Laplace transform terms by:
\begin{equation}\label{eq:sys_lt}
  \ys = \frac{b(s)}{a(s)} \lb \us + \frac{b_\xi(s)}{a_\xi(s)}\zsu \rb
  + \frac{\ddd(s)}{a(s)a_\xi(s)}
\end{equation}
where $\ys$, $\us$, and $\zsu$ are the Laplace transformed system
output, control input and input disturbance. $\frac{b(s)}{a(s)}$ is
the transfer function relating $\ys$ and $\us$ and
$\frac{b_\xi(s)}{a_\xi(s)}$ provides a transfer function model of the
input disturbance. It is assumed that both transfer functions are
strictly proper. The overall system initial conditions are represented
by the polynomial $\ddd(s)$\footnote{Transfer function representations of
  continuous-time systems and initial conditions are discussed, for
  example, by \citet[Ch. 4]{GooGraSal01} }.
The polynomials $a(s)$, $b(s)$ and $\ddd(s)$ are of the form:
\begin{align}
  a(s) &= a_0s^n + a_1 s^{n-1} + \dots + a_n\label{eq:a(s)}\\
  b(s) &= b_1 s^{n-1} + \dots + b_n\label{eq:b(s)}\\
  \ddd(s) &= \ddd_1 s^{n-1} + \dots + \ddd_n\label{eq:dp(s)}\\
  a_\xi(s) &= \alpha_0s^{n_\xi} + \alpha_1 s^{n_\xi-1} + \dots + \alpha_{n_\xi}\label{eq:axi(s)}\\
  b_\xi(s) &= \beta_0s^{n_\xi} + \beta_1 s^{n_\xi-1} + \dots + \beta_{n_\xi}\label{eq:bxi(s)}
\end{align}

Finally, defining the Hurwitz polynomial $c(s)$ as:
\begin{align}
  c(s) &= c_0s^{N} + c_1 s^{N-1} + \dots + c_N  \label{eq:c(s)}\\
\text{where } N &= n + n_\xi + 1 \label{eq:NN}
\end{align}
Equation (\ref{eq:sys_lt}) may be rewritten as:
\begin{equation}\label{eq:sys_lt_c}
  \frac{a(s)a_\xi(s)}{c(s)}\ys = 
    \frac{b(s)a_\xi(s)}{c(s)} \us 
    + \frac{b(s)b_\xi(s)}{c(s)}\zsu 
    + \frac{\ddd(s)}{c(s)}
\end{equation}

For the purposes of this Chapter, the polynomials $a_\xi(s)$, $b_\xi(s)$ are
defined as:
\begin{align}
  a_\xi(s) &= s \label{eq:a_xi}\\
  b_\xi(s) &= 1 \label{eq:b_xi}
\end{align}
With this choice, Equation (\ref{eq:sys_lt_c}) simplifies to
\begin{equation}\label{eq:sys_lt_c_0}
  \frac{sa(s)}{c(s)}\ys = 
    \frac{sb(s)}{c(s)} \us 
    + \frac{b(s)}{c(s)}\zsu 
    + \frac{d(s)}{c(s)}
\end{equation}

In the special case that the input disturbance is a jump to a constant
value $d_\xi$ at time $t=0^+$, then this can be modelled using Equations
(\ref{eq:a_xi}) and (\ref{eq:b_xi}) and
\begin{align}
  \xi(t) &= d_\xi \delta(t)\\
  \text{and } \zsu &= d_\xi   \label{eq:xi_u}
\end{align}
where $\delta(t-t_k$ is the Dirac delta function.

Equation (\ref{eq:sys_lt_c}) then becomes:
\begin{align}
  \frac{sa(s)}{c(s)}\ys &= 
    \frac{sb(s)}{c(s)} \us 
    + \frac{d(s)}{c(s)} \label{eq:sys_lt_c_1}\\
    \text{where } d(s) &= \ddd(s) + d_\xi b(s)
\end{align}

% In the special case that the input disturbance is piecewise constant
% with jumps of ampliture $d_k$ at time $t_k$, the time domain input
% disturbance $\ztu$ is given by:
% \begin{equation}
%   \label{eq:ztu}
%   \ztu = \sum_{k=1}^{n_k} d_k \delta(t-t_k)
% \end{equation}
% where $\delta(t-t_k$ is the Dirac delta function at time $t_k$.
% and the corresponding Laplace transform is
% \begin{equation}
%   \label{eq:zsu}
%   \zsu = \sum_{k=1}^{n_k} d_k e^{-st_k}
% \end{equation}

Equation (\ref{eq:sys_lt_c_0}) can be rewritten in non-minimal
state-space form as:
\begin{align}
  \ddt{}\phi_y(t) &= A_s \phi_y(t) - B_s y(t) \label{eq:svf_y}\\
  \ddt{}\phi_u(t) &= A_s \phi_u(t) + B_s u(t) \label{eq:svf_u}\\
  \ddt{}\phi_{ic}(t) &= A_s \phi_{ic}(t),\; 
  \phi_{ic}(0) = \phi_{ic0}  \label{eq:svf_ic}\\
  \text{where } A_s & = \companc{-c}\\
  \text{and } B_s &=  \begin{bmatrix}
    1\\0\\ \dots\\0
  \end{bmatrix} 
\end{align}

It follows that:
\begin{align}
  \eee(t) &= \theta^T\phi(t)\label{eq:ee}\\
  \text{where } \theta &=
  \begin{bmatrix}
    \ba \\ \bb \\ \bd
  \end{bmatrix} 
  \text{ and } \phi(t) =
  \begin{bmatrix}
    \phi_y \\ \phi_u \\ \phi_{ic}
  \end{bmatrix}\\
  \text{and }\ba &=
  \begin{bmatrix}
    a_0 & a_1 & \dots & a_n & \bz
  \end{bmatrix}^T\label{eq:ba}\\
  \bb &=
  \begin{bmatrix}
    0 & b_1 & \dots & b_n & \bz
  \end{bmatrix}^T\label{eq:bb}\\
  \bd &=
  \begin{bmatrix}
     0 & d_1 & \dots & d_n & \dots & d_{n_c}
  \end{bmatrix}^T\label{eq:bd}
\end{align}
\subsection{Continuous-time Parameter Estimation}
\label{sec:aic:cte}
As discussed by \citet{You81a}, \citet{Gaw82b,Gaw87c},
\citet{UnbRao87,UnbRao90} and \cite{GarWan08}, least-squares parameter estimation can 
performed in the continuous-time domain (as opposed to the more usual
discrete-time domain as described, for example, by \citet{Lju99}). A
brief outline of the method used in the following examples is given in
this section.
\begin{align}
  \eh(t) &= \thh^T\phi(t)\label{eq:e}\\
  J(\thh) &= \frac{1}{2}\int_0^t e^{\lambda (t-t^\prime)} \eh(t^\prime)^2 dt^\prime\notag \\
  &= \thh^T \bS(t) \thh\notag \\
  &= \thh^T \bS_{uu}(t) \thh 
  + \thh^T \bS_{uk}(t) \thk 
  + \thk^T \bS_{kk}(t) \thk \\
  \text{where } \bS(t) &= \int_0^t 
  e^{\lambda (t-t^\prime)}
  \phi(t^\prime) \phi^T(t^\prime)dt^\prime\label{eq:S}\\
\end{align}
and the symmetrical matrix $\bS(t)$ has been partitioned as:
\begin{equation}
  \label{eq:S_part}
  \bS(t) = 
  \begin{bmatrix}
    \bS_{uu}(t) & \bS_{uk}(t)\\
    \bS^T_{uk}(t) & \bS_{kk}(t)
  \end{bmatrix}
\end{equation}

Differentiating the cost function $J$ with respect to the 
vector of unknown parameters $\thh$ gives:
\begin{equation}
  \label{eq:dJ}
  \frac{dJ}{d\thh} =  \bS_{uu}(t) \thh + \bS_{uk}(t) \thk
\end{equation}
Setting the derivative to zero gives the optimal solution:
\begin{equation}
  \label{eq:J_opt}
  \thh(t) = -\bS_{uu}^{-1}(t)\bS_{uk}(t) \thk
\end{equation}

Differentiating $\bS$ (\ref{eq:S}) with respect to time gives
\begin{equation}
  \label{eq:dS}
  \frac{d\bS}{dt} + \lambda\bS(t) = \phi(t)\phi^T(t)
\end{equation}

\subsection{Intermittent Parameter Estimation}
\label{sec:aic:ipe}
The incremental information matrix $\St_i$ from the $i$th intermittent
interval is defined as
\begin{equation}\label{eq:St}
  \St_i = \int_{t_{i-1}}^{t_i}
\phi(t^\prime) \phi^T(t^\prime)dt^\prime
\end{equation}
Equation (\ref{eq:St}) may be implemented using the differential
equation (\ref{eq:dS}) with zero initial condition at time $t_{i-1}$.
The intermittent information matrix $\bS_i$ at  the $i$th intermittent
interval is defined as:
\begin{equation}
  \label{eq:S_ic}
  \bS_i =   \lambda_{ic} \bS_{i-1} + \St_i 
\end{equation}
Partitioning  $\bS_i$ as Equation (\ref{eq:S_part}) gives the
parameter estimate of Equation (\ref{eq:J_opt}).

If there is a disturbance characterised by Equations (\ref{eq:a_xi}),
(\ref{eq:b_xi}) and (\ref{eq:xi_u}), the parameters corresponding to
$d(s)$ jump when the disturbance jumps. As such a jump will give rise
to an event, a new set of $d$ parameters should be estimated; this is
achieved by adding a diagonal matrix to the elements of $\bS_i$
corresponding to $d(s)$.

%%\tbd{sec:aic:ipe}{Much more explanation and equations required here!}

\section{Examples: adaptive human balance}
\label{sec:ex:balance}
%% Balance figures
\begin{figure}[htbp]
  \centering
  \SubFig{bal_adapt_0_db_0_y}{Output $y$. No adaption $\Delta b_0 = 0$}{0.47}
  \SubFig{bal_adapt_0_db_0_PP}{Phase-plane. No adaption $\Delta b_0 = 0$}{0.47}\\
  \SubFig{bal_adapt_0_db_1_y}{Output $y$. No adaption $\Delta b_0 = 0.1$}{0.47}
  \SubFig{bal_adapt_0_db_1_PP}{Phase-plane. No adaption $\Delta b_0 = 0.1$}{0.47}\\
  \SubFig{bal_adapt_1_db_1_y}{Output $y$. Adaption $\Delta b_0 = 0.1$}{0.47}
  \SubFig{bal_adapt_1_db_1_PP}{Phase-plane. Adaption $\Delta b_0 = 0.1$}{0.47}
  \caption[Adaptive Balance control.]{Adaptive Balance control. (a)\&(b) Correct parameters, no
    adaptation. (c)\&(d) Incorrect parameters, no adaptation.
    (e)\&(f) Incorrect parameters, with adaptation -- the initial
    behaviour corresponds to (c)\&(d) and the final behaviour
    corresponds to  (a)\&(b). For clarity, lines are coloured grey for
  $t<60$ and black for $t\ge60$.}
  \label{fig:balance}
\end{figure}

\begin{figure}[htbp]
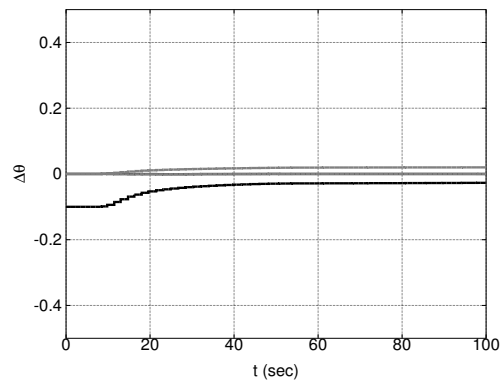

  \centering
  \Fig{bal_adapt_1_db_1_theta}{0.44}
  \caption[Adaptive Balance control: estimated Parameter.]  {Adaptive
    Balance control: estimated Parameter. The parameter estimate
    errors ($\hat{a}-a$ and $\hat{b}-b$) become smaller as time
    increases.}
  \label{fig:est_balance}
\end{figure}
As discussed by \citet{GawLorGolLak14}, it can be argued that the
human balance control system generates ballistic control trajectories
that attempt to place the unstable system at equilibrium; this leads
to homoclinic orbits \citep{HirSmaDev12}. However, such behaviour is
dependent on a good internal model. This section looks at the same
ballistic balance control system as that of \citet{GawLorGolLak14} but
in the context of parameter adaptation.

The controlled system is given by the transfer function:
\begin{equation}
  \label{eq:G(s)}
  G(s) = \frac{b}{s^2 + a}
\end{equation}
The actual system parameters are:
\begin{align}
  a &= -1 \label{eq:a}\\
  b &= 1.1 \label{eq:b}
\end{align}
The parameters $a$ and $b$ are estimated using the intermittent
parameter estimation method of Section \ref{sec:aic:ipe} with initial
values:
\begin{align}
  \hat{a} &= -1 \label{eq:a0}\\
  \hat{b} &= 1 \label{eq:b0}
\end{align}

Figures \ref{subfig:bal_adapt_0_db_0_y} and
\ref{subfig:bal_adapt_0_db_0_PP} show the non-adaptive controller with
correct parameters of Equations (\ref{eq:a}) and (\ref{eq:b}); the
behaviour approximates that of the ideal ballistic controller.

Figures \ref{subfig:bal_adapt_0_db_1_y} and
\ref{subfig:bal_adapt_0_db_1_PP} show the non-adaptive controller with
the incorrect parameters of Equations (\ref{eq:a0}) and (\ref{eq:b0});
the behaviour is now a limit cycle.

Figures \ref{subfig:bal_adapt_1_db_1_y} and
\ref{subfig:bal_adapt_1_db_1_PP} shows the adaptive controller with
the initial incorrect parameters of Equations (\ref{eq:a0}) and
(\ref{eq:b0}). Initially, the behaviour corresponds to that of Figures
\ref{subfig:bal_adapt_0_db_0_y} and \ref{subfig:bal_adapt_0_db_0_PP};
but after about 50sec the behaviour corresponds to that of Figures
\ref{subfig:bal_adapt_0_db_1_y} and \ref{subfig:bal_adapt_0_db_1_PP}.
The corresponding parameter estimate errors ($\hat{a}-a$ and
$\hat{b}-b$) are given in Figure \ref{fig:est_balance}.

\section{Examples: adaptive human reaching}
\label{sec:ex:reach}
\begin{figure}[htbp]
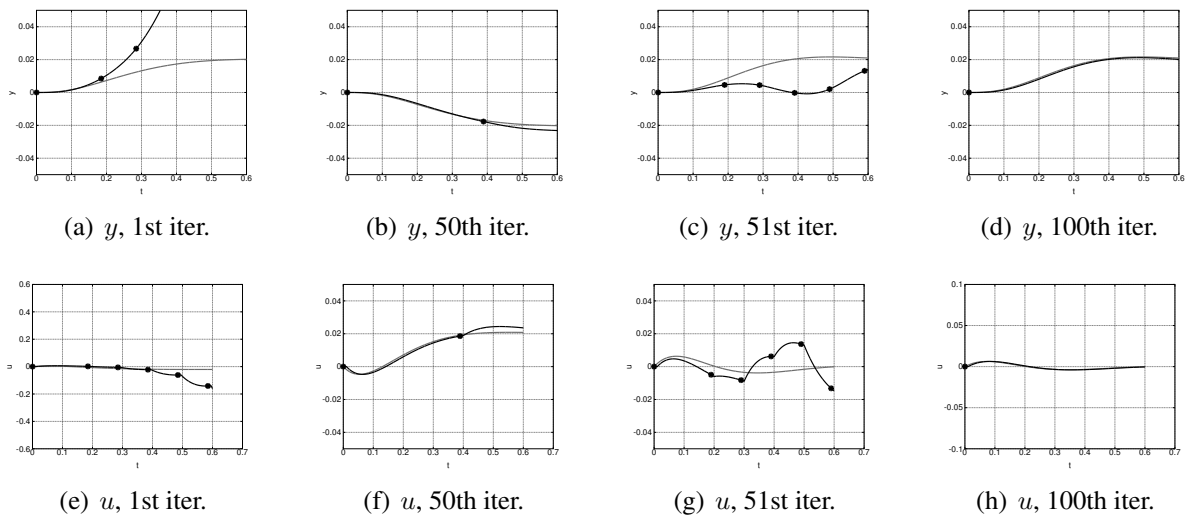

  \centering
  \SubFig{arm2_1_y}{$y$, 1st iter.}{0.22}
  \SubFig{arm2_50_y}{$y$, 50th iter.}{0.22}
  \SubFig{arm2_51_y}{$y$, 51st iter.}{0.22}
  \SubFig{arm2_100_y}{$y$, 100th iter.}{0.22}\\
  \SubFig{arm2_1_u}{$u$, 1st iter.}{0.22}
  \SubFig{arm2_50_u}{$u$, 50th iter.}{0.22}
  \SubFig{arm2_51_u}{$u$, 51st iter.}{0.22}
 \SubFig{arm2_100_u}{$u$, 100th iter.}{0.22}
  \caption[Reaching in a force-field: transverse position.]{Reaching
    in a force-field: transverse position. The additional transverse
    force field is applied throughout, but the initial parameters
    correspond to zero force field. The sample instants are denoted by
    the $\bullet$ symbol. The behaviour improves, and the intermittent
    interval increases, from iteration 1 to iteration 50.}
  \label{fig:iterative_yu_1}
\end{figure}

% \begin{figure}[htbp]
%   \centering
%   \SubFig{arm2_51_y}{$y$, 51st iter.}{0.22}
%   \SubFig{arm2_51_u}{$u$, 51st iter.}{0.22}
%   \SubFig{arm2_100_y}{$y$, 100th iter.}{0.22}
%   \SubFig{arm2_100_u}{$u$, 100th iter.}{0.22}
%   \caption[Reaching in a force-field: transverse position.]{Reaching
%     in a force-field: transverse position. The sample instants are
%     denoted by the $\bullet$ symbol. The additional force field is
%     removed immediately after iteration 50. The behaviour improves,
%     and the intermittent interval increases, from iteration 51 to
%     iteration 100.}
%   \label{fig:iterative_yu_2}
% \end{figure}

%\newpage
\begin{figure}[htbp]
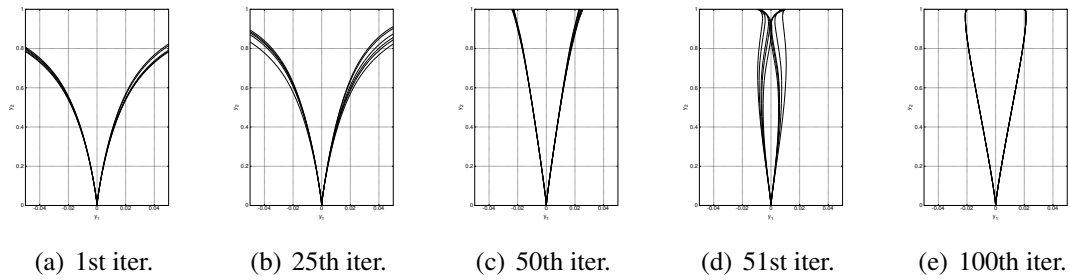

  \centering
  \SubFig{arm2_1_yy}{1st iter.}{0.15}
  \SubFig{arm2_20_yy}{25th iter.}{0.15}
  \SubFig{arm2_50_yy}{50th iter.}{0.15}
  \SubFig{arm2_51_yy}{51st iter.}{0.15}
  \SubFig{arm2_100_yy}{100th iter.}{0.15}
  \caption[Reaching in a force-field.]{Reaching in a force-field. The
    data from Figure \ref{fig:iterative_yu_1} are re-plotted against longitudinal
    position.}
  \label{fig:iterative_2}
\end{figure}

\begin{figure}[htbp]
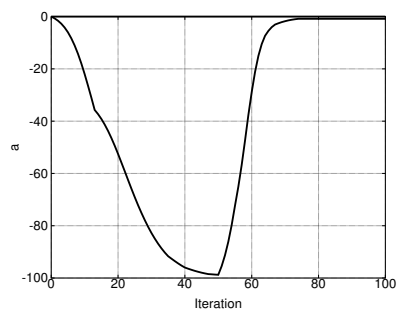

  \centering
  \SubFig{arm2_a}{$\hat{a}$}{0.35}  
  \caption[Reaching in a force-field: parameters.]{Reaching in a
    force-field: parameters. The transverse force field parameter is
    $a=-100$ for $0 \le \text{iteration} \le 50$ and $a=0$ for
    $\text{iteration} > 50$. The initial estimate is $\hat{a}=0$. }
  \label{fig:iterative_2_params}
\end{figure}
Repetitive reaching and pointing has been examined by a number of
authors including \citet{ShaMus94} (see also \citet{ShaMus12}),
\citet{BurTeeMar06} and \citet{TeeFraKaw10}. An iterative learning
control explanation of these results is given by
\citet{ZhoOetTanBurMar12}.

As discussed by \citet{BriThaAll06}, ``iterative learning control (ILC)
is based on the notion that the performance of a system that executes
the same task multiple times can be improved by learning from previous
executions (trials, iterations, passes).''. A number of survey papers
are available, including those of \citet{BriThaAll06},
\citet{AhnCheMoo07} and \citet{WanFurDoy09}, as well as a book
\citet{XuTan03}.
ILC is closely related to repetitive control \citep{CuiZhaZhu04} and
to multi-pass control \citep{Edw74,Owe77}.

This example shows how the intermittent parameter estimation method of
Section \ref{sec:aic:ipe} can be used in the context of iterative
learning.

The system similar to that described in Section IV, case 3 of the paper by
\citet{ZhoOetTanBurMar12} was used. 
The lateral motion of the arm in the force field was described the
transfer function of Equation (\ref{eq:G(s)}) with
\begin{align}
  a &=
  \begin{cases}
    -100 & i \le 50\\
    0 & i > 50
  \end{cases}
 \label{eq:a_arm}\\
  b &= 100 \label{eq:b_arm}
\end{align}
The lateral target position was randomly set to $\pm 0.01$m.

The parameters $a$ and $b$ are estimated using the intermittent
parameter estimation method of Section \ref{sec:aic:ipe} with initial
values:
\begin{align}
  \hat{a} &= 0 \label{eq:a0_arm}\\
  \hat{b} &= 200 \label{eq:b0_arm}
\end{align}

Figure \ref{fig:iterative_yu_1} shows
the system output (transverse position) $y$ and control input $u$
for five of the iterations; the sample instants are denoted by the
$\bullet$ symbol and the ideal trajectory by the grey line.
The initial behaviour (Figures \ref{subfig:arm2_1_y} and
\ref{subfig:arm2_1_u}) is unstable and sampling occurs at the minimum
interval of 100ms the behaviour at the 50th iteration (Figures \ref{subfig:arm2_50_y} and
\ref{subfig:arm2_50_u})just before the
parameter change is close to ideal even though the trajectory is open
loop for nearly 400ms. 
The behaviour at the 51st iteration (Figures \ref{subfig:arm2_51_y} and
\ref{subfig:arm2_51_u}) just after the parameter change is again poor
(although stable) but has become ideal and open loop by iteration 100
(Figures \ref{subfig:arm2_100_y} and \ref{subfig:arm2_100_u}).
The data is replotted in Figure \ref{fig:iterative_2} to show the
transverse position $y$ plotted against longitudinal position.

Figure \ref{fig:iterative_2_params} shows the evolution of the
estimated parameters with iteration number.

\section{Conclusion}
\label{sec:conc}
\begin{itemize}
\item This chapter has given an overview of the current
  state-of-the-art of the event-driven Intermittent Control 
  discussed, in the context of physiological systems, by
  \citet{GawLorLakGol11}.  In particular, Intermittent Control has
  been shown to provide a basis for the human control systems
  associated with balance and motion control.
%From Ian
\item Intermittent control arose in the context of applying control to
  systems and constraints which change through time
  \citep{RonArsGaw99}.  The intermittent control solution allows slow
  optimisation to occur concurrently with a fast control action.
  Adaptation is intrinsic to intermittent control and yet the formal
  relationship between adaptive and intermittent control remains to be
  established.  Some results of experiments with human subjects
  reported by
  \citet{LorGolLakGaw10}, together with the simulations of Sections
  \ref{sec:ex:balance} and \ref{sec:ex:reach}, support the intuition
  that the intermittent interval somehow simplifies the complexities
  of dual control .  A future challenge is to provide a theoretical
  basis formally linking intermittent and adaptive control.  This
  basis would extend applicability of time varying control and would
  enhance investigation of biological controllers which are adaptive
  by nature.
\item It is an interesting question as to where the event-driven
  intermittent control algorithm lies in the human nervous system. IC
  provides time within the feedback loop to use the current state to
  select new motor responses (control structure, law, goal,
  constraints). This facility provides competitive advantage in
  performance, adaptation and survival and is thus likely to operate
  through neural structures which are evolutionarily old as well as
  new \citep{Bre11}.  Refractoriness in humans is associated with
  a-modal response selection rather than sensory processing or motor
  execution \citep{PasJoh98}.  This function suggests plausible
  locations within premotor regions and within the slow
  striatal-prefrontal gating loops
  \citep{JiaKan03,DuxIvaAspMar06,HouBasFanFisFra07,Sei10,BatBuiCamFerLac14,LorKamLakGolGaw14}.
\item It seems plausible that Intermittent Control has
  applications within a broader biomedical context. Some possible
  areas are:
%
% \tbd{sec:conc}{Henrik to write a paragraph for the conclusion on
%     applications to rehabilitation?}
\begin{description}
\item [Rehabilitation] practice, following neuromuscular disease such
  as stroke and spinal cord injury, often uses passive closed loop
  learning in which movement is externally imposed by therapists or
  assistive technology (e.g. robotic assisted rehabilitation
  \citep{HuaKra09}). \citet{LorGolLakGaw10} have shown that adaptation
  to parameter changes during human visual-manual control can be
  facilitated by using an explicitly intermittent control strategy.
  For successful learning, active user input should excite the system,
  allowing learning from the observed intermittent open-loop behaviour
  \citep{LorGolLakGaw10}.

\item [Cellular control systems] in general and gene regulatory
  networks in particular seem to have a intermittent nature
  \citep{AlbBurSpeLau08,BalOudCol11,LiuJia12}. It would be interesting
  to examine whether the intermittent control approaches of this paper
  are relevant in the context of cellular control systems.
\end{description}
\item The particular Intermittent control algorithm discussed within
  this Chapter has roots and applications in control engineering
  \citep{RonArsGaw99,GawWan06,GawWan09,GawNeiWag12} and it is hoped
  that this chapter will lead to further cross-fertilisation of
  physiological and engineering research. Some possible areas are:
\begin{description}
\item [Decentralised control] \citep{SanVarAthSaf78,BakPap12} is a pragmatic
approach to the control of large-scale systems where, for reasons of
cost, convenience and reliability it is not possible to control the
entire system by a single centralised controller.
Fundamental control-theoretic principles arising from decentralised
control have been available for some time
\citep{Cle79,AndCle81,GonAld92}.
More recently, following the implementation of decentralised control
using networked control systems \citep{MoyTil07}, attention has focused
on the interaction of communication and control theory
\citep{BaiAnt07,NaiFagZamEva07} and fundamental results have appeared
\citep{NaiEva03,NaiEvaMarMor04,HesNagYon07}. It would be interesting
to apply the physiologically inspired approaches of this Chapter to
decentralised control as well as to reconsider Intermittent Control
in the context of decentralised control systems.
\item [Networked control systems] lead to the ``sampling period
jitter problem'' \citep{Sal07,MoyTil07} where uncertainties in
transmission time lead to unpredictable non-uniform sampling and
stability issues \citep{CloWouHeeNij09}. A number of authors have
suggested that performance may be improved by replacing the standard
zero-order hold by a generalised hold \citep{FeuGoo96,Sal05,Sal07} or
using a dynamical model of the system to interpolate between samples
\citep{ZhiMid03,MonAnt03}.  This can be shown to improve stability
\citep{MonAnt03,HesNagYon07}.  As shown %discussed
by \citet{GawWan11}, the intermittent controller has a similar
feature; it therefore follows that the physiologically inspired form of
intermittent controller described in this chapter has application to
networked control systems.

\item [Robotics.] It seems likely that understanding the control mechanisms behind
  human balance and motion control
  \citep{LorLakGaw09,KamGawGolLakLor13,KamGawGolLor13} and stick
  balancing \citep{GawLeeHalODw13} will have applications in
  robotics. In particular, as discussed by \citet{KamGawGolLakLor13},
  robots, like humans contain redundant possibilities within a
  multi-segmental structure. Thus the multivariable constrained
  intermittent control methods illustrated in Section
  \ref{sec:ex:stand}, and the adaptive versions illustrated in Section
  \ref{sec:aic} may well be applicable to the control of autonomous robots.
\end{description}

\end{itemize}

\newpage
\bibliography{common}

%%% Local Variables:
%%% TeX-master: "Book.tex"
%%% mode: TeX-PDF
%%% End: